%% file: main.tex
\pdfoutput=1
\documentclass[letterpaper,11pt]{article}
\usepackage[margin = 1 in]{geometry}
\usepackage[utf8]{inputenc}
\usepackage{mathtools}
\usepackage{amsthm,amsmath}
\usepackage{amssymb}
\usepackage{stmaryrd}
\usepackage{algorithm}
\usepackage[noend]{algpseudocode}
\usepackage{caption}
\usepackage{relsize}
\usepackage{lineno}
\usepackage{enumitem}
\usepackage{graphicx}
\usepackage{color, colortbl}
\definecolor{Gray}{gray}{0.9}

\usepackage{tablefootnote}
\newtheorem{theorem}{Theorem}
\newtheorem{fact}{Fact}
\newtheorem{lemma}{Lemma}
\newtheorem*{remark}{Remark}
\newtheorem{definition}{Definition}
\newtheorem{cor}{Corollary}
\newtheorem{claim}{Claim}

\DeclareMathOperator{\polylog}{polylog}
\usepackage{caption}
\usepackage[unicode,psdextra]{hyperref}

\newcommand*{\rom}[1]{\expandafter\@slowromancap\romannumeral #1@}
\RenewExpandableDocumentCommand{\textsuperscript}{m}{\raisebox{1.25ex}{\scriptsize #1}}

\usepackage[all]{nowidow}
\title{Faster Multi-Source Reachability and Approximate Distances via Shortcuts, Hopsets and Matrix Multiplication}
\author{Michael Elkin \\ Ben-Gurion University of the Negev, Israel. \\ \texttt{elkinm@bgu.ac.il }  \and Chhaya Trehan \\Reykjavik Uniersity, Iceland.\\ \texttt{chhaya.dhingra@gmail.com} }
 
\begin{document}
\maketitle
\begin{abstract}
Given an $n$-vertex $m$-edge directed graph $G = (V,E)$ and a designated source vertex $s \in V$, let $\mathcal V_s \subseteq V$ be a subset of vertices 
reachable from $s$ in $G$. Given a subset $S \subseteq V$ of $|S| = n^{\sigma}$, for some $0 \le \sigma \le 1$, designated sources,
the $S \times V$ \emph{reachability problem} is to compute the sets $\mathcal V_s$ for every $s \in S$.
Known naive algorithms for this problem in the \emph{centralized} setting either run a BFS/DFS exploration separately from every source, and as a result require $O(m \cdot n^{\sigma})$ time, 
or compute the transitive closure $TC(G)$ of the graph $G$ in $\tilde O(n^{\omega})$ time\footnote{We use the  notation $\tilde O$ to hide polylogarithmic factors.}\footnote{The expression $O(n^{\omega})$ is the running time of the state-of-the-art algorithm for computing a matrix product of two square matrices of dimensions $n \times n$.}, where $\omega \le 2.371552\ldots$ is the \emph{matrix multiplication exponent}.
Hence, the current state-of-the-art bound for the problem on graphs with $m = \Theta(n^{\mu})$ edges is $\tilde O(n^{\min \{\mu + \sigma, \omega \}})$,
which may be as large as $\tilde O(n^{\min \{ 2 + \sigma, \omega\}})$.

Based on recent constructions of shortcuts~\cite{KoganParterDiShortcuts1, KoganParterDiShortcuts2}, we
devise a centralized algorithm with running time $\tilde O(n^{1 + \frac{2}{3} \omega(\sigma)})$ for this problem on general graphs, where $\omega(\sigma)$ is the
rectangular matrix multiplication exponent (for multiplying an
$n^{\sigma} \times n$  matrix by an $n \times n$ matrix). Using current state-of-the-art estimates on $\omega(\sigma)$, our exponent is better than
$\min \{2 + \sigma, \omega \}$ for $0.3336 \le \sigma \le 0.53$.
For graphs with $m = \Theta(n^{\mu})$ edges, for $1 \le \mu \le 2$, our algorithm has running time $\tilde O(n^{\frac{1 + \mu + 2 \cdot \omega(\sigma)}{3}})$. 
For every $0.3336 \le \sigma < 1$, there exists a non-empty interval
$I_{\sigma} \subseteq [1,2]$, so that our running time is better than the state-of-the-art one whenever the input graph has $m = \Theta(n^{\mu})$ edges with $\mu \in I_{\sigma}$.

In a classical paper, Cohen~\cite{Cohen93}
 devised \emph{parallel} algorithms for $S \times V$ reachability problem on graphs 
 that admit balanced recursive separators of size $n^{\rho}$, for $\rho < 1$.
 Her algorithm assumes that an appropriate separator decomposition of the input graph is provided as input
 or can be readily computed. It requires polylogarithmic time and work $\tilde{O}(n^{\max \{\omega \rho, 2\rho + \sigma \}})$. We significantly improve, extend and generalize the results of~\cite{Cohen93}. First, the work complexity of our parallel algorithm for graphs with small recursive 
 separators is better than that of~\cite{Cohen93} in a wide range of parameters.
Second, we generalize the algorithm to graphs of treewidth at most $n^{\rho}$, 
 which may be quite dense, and provide a \emph{centralized} algorithm that works
\emph{from scratch} and requires $\tilde O(n^{\max \{\omega\rho + 1 - \rho,\omega(\sigma)\} })$ deterministic time, i.e., it  outperforms the current state-of-the-art bounds for $S\times V$ reachability on
such graphs. We also provide analogous results for $K_h$-minor free graphs (for $h = n^{\eta}$, $0< \eta < 1/2$),
 geometric $k$-nearest neighbor graphs and their extensions.
Finally, we extend these results to $(1 + \epsilon)$-approximate distance computation and provide parallel versions of these results.
 \end{abstract}
\section{Introduction}\label{sec:Intro}
\subsection{Background}\label{sec:IntroBGND}
Consider an $n$-vertex directed graph $G = (V,E)$, possibly weighted with non-negative edge weights with aspect ratio (the ratio
between maximum non-zero weight and minimum non-zero weight) at most polynomial in $n$.
 Suppose that we are also given a subset $S \subseteq V$ 
  of designated \emph{sources}, $|S| = n^{\sigma}$, $0 \le \sigma \le 1$.
In the $S \times V$ \emph{reachability} problem we need to compute, for every source $s \in S$, the set of vertices \emph{reachable} from $s$ in $G$. (A vertex $v$ is said to be \emph{reachable} from $s$ if there exists a path from $s$ to $v$ in $G$.) In the more general \emph{$(1+ \epsilon)$-approximate $S \times V$ distance computation} problem we also need to 
compute $(1 + \epsilon)$-approximate distances in $G$ for all pairs $(s, v) \in S \times V$, where $\epsilon > 0$ is a small parameter.

These are fundamental graph-algorithmic problems.
The \emph{single-source} case ($\sigma = 0$) and the \emph{all-pairs} case ($\sigma = 1$) were extensively studied~\cite{Dijkstra59, Munro71, Purdom70, Seidel92, ALON199,  ChanM2010, RodittyS11, Cohen93, APSPDirectedZwick,Fineman20, CaoFR20, CaoFR22, CaoF23, RozhonHMZ23, Fineman24}.
We focus on the intermediate regime of constant $\sigma$, $0  < \sigma < 1$, which received so far less attention. Specifically,
the $S \times V$ reachability problem, as well as its extension to \emph{exact} $S \times V$ distance computation, was studied in a classical paper of Cohen~\cite{Cohen93} in the context of graphs that admit \emph{sublinear recursive separators}
(see Section~\ref{sec:separatorDef} for relevant definitions) in the \emph{parallel} model (\textsf{PRAM}) of computation.
We study the $S \times V$ reachability problem in both centralized and parallel models, and also both in general graphs and in 
restricted graph families. Most of our results extend also to the $(1 + \epsilon)$-approximate $S \times V$ distance computation problem.

There are two currently known trivial solutions to the $S \times V$ reachability problem.
The first one runs a BFS (or DFS) from every $s \in S$, and as a result requires $O(m \cdot n^{\sigma})$ time in the centralized setting.
We call this algorithm \emph{BFS-based}. For dense graphs this running time may be as large as $O(n^{2 + \sigma})$.
The second solution (which we call \emph{TC-based}) uses fast square matrix multiplication to compute the transitive closure 
$TC(G)$ in $\tilde O(n^{\omega})$ centralized time.
The upper bound $\omega \le 2.371552\ldots$ is due to~\cite{VassilevskaBestRectangular} (see also~\cite{gall2023faster,LeGallBestRectangular,LeGallBestSquare, COPPERSMITH}).

Zwick~\cite{APSPDirectedZwick} showed that the TC-based approach generalizes to the $(1 + \epsilon)$-approximate distance
computation problem with essentially the same running time $\tilde  O(n^{\omega})$.
Zwick's algorithm also parallelizes seamlessly, providing polylogarithmic time (depth) and 
$\tilde O(n^{\omega})$ processors (see e.g.,~\cite{ElkinN22}).
This, however, is not the case with BFS-based approach: the best currently known single-source parallel reachability algorithm 
due to Cao et al.~\cite{CaoFinemanRuseell20} requires $n^{1/2 + o(1)}$ time and $\tilde O(m)$ work.
If one restricts oneself to polylogarithmic parallel time, then in general graphs only the TC-based 
solution is applicable.

On the other hand, Cohen~\cite{Cohen93} devised a parallel algorithm with polylogarithmic time for 
graphs that admit $k^{\rho}$-recursive separators, for some constant $\rho < 1$.
(These are graphs $G = (V,E)$, such that every $k$-vertex subgraph of $G$ admits a balanced separator 
of size $O(k^{\rho})$, see Section~\ref{sec:separatorDef} for the definition of a balanced separator.)
Specifically, the algorithm of~\cite{Cohen93} has work complexity $\tilde O(n^{\omega \cdot \rho}  + n^{\max\{2\rho, 1 \} + \sigma  })$.
Her algorithm also assumes that an appropriate separator decomposition of the input graph $G$ is provided to the algorithm as input, or can be readily computed. (This is the case, for example, in geometric graphs, which is the prime example in~\cite{Cohen93}. See Sections~\ref{sec:GeometricGraphs} and~\ref{sec:reachDistGeometric} for more details.)
Cohen~\cite{Cohen93} also provided a variant of her algorithm for \emph{exact} $S \times V$ distance computation on general-weighted graphs with $k^{\rho}$-recursive separators. This variant has work complexity 
$\tilde O(n^{\max \{3 \cdot \rho, 1\} } + n^{2\rho + \sigma})$.

Note that in the \emph{centralized} setting, the algorithm of~\cite{Cohen93} provides no non-trivial bounds.
Indeed, graphs with  $k^\rho$-recursive separators for $\rho < 1$ have $O(n)$ edges (\cite{LiptonRoseTarjan79}, Theorem 12).
Thus, the BFS-based centralized $S \times V$ reachability algorithm provides time $O(n^{1 + \sigma})$, which is better than the work complexity of~\cite{Cohen93}.

\subsection{Our Results}\label{sec:results}
We provide two types of results. Results of the first type apply to the $S \times V$ reachability problem for 
\emph{general} graphs in the centralized setting. They extend neither to distance computation nor to the parallel setting, and they are achieved via randomized algorithms. 
Results of the second type apply to both the reachability and $(1 + \epsilon)$-approximate distance computation problems,
both in the centralized and parallel settings, and they are achieved via deterministic algorithms. They are, however, restricted to certain graph families.
These graph families include sparse graphs, such as graphs with sublinear separators and graphs with genus $g = o(n)$.
(Both these families contain $m = O(n)$ edges.) But we also allow \emph{much denser} graphs, such as graphs with treewidth $n^{\rho}$ for some $0 < \rho < 1$, graphs that exclude a fixed $h$-sized minor for $h = n^{\eta}$, $0 < \eta < 1$, and
geometric graphs. For these denser graph families our results (unlike the results of~\cite{Cohen93}) are meaningful 
in the centralized setting as well.
\subsubsection{General Graphs}\label{sec:resultsGeneral}
In the context of general graphs, based on recent breakthroughs of Kogan and Parter~\cite{KoganParterDiShortcuts1, KoganParterDiShortcuts2} concerning directed shortcuts, we devise an algorithm that outperforms both the BFS-based and the 
TC-based $S \times V$ reachability algorithms whenever the parameter $\sigma = \log_n|S|$ is in the range 
$\tilde \sigma < \sigma \le 0.53$, where $1/3 < \tilde \sigma < 0.3336$ is a universal constant.

The running time of our algorithm is $\tilde O(n^{1 +  \frac{2}{3} \omega(\sigma)})$, where $\omega(\sigma)$
is \emph{rectangular matrix multiplication exponent}, i.e., $O(n^{\omega(\sigma)})$ is the running time\footnote{For $\sigma \le \alpha$, where $\alpha \approx 0.321334$ is the {\em dual matrix multiplication exponent}, $\omega(\sigma) = 2$, but the running time is $n^{2+o(1)}$. For larger $\sigma$, $\omega(\sigma) > 2$, and the running time is at most $n^{\omega(\sigma)}$.} of the state-of-the-art algorithm for
 multiplying a rectangular matrix with dimensions $n^{\sigma} \times n$ by a square matrix of dimensions $n \times n$.
See Table~\ref{tab:tableLeGall} for the values of this exponent (due to~\cite{VassilevskaBestRectangular}),
and Table~\ref{tab:tableComparisonDense} for the comparison of running time of our algorithm 
with that of previous algorithms for $S \times V$ reachability problem in the range 
$\tilde \sigma \le \sigma \le 0.53$ (in which we improve the previous state-of-the-art). 
The largest gap between our exponent and the previous state-of-the-art is for $0.37 \le \sigma \le 0.38$. 
Specifically, our exponent for $\sigma = 0.38$ is $2.33751$, while the state-of-the-art is $\omega \approx 2.371552$.

Moreover, we show that our algorithm improves the state-of-the-art solutions for $S \times V$ reachability problem 
in a much wider range of $\sigma$ on sparser graphs. 
Specifically, we show that for every $\tilde \sigma \le \sigma < 1$, there exists a non-empty interval $I_{\sigma} \subset [1,2]$ of values $\mu$,
so that for all input graphs with $m = \Theta(n^{\mu})$ edges, $\mu \in I_{\sigma}$, our algorithm outperforms both the BFS-based and the TC-based
solutions. The running time of our algorithm on such graphs is $\tilde O(n^{ \frac{1 + \mu + 2\cdot \omega(\sigma)}{3}})$.

We provide some sample values of these intervals in Table~\ref{tab:table1}. 
See also Table~\ref{tab:tableComparisonMu} for the comparison between our new bounds for $S \times V$ reachability on graphs with $m = \Theta(n^{\mu})$ edges, $\mu \le 2$,
 and the state-of-the-art ones.
For example, for $\sigma = 0.5$ (when $|S| = \sqrt n$), the interval is $I_{\sigma} = (1.793, 2]$, i.e., our algorithm improves
the state-of-the-art solutions as long as $\mu > 1.793$.
Specifically, when $\mu = 1.9$, our algorithm requires $\tilde O(n^{2.3287\ldots})$ time for computing reachability from $\sqrt n$ sources,
 while the state-of-the art solution is the TC-based one, and it requires $\tilde O(n^{2.371552\ldots})$ time.
 Another example is when $\sigma = 0.6$.
 The interval $I_{\sigma} = I_{0.6}$ is $(1.693, 1.93)$, and for $\mu = 1.75$ our solution requires $\tilde O(n^{2.312\ldots})$ time, while the state-of-the-art BFS-based solution requires $O(n^{2.35})$ time.

 Since we employ randomized constructions of shortcuts \cite{KoganParterDiShortcuts1,KoganParterDiShortcuts2},
 our algorithms for multi-source reachability in general graphs are randomized as well.

 \subsubsection{Restricted Graph Families}\label{sec:sec:resultsRestricted}
We significantly improve and generalize Cohen's results~\cite{Cohen93}.
Specifically, for graphs that admit $k^{\rho}$-recursive separators, our algorithm computes 
$S \times V$ reachability and $(1 + \epsilon)$-approximate $S \times V$-distances in parallel 
polylogarithmic time and work $\tilde O(n^{\max\{\omega \rho, \omega(\sigma)  \}})$.
This exponent $F(\sigma, \rho) = \max \{\omega \rho, \omega(\sigma) \}$ improves the exponent 
$C(\sigma, \rho) = \max\{\omega \rho , 2\rho + \sigma\}$ of the algorithm of~\cite{Cohen93} in a
wide range of parameters. Both our algorithm and that of \cite{Cohen93} are deterministic.

Specifically, for every $\sigma > \frac{2(\omega - 2)}{\omega} \approx 0.31334 $,
there is a non-empty interval $J_{\sigma} = (\frac{\omega(\sigma) - \sigma}{2}, \frac{\sigma}{\omega - 2})$
such that for every $\rho \in J_{\sigma}$, we have $F(\sigma, \rho) < C(\sigma, \rho)$.
The gap between these two exponents may be quite significant.
In particular, for $\sigma = 0.5$ and $\rho = 0.85$, our exponent is $F(\sigma, \rho) = \omega(0.5) = 2.0429 $, 
while the exponent of the algorithm of~\cite{Cohen93} is $C(\sigma, \rho) = 2.2$.
See Tables~\ref{tab:rhoInterval2} and~\ref{tab:comparisonRho2} for a  detailed comparison of these exponents.
Note that both our algorithm and the algorithm of~\cite{Cohen93} assume that an appropriate separator decomposition 
is provided as input. 

Observe also that as graphs with $k^\rho$-recursive separators for $\rho < 1$ have $O(n)$ edges \cite{LiptonRoseTarjan79},
in the \emph{centralized} setting none of these algorithms can compete with a naive BFS-based solution for the problem,
as the latter requires just $O(n^{1 + \sigma})$ time.
We, however, extend our algorithm to graphs with \emph{treewidth} at most $n^{\rho}$, for a parameter $0 < \rho < 1$.
In this case the running time of our algorithm in the centralized setting becomes $\tilde O (n^{F'(\sigma, \rho)})$,
where $F'(\sigma, \rho) = \max \{\omega \rho + 1 - \rho , \omega(\sigma) \}$.
(In the parallel setting, our algorithm has polylogarithmic in $n$ running time and work $\tilde O (n^{F'(\sigma, \rho)})$.)

In the centralized setting, using an approximation algorithm of Chuzhoy and Saranurak~\cite{ChuzhoyS21},
one can compute the respective tree/separators decomposition within the stated time bounds. Thus, our algorithm 
provides this running time (for $S \times V$ reachability and $(1 + \epsilon)$-approximate distance computation problems)
\emph{from scratch}.

We compare the running time of our (centralized deterministic) algorithm with the minimum running time of the BFS-based and the TC-based
algorithm for this problem. In the densest regime graphs with treewidth $O(n^{\rho})$ may have $O(n^{1 + \rho})$ edges, and thus the running time of the current state-of-the-art (naive) solution is $\tilde O(n^{\min \{\omega, 1 + \rho + \sigma \}})$.
Denote by $N'(\sigma, \rho) = \min \{\omega, 1 + \rho + \sigma \}$ the exponent of this running time.
Our exponent $F'(\sigma, \rho)$ is better than $N'(\sigma, \rho)$ in a wide range of parameters.
Specifically, for any $\frac{\omega - 2}{\omega -1} < \sigma < 1$, there is a non-empty interval 
$K'_{\sigma} = (\omega(\sigma) - (1 + \sigma)   , \min \{1, \frac{\sigma}{\omega - 2} \} )$, such that for 
every $\rho \in K'_{\sigma}$, we have $F'(\sigma, \rho) < N'(\sigma, \rho)$.
The improvements may be very significant. 
For example, for $\sigma = 0.6$, $\rho = 0.8$, our exponent is $F'(0.6, 0.8) = 2.0972$,
while the previous state-of-the-art is $N'(0.6, 0.8) = \omega = 2.371552$.
See Theorem~\ref{thm:treewidthDistance} and Tables~\ref{tab:tableRhoInterval} and~\ref{tab:FlatBoundCompare} for a detailed comparison between these two exponents. Table ~\ref{tab:tableRhoInterval} also provides sub-intervals $K''_\sigma \subseteq K'_\sigma$ in which our exponent $F'(\sigma,\rho)$ is equal to the optimal exponent $\omega(\sigma)$. These sub-intervals are not empty for all $\frac{\omega-2}{\omega-1} < \sigma < 1$, and they become wider as $\sigma$ grows.

Moreover, even if the input graph contains $n^{\mu}$ edges for some $\mu < 1 + \rho$, still there are wide ranges of parameters in which our exponent is better than the previous state-of-the-art one (obtained as a minimum of respective exponents of BFS-based and TC-based exponents).
Specifically, we show (Theorem~\ref{thm:flatBoundDistance}) that for any $\omega - 1 < \mu \le 1 + \rho$,
there is a non-empty range $(\sigma'_{thr}(\rho, \mu), 1)$, where $\sigma'_{thr}(\rho, \mu)$ is given by Equation~\eqref{eq:sigmaThresh},
such that for every $\sigma$ in this range, our exponent improves the state-of-the-art one.
For example, for $\sigma = 0.7$, $\rho = 0.7$ and $\mu = 1.6$, our exponent is $\omega(0.7) = 2.1530$, while the state-of-the-art exponent is $2.3$. See Tables~\ref{tab:FlatBoundSparseComp1} and~\ref{tab:FlatBoundSparseComp2} for more details. 

These improvements carry over to the parallel setting as well, but there one needs to provide our algorithm with a separator decomposition as a part of the input. Recall that this is also the case for the algorithm of \cite{Cohen93}. Both parallel algorithms (ours and of \cite{Cohen93} are deterministic).

Another variant of our algorithm provides improved bounds for graphs that exclude minor $K_h$, for $h = n^{\eta}$,
where $0 < \eta < 1/2$ is a parameter.
(Such graphs have treewidth at most $\tilde O(n^{1/2 + \eta})$, and some non-trivial bounds for them can be derived from the 
variant of our algorithm that applies to graphs with treewidth $\tilde O(n^{\rho})$, $\rho \le 1/2 + \eta$.
However better bounds are obtained by considering minor-free graphs directly.)
We argue that for any $\omega - 2 < \eta < 1/2$ and $\omega - 1 < \mu \le 1 + \eta$
(recall that $\mu$ is the exponent of the number of edges, i.e., $m = n^{\mu}$),
there is a threshold value $\sigma_{thr} < 1$ (see Equation~\eqref{eq: muOmegaEta1}) such that for any $\sigma > \sigma_{thr}$,
our exponent is better than the previous state-of-the-art one.
Here the improvements are typically smaller than for graphs with bounded treewidth.
Among the  most noticeable improvements is $\mu = 1.45$, $\eta = 0.45$ and $\sigma = 0.9$.
For these parameters our exponent is $\omega(0.9) = 2.2942$, while the state-of-the-art exponent is $2.35$. 
See Theorem~\ref{thm:excludedMinorDistances} and Tables~\ref{tab:sigmaThr1}-\ref{tab:sigmaThr4} for more details. 
We remark that our centralized (deterministic) algorithm for minor-free
graphs works \emph{from scratch}, i.e., it can compute the required tree decomposition within the stated 
resource bounds. This is also the case for geometric graphs, discussed below.
(Moreover, for geometric graphs this is also the case in the \emph{parallel} setting as well, but the resulting algorithm becomes randomized.)

We note that Wulff-Nilsen~\cite{WulffNisen11}, improving over a result by Yuster~\cite{YUSTER2010},
devised a single source shortest paths algorithm in $K_h$-minor-free graphs with weights 
bounded by $\operatorname{poly}(n)$ \emph{in absolute value}.
The running time of his algorithm is $\tilde O(\operatorname{poly} (h) \cdot n^{1/3})$,
where $\operatorname{poly}(h)$ is a high-degree unspecified polynomial. This result is 
incomparable to ours, as for the values of $h$ for which our algorithm outperforms 
the BFS-based and the TC-based solutions, the algorithm of~\cite{WulffNisen11}  provides no  meaningful bounds.
On the other hand, we restrict weights to be non-negative, while this is not the case in~\cite{WulffNisen11,YUSTER2010}.

We also separately consider \emph{geometric} graphs.
We focus on \emph{$k$-nearest neighbors} ($k$-NN) graphs, and more generally, intersection (and $r$-overlap for $r = \Theta(1)$) graphs of $k$-ply neighborhood systems.
Given a point set $\mathcal P = \{p_1, p_2, \ldots, p_n\}$ of $n$ points in $\Re^d$,
for some constant dimension $d \ge 2$, and a parameter $k$, the $k$-NN graph of $\mathcal P$
has vertex set $\mathcal P$, and there is an edge $(p_i, p_j)$ iff either $p_i$ is among $k$ 
closest points to $p_j$ or vice-versa. See Section~\ref{sec:GeometricGraphs}
for the definition of $k$-ply neighborhood systems, and further discussion about this important graph family.

The crucial property of these graphs that we (and~\cite{Cohen93}) exploit is that they admit
recursive separators of size $O(k^{1/d} n^{1 - 1/d})$, and moreover, these separators are efficiently computable within $n^{2 + o(1)}$ (or less) centralized time and in parallel polylogarithmic time using $n^{2 + o(1)}$ (or less) work.
(These results are due to~\cite{MillerThurston90, teng1991points, MillerTTV97, EppsteinMillerTeng95}.
We refer the interested reader to~\cite{teng1991points} for an excellent exposition of this fascinating subject.) We write $k = n^q$, for a constant parameter $0 < q < 1$. 
%
Considering $d = 2$ (i.e., $2$-dimensional point sets), we show that for any 
$\omega - 1 < \mu \le 1 + q$
(recall that $m = n^{\mu}$ is the number of edges in the input graph), there exists a
threshold $\sigma'_{thr} = \sigma'_{thr}(\mu, q) < 1$ (given by~\eqref{eq:sigmaThreshGeometric1}) such that for any 
$\sigma \in (\sigma'_{thr}, 1) $, the exponent $F(\sigma, q)$ of the running time of our 
centralized algorithm is better than the exponent of the previous state-of-the-art.
For example, for $q = 0.7$, $\mu = 1.6$ and $\sigma = 0.7$, our exponent is $\omega(0.7) = 2.1530$,
while the previous state-of-the-art is $2.3$. See Theorem~\ref{thm:geometricTimeDim2Cetlzd} and Tables~\ref{tab:GeometricDim2Comp1} and~\ref{tab:GeometricDim2Comp2} for further details.

We note that the improvements in the parallel setting are even more significant (see Theorem~\ref{thm:geometricTimeDim2}).
Also for geometric graphs, both centralized and parallel variants our our algorithm work \emph{from scratch}, i.e., the respective separator decompositions can be computed (via aforementioned algorithms by~\cite{MillerThurston90, MillerTTV97, teng1991points, EppsteinMillerTeng95} within stated resource bounds). We also show that for any constant dimension $d \ge 2$, there are non-empty ranges of parameters $q, \mu, \sigma$, such that our algorithms for $S \times V$ reachability and $(1 + \epsilon)$-approximate distance computation on $k$-NN $d$-dimensional graphs with $k = n^q$, $m = n^{\mu}$ edges (and their aforementioned extensions)
outperform the current state-of-the-art bounds. (See Theorem~\ref{thm:geometricGraphsTimeLarged}.)

In all our results the edges can be oriented and weighted arbitrarily 
(as long as weights are non-negative and the aspect ratio is at most $\operatorname{poly}(n)$),
while the structural assumptions (such as having bounded treewidth or excluded minor or being $k$-NN)
apply to the \emph{unweighted undirected} skeleton of the input graph.

\subsection{Our Techniques}\label{sec:ourTechnqs}
 In this section we provide a high-level overview of our algorithms.
 
\subsubsection{General Graphs}\label{sec:ourTechnqsGnrl}
 Given a digraph $G = (V,E)$, and a positive integer parameter $D$,
  we say that a graph $G' = (V,H)$ is a \emph{$D$-shortcut} of $G$ if
  for any ordered pair $u,v \in V$ of vertices, $v$ is reachable from $u$ in $G$
  iff it is reachable from $u$ in $G \cup G' = (V, E \cup H)$ using at most $D$ hops.
  Kogan and Parter~\cite{KoganParterDiShortcuts1, KoganParterDiShortcuts2} showed that
  for any parameter $1 \le D \le \sqrt n$, for any digraph $G$, there exists a $D$-shortcut $G'$ with
  $\tilde O(\frac{n^2}{D^3} + n)$ edges, and moreover, this shortcut can be constructed in $\tilde O(m \cdot n/D^2 + n^{3/2})$ randomized time~\cite{KoganParterDiShortcuts2}.
  In a subsequent work, Kogan and Parter~\cite{KoganParterDiShortcuts3} further improved the time complexity of their shortcut construction to  $\tilde O(m \cdot n/D^2)$, for any $1 \le D \le \sqrt n$.

   Our algorithm starts with invoking the algorithm of~\cite{KoganParterDiShortcuts2} for an appropriate parameter $D$.\footnote{We note here that all our algorithms use rectangular matrix multiplication as a sub-routine and hence require $\Omega(n^{2})$ time. Therefore, the additional saving of  $n^{3/2}$ in the time complexity of~\cite{KoganParterDiShortcuts3} in comparison to~\cite{KoganParterDiShortcuts2} does not improve any of our algorithms.}
 We call this the \emph{diameter-reduction step} of our algorithm.
  We then build two matrices.
  The first one is the adjacency matrix $A'$ of the graph $G \cup G'$, with all the diagonal entries equal to $1$.
  This is a square $ n \times n$ matrix, and the diagonal entries correspond to self-loops.
  The second matrix $B$ is a rectangular $|S| \times n$ matrix.
  It is just the adjacency matrix of of the graph $G \cup G'$ restricted to $|S|$ rows corresponding to designated sources of $S$.
  
  Next, our algorithm computes the Boolean matrix product $B \cdot A'^{D}$.
  Specifically, the algorithm computes the product $B \cdot A'$, and then multiplies it by $A'$, etc., and does so $D$ times.
  Hence, this computation reduces to $D$ Boolean matrix products of a rectangular matrix of dimensions $|S| \times n$ by a square 
  matrix with dimensions $n \times n$. Each of these $D$ Boolean matrix products can be computed using standard matrix multiplication (see Section~\ref{sec:bmm}).
    We use fast rectangular matrix multiplication algorithm by~\cite{VassilevskaBestRectangular} to compute these products.
  As $|S| = n^{\sigma}$, this computation, henceforth referred to as the \emph{reachability computation step}, requires $O(D \cdot n^{\omega(\sigma)})$ time. 
  Since the diameter of the graph $G \cup G'$ is at most $D$ (by definition of a $D$-shortcut), the 
  $D$-hop reachabilities that we are computing in this step with respect to $G \cup G'$ are equivalent to 
  the general unrestricted reachabilities in the original graph $G$.
  
  Observe that the running time of the reachability computation step grows with $D$, and the running time of
  the diameter-reduction step decreases with $D$. It is now left to balance these two running times by carefully choosing $D$.
  This completes the overview of our reachability algorithm for general graphs, that improves 
  the state-of-the-art running time for $S \times V$ reachability, $|S| = n^{\sigma}$, for $\tilde \sigma \le \sigma \le 0.53$, $\tilde \sigma < 0.3336$.
  Moreover, as was mentioned above, for any $\sigma$, $\tilde \sigma \le \sigma < 1$, there is a non-empty interval $I_{\sigma} \subseteq [1,2]$ of 
  values $\mu$, such that this algorithm improves upon state-of-the-art solutions for this problem on graphs with $m = \Theta(n^{\mu})$ edges.
  Since there is no known efficient (polylogarithmic-time) construction of shortcuts that are on par with that of \cite{KoganParterDiShortcuts1,KoganParterDiShortcuts2}, currently these results apply only in the centralized setting. 

  \subsubsection{Restricted Families of Graphs}\label{sec:OurTechnqsRestricted}
The starting point of our algorithm for $S \times V$ reachability in graphs that admit $k^{\rho}$-recursive separators, for a parameter $\rho < 1$, is that the algorithm of~\cite{Cohen93} implicitly builds an $O(\log n)$-shortcut for its input graph.
This shortcut has size $\tilde O(n^{\max \{2\rho, 1 \}})$, and it is computed in parallel deterministic polylogarithmic time and work $\tilde O(n^{\omega \rho})$.
Cohen's algorithm then augments the input graph $G$ with the constructed shortcut $G'$ and obtains an augmented graph $\hat G$. It then runs BFS (or DFS) to depth $O(\log n)$ on $\hat G$.
Since $G'$ is an $O(\log n)$-shortcut, the computed reachabilities (restricted to depth $O(\log n)$)
are the same as the (unrestricted) reachabilities in $G$.
The number of edges in $\hat G$ is dominated by the number of edges in the shortcut $G'$, i.e., 
it is $\tilde O(n^{\max \{1,2\rho \}})$.
Thus, $|S|$ explorations from each source $s \in S$ in $\hat G$ require polylogarithmic deterministic time (because 
they are $O(\log n)$-restricted) and $\tilde O(n^{2\rho + \sigma})$ work.
(We focus on the case $\rho \ge 1/2$.) Thus, the overall work complexity of the algorithm of~\cite{Cohen93} is 
$\tilde O (n^{\omega \rho} + n^{2\rho + \sigma} )$.

Our first algorithm (for computing $S \times V$ reachability on graphs that admit $k^{\rho}$-recursive separators) builds the same shortcut as the algorithm of~\cite{Cohen93}, but then multiplies the rectangular matrix $B$ created from $|S|$ rows that correspond to the sources $s \in S$ of the adjacency matrix $A$ of the augmented graph $\hat G$, by the matrix $A$. It then multiplies the product by $A$ again, and iterates for $O(\log n)$ times. 
The resulting parallel algorithm improves the algorithm of~\cite{Cohen93} in a wide range of parameters.

We then extend this approach in a number of directions.
First, we adapt the algorithm of~\cite{Cohen93} that builds $O(\log n)$-shortcuts so that it builds
$(1 + \epsilon, O(\log n))$-hopsets (see Definition~\ref{def:hopset}) within roughly the same time.
This extension (and analogous extensions in other settings) enables us to solve
$S \times V$ $(1 + \epsilon)$-approximate distance computation within essentially (up to polylogarithmic in $n$ and polynomial in $\epsilon^{-1}$ factors)
the same time as we solve $S \times V$ reachability problem.

The algorithm of~\cite{Cohen93} for building the shortcut works on a separator decomposition tree $T_G$ of the input graph $G$.
It starts with computing all reachabilities in the constant-sized leaf nodes of $T_G$,
and then, for any internal node $t$ whose children $t_1$, $t_2$ were already taken care of,
it builds a number of carefully crafted graphs over certain small subsets of vertices that appear in $t$, and computes reachabilities between all pairs of vertices in these graphs. Cohen~\cite{Cohen93} also devised a variant of her algorithm that computes an \emph{exact}
$O(\log n)$-hopset of $G$, albeit its work complexity is $\tilde O(n^{\max \{3\rho, 1\}})$
(rather than $\tilde O(n^{\max \{\omega \rho, 1\}}) $ for computing shortcuts).
This variant invokes all-pairs-shortest-paths computation in each of the aforementioned subgraphs, rather than reachability computations.
Our algorithm instead invokes the algorithm of~\cite{APSPDirectedZwick} for $(1 + \epsilon)$-approximate all-pairs
shortest paths computation in each of these subgraphs, leading to overall complexity of $\tilde O(n^{\omega \rho})$ for building the hopset.

The main obstacle is, however, in analyzing the stretch of the resulting hopset.
The stretch accumulates as we go up from leaves of $T_G$ towards its root.
However, we argue that for a node $t$ of level $\ell = \ell(t)$ in $T_G$,
the stretch of a hopset edge $(u,v)$ between a pair of vertices $u, v \in t$
is at most $(1 + \epsilon)^{O(\ell)}$.
As $\ell = O(\log n)$, we rescale $\epsilon' = \frac{\epsilon}{O(\log n)}$, and obtain an overall stretch of $(1 + \epsilon')$, at the expense of increasing the size by a factor of $O(\log n)$. 

We also extend the algorithm of~\cite{Cohen93} to work on graphs with treewidth $O(n^{\rho})$,
and to graphs that exclude $K_h$ as a minor for $h = n^{\eta}$, and for $k$-NN geometric graphs (and their extensions) for $k = n^q$, for parameters $0 < \rho < 1$, $0 < \eta < 1$, $0 < q < 1$.
Note that these graphs do not necessarily admit sublinear recursive separators.
Indeed, a subgraph with $n^{\rho}$ vertices of a graph with treewidth $O(n^{\rho})$
may not have a sublinear separator.
(Consider, for example, a graph with $n^{1 - \rho}$ layers, with $n^{\rho}$
vertices each, with a complete bipartite graph between each pair of 
consecutive layers. It has treewidth $O(n^{\rho})$, but no sublinear recursive separator.
In particular, it has $O(n^{1 + \rho})$ edges while graphs with sublinear recursive separators have $O(n)$ edges.)
The situation is similar with $K_h$-minor free graphs: their subgraphs of size $h = n^{\eta}$ may not 
admit sublinear separators.
This is also the case for $k$-NN graphs and subgraphs with $k = n^q$ vertices. Thus, the result of~\cite{Cohen93} is not applicable to them.

We generalize the construction of shortcuts from~\cite{Cohen93}
(and our extension to $(1 + \epsilon, O(\log n))$-hopsets) to these families of graphs.
In particular, we show that such shortcuts and hopsets can be constructed in $\tilde O (n^{\omega \rho + 1 - \rho})$ time for graphs with treewidth $O(n^{\rho})$ (Corollary~\ref{cor:flatTime}), in $\tilde O ((h \sqrt n)^{\omega})$
time for graphs that exclude $K_h$-minor (Corollary~\ref{cor:excludedMinor}), and in $\tilde O (k^{\omega/d} \cdot n^{\omega - \omega/d})$ time for geometric $d$-dimensional $k$-NN graphs 
and their extensions, (Corollary~\ref{cor:knn}).
These results lead to the first non-trivial bounds for \emph{centralized}
$S \times V$ reachability and approximate distance computation on these families of graphs.
Indeed, the result of~\cite{Cohen93} applies only to graphs with $O(n)$ edges,
for which these problems are easy in the centralized setting. 
By extending her approach to denser graph families we obtained non-trivial bounds for the centralized setting as well.

\textbf{Related Work}
  Fast rectangular matrix multiplication in conjunction with hopsets was recently employed in a similar way in the context of
  distance computation in \emph{undirected} graphs by Elkin and Neiman~\cite{ElkinN22}.
  Directed hopsets, based on Kogan-Parter constructions of shortcuts~\cite{KoganParterDiShortcuts1, KoganParterDiShortcuts2}, 
  were devised in~\cite{BernsteinWein}.
  Shortest paths in disk graphs were extensively studied; see~\cite{ChanS19a,HaitaoJ20,brewer2024, debergC2025} and references therein.
These algorithms, though very efficient, apply only to undirected graphs, while our graphs
are directed (and arcs can be oriented arbitrarily).
Similarly these algorithms apply either on unit-weight graphs, or on graphs with a naturally defined geometric weight function.
Our algorithm applies to geometric graphs with a general weight function.
Finally, these algorithms apply for $d = 2$, while our algorithm provides 
non-trivial bounds for geometric graphs in any constant dimension.

Another thread of related work is the work on the {\em fully polynomial FPT\footnote{FPT stands for {\em fixed parameter tractability}.} framework} \cite{Fomin2015}. 
Our algorithm can be viewed as a fully polynomial FPT algorithm, parameterized by treewidth, for multi-source reachability and approximate distance computation.
Fomin et al.\ \cite{Fomin2015} initiated a systematic study of polynomially solvable problems parameterized by treewidth $t$, and devised algorithms with running time of the form $O(t^d p(n))$, for some constant $d$ and polynomial $p(n)$, for a number of fundamental graph and matrix problems, including maximum matching and maximum flow. Abboud et al.\ \cite{abboud2015} devised FPT algorthms for computing diameter and eccentricities parameterized by treewidth; the dependencies of their running times on treewidth are, however, at least exponential in $t$. Distance oracles for graphs with constant treewidth were devised by Chaudhuri and Zaroliagis \cite{ChaudZar00,ChaudZar98}. 

\subsection{Outline}\label{sec:Outline}
In Section~\ref{sec:prelims} we introduce definitions and basic concepts.
We describe our reachability algorithm for general graphs in Section~\ref{Sec:BTwoStepTheAlgo}.
Section~\ref{Sec:BoundedSep} is devoted to our algorithm for multi-source reachability for graphs that admit $k^\rho$-recursive separators and graphs with treewidth $O(n^\rho)$. Section~\ref{sec:ApproxDistance}
extends this algorithm to the problem of computing approximate distances in these graphs.
In Section~\ref{sec:Shorcuts-GraphFamilies} we analyze the running time of our shortcut
and hopset constructions in specific graph families. In Section~\ref{sec:reachShortFamilies}
we employ the analysis of our shortcut/hopset constructions from Section~\ref{sec:Shorcuts-GraphFamilies} to analyze the running time of our reachability and distance estimation algorithms in these graphs.

\section{Preliminaries}\label{sec:prelims}
We abbreviate the phrases \textsf{with high probability}, \textsf{with respect to} and \textsf{without loss of generality} by \textsf{w.h.p.}, \textsf{w.r.t.}  and \textsf{w.l.o.g.}, respectively.
\subsection{Graphs} 
For an $n$-vertex, $m$-edge directed graph $G = (V,E)$, let $TC(G)$ denote its transitive closure.


\noindent For a pair $(u, v) \in TC(G)$, the distance from $u$ to $v$ in $G$
denoted $d_G(u,v)$ is the length of a shortest path from $u$ to $v$.
For $(u, v) \notin T C(G)$, we have $d_G(u, v) = \infty$. 
The \emph{diameter} of $G$, denoted $Diam(G)$, is defined as 
$\max_{(u,v) \in TC(G)} d_G(u,v)$. 

\noindent An edge $(u,v) \in TC(G)$ which is not in $E$ is called a \emph{shortcut} edge.

\begin{definition} \label{def:kBoundedDistance}
Given an integer parameter $k$ and a pair $(u,v) \in TC(G)$, the \emph{$k$-bounded distance} from $u$ to $v$ in $G$ denoted  $d^{(k)}_G(u,v)$ is the length of a shortest $u-v$-path in $G$ that 
contains no more than $k$ edges. Note that that $d_G(u,v)$ is the unrestricted version of $d^{(k)}_G(u,v)$.
\end{definition}

As mentioned in Section~\ref{sec:Intro}, a naive solution to the $S \times V$ reachability involves performing a BFS or DFS exploration 
from each of the source vertices in $S$.
\begin{definition}\label{def:naiveReach}
For an $n$-vertex $m$-edge digraph $G = (V,E)$, we refer to the naive method for $S\times V$ reachability (which involve running a BFS or DFS separately
 from each of the sources in $S$) as $S \times V$-\emph{naiveReach} method.
We denote by $\mathcal{T}_{Naive}$ the time complexity of  $S \times V$-naiveReach, and it is given by $\mathcal{T}_{Naive} = O(m \cdot |S|)$.
\end{definition}

Another way to solve $S \times V$ reachability is to compute the transitive closure of the input digraph using square matrix multiplication.
\begin{definition}\label{def:squareReach}
For an $n$-vertex $m$-edge digraph $G = (V,E)$, we refer to the technique of computing the transitive closure of $G$ using 
square matrix multiplication as $S \times V$-\emph{squareReach}. Observe that computing the transitive closure of $G$ 
solves $S \times V$ reachability for up to $|S| = n$ sources.
We denote by $\mathcal{T}_{Square}$ the time complexity of  $S \times V$-squareReach, and it is given by $\mathcal{T}_{Square} = \tilde O(n^{\omega})$,
where $\omega$ is the exponent of $n$ in the number of operations required to compute the product of two $n \times n$ square matrices.
\end{definition}

\begin{definition}\label{def:shortcut}
For a digraph $G$, a set of shortcut edges $H \subseteq TC(G)$ is called a $D$-shortcut if $Diam(G \cup H) \le D$.
\end{definition}
The following theorem summarizes the fast shortcut computation algorithm by Kogan and Parter~\cite{KoganParterDiShortcuts2}.
\begin{theorem}\label{thm:koganParter1}
 (Theorem $1.4$,~\cite{KoganParterDiShortcuts2}) There exists a randomized algorithm that for every $n$-vertex $m$-edge digraph $G$ and $D = O(\sqrt n)$, computes, w.h.p., in time $\tilde O(m \cdot n/D^2 + n^{3/2})$
a $D$-shortcut $H \subseteq T C(G)$ with $|E(H)| = \tilde O(n^2/D^3 + n)$ edges.
\end{theorem}
The \emph{hopsets} are equivalent of shortcuts for weighted graphs. 
Formally they are defined as follows:
\begin{definition}\label{def:hopset}
Given a weighted graph $G = (V, E, w)$ with real edge weights and two parameters $0 \le \epsilon < 1$ and $\beta > 0$, a $(1+\epsilon, \beta)$-hopset of $G$ is another graph $G' = (V,H, w')$
on the same set of vertices such that the following holds in $\tilde G = G \cup G' = G(V, E \cup H)$ for every pair $(u,v) \in TC(G)$
\[
       d_G(u,v) \le  d^{(\beta)}_{\tilde G}(u,v) \le (1 + \epsilon) \cdot d_G(u,v),
\]
where $d^{(\beta)}_{\tilde G}(u,v)$ is the $\beta$-bounded distance (see Definition~\ref{def:kBoundedDistance})
from $u$ to $v$ in 
$\tilde G$
\end{definition}
For a $(1 + \epsilon, \beta)$-hopset $H$ of a graph $G$, $(1+\epsilon)$ is called the \emph{stretch} and $\beta$ is called the \emph{hopbound} of $H$.
A hopset $H$ with stretch $1$, i.e, $\epsilon =0$, is called an \emph{exact} hopset.
We call an exact hopset ($\epsilon = 0$) with hopbound $\beta$ simply a $\beta$-hopset.


\subsection{Matrix Multiplication over Semirings}\label{sec:bmm}
\textbf{Matrix Notation.}
For a matrix B, we denote by $B[i,j]$ the entry in row $i$ and column $j$ of $B$. 
The transpose of $B$ is denoted $B^T$ . 

We denote by $\mathbb Z$ the set of integer numbers.
Let $R = \mathbb Z \cup \{-\infty, \infty \}$ and suppose that for two binary operations $\oplus : R \times R \rightarrow R $ and $\otimes : R \times R \rightarrow R$,
the structure $(R, \oplus, \otimes)$ is a ring without the requirement of having an inverse for the $\oplus$ operation.
Given two matrices $K$ and $L$ over $R$ of dimensions $ t \times n$ and $ n\times n$, respectively, for some positive integers $t$ and $n$,
the matrix product of $K$ and $L$ over $(R, \oplus, \otimes)$ is the $t \times n$ matrix $M$ defined as $M[i, j] = \oplus^n
_{k=1} (K[i, k] \otimes L[k, j])$ for
any $(i, j) \in \{1,\ldots , t\} \times \{1,\ldots, n\}$.
The Boolean matrix product is the matrix product over the
semiring $(\{0, 1\}, \vee, \wedge)$. 
In our setting, $t$ is typically $n^{\sigma}$, for some $0 < \sigma < 1$.

Note that the Boolean matrix product (henceforth, BMM) of $K$ and $L$ can be easily computed by treating the two matrices as standard integer matrices and computing 
the standard integer matrix product $\tilde M = K \cdot L$ over the integers, and then setting $M[i,j] = 1$ if $\tilde M[i,j] \ge 1$ and $0$, otherwise. 
Thus, the number of operations required to compute the BMM of $K$ and $L$ as above 
is $O(n^{\omega{(\sigma)}})$.
Vassilevska Williams et al. presented the current state-of-the-art upper bounds on $\omega(\sigma)$ (Table 1 of~\cite{VassilevskaBestRectangular}).
We present these bounds in Table~\ref{tab:tableLeGall} here for completeness.)
\begin{table}[!ht]
 \captionsetup{font=scriptsize}
  \begin{center}  
  {\small  
    \begin{tabular}{l|c} 
     $\sigma$ & \textbf{upper bound on $\omega(\sigma) $}\\
      \hline
                0.321334 & 2\\
                0.33 & 2.000100\\
                0.34 &  2.000600\\
                0.35 & 2.001363\\
                0.40 & 2.009541\\
                0.45 & 2.023788\\
        \end{tabular}
        \hspace{1em}
         \begin{tabular}{l|c}
          $\sigma$ & \textbf{upper bound on $\omega(\sigma) $}\\
          \hline
            0.50 & 2.042994\\
            0.527661 & 2.055322\\  
            0.55 & 2.066134\\
             0.60 & 2.092631\\
             0.65 & 2.121734\\
              0.70 & 2.153048\\                               
          \end{tabular}
          \hspace{1 em}        
      \begin{tabular}{l|c}
          $\sigma$ & \textbf{upper bound on $\omega(\sigma) $}\\
          \hline
                 0.75 & 2.186210\\               
                0.80 & 2.220929\\
                0.85 & 2.256984\\
                0.90 & 2.294209\\
                0.95  &2.332440\\
                1.00 & 2.371552\\
         \end{tabular}  
         }                  
  \end{center}
  \caption{Upper bounds on the exponent of $n$ in the number of operations required to 
multiply  an $n^{\sigma} \times n$  matrix by an $n \times n$ 
matrix (reproduced from~\cite{VassilevskaBestRectangular} here for convenience).} 
   \label{tab:tableLeGall}  
\end{table}
We denote by $\alpha$ the \emph{dual matrix multiplication exponent}, defined as the largest (supremum) $\alpha$ for which $\omega(\alpha) = 2$.

\emph{Distance Product} of two matrices $K$ and $L$ of appropriate dimensions is
denoted by $M = K \star L$. It is denoted as the product over $(R, \oplus, \otimes)$
with $\oplus = \min$ and $\otimes = +$, i.e., $M[i,j] = \min_{k} (K[i,k] + L[k,j])$.

\input{BasicTwoStep}
\input{BoundedSeparators}

\clearpage
\pagenumbering{gobble}
\bibliographystyle{alphaurl}
\bibliography{DirectedReachabilityBib}
\clearpage
\pagenumbering{roman} 
\setcounter{page}{1}
\pagestyle{plain}
\input{appendix}

\end{document}

%% file: BasicTwoStep.tex
\section{The Algorithm for General Graphs}\label{Sec:BTwoStepTheAlgo}
We describe our $S \times V$ reachability algorithm for general graphs in Section~\ref{Sec:BTwoStepTheAlgo}.
Our results for dense graphs appear in Section~\ref{sec:basicDense} and our results for sparse graphs appear in Section~\ref{sec:basicSparse}.

\subsection{The Algorithm}
In this section, we describe our two-step reachability  scheme.

Let $G = (V,E)$ be an $n$-vertex directed unweighted graph.
Let $S \subseteq V$ be a set of source vertices such that $|S| = n^{\sigma}$ for some $0< \sigma < 1$.
We compute $S \times V$ reachability by executing Algorithm~\ref{alg:nonrecursiveDireach}.
We will henceforth refer to it as \emph{Basic $S \times V$ Reachability Algorithm}.
The algorithm accepts as input the graph $G$ and the set $S$ of sources.
In addition, it accepts as input a parameter $D$. The algorithm is described for an arbitrary choice of $D$.
However, in fact, we will set it as a certain function of $n$ and $|S|$, which will be explicated in the sequel.
(See Equation~\eqref{eq:deltaNoRecursion2}.)
\begin{algorithm}[h!]
   {\small
   \caption{\small{DiReach($G,S, D$)}}
  	\label{alg:nonrecursiveDireach}
 		\begin{algorithmic}[1]
			\State \text{Compute a} $D$\text{-shortcut} $H$ \text {for} $G$ \text{(see Theorem~\ref{thm:koganParter1}) }; \label{ln:Step1}
			\State \text{Define a matrix} $A' =A + I,$ \text{~~$A$ is the adjacency matrix of} $G \cup H$; \label{ln:Step2}
			\State \text{Let} $B^{(1)} = A_{S*}$; \label{ln:Step3} \Comment{Matrix $A$ restricted to the rows corresponding to vertices in the set $S$.}
		 	\For{$k = 1$ \text{to} $D-1$}\label{ln:Step4}
        				\State \text{Compute} $B^{(k+ 1)} = B^{(k)} \cdot A'$;
     			 \EndFor\label{ln:Step5}
		 	\State \textbf{return} $B^{(D)}$;		 	
		\end{algorithmic}
		}
\end{algorithm}
The lines~\ref{ln:Step1} and~\ref{ln:Step2}  constitute the diameter reduction step.
We compute a $D$-shortcut $H$ (line~\ref{ln:Step1}) of the input digraph $G$ using~\cite{KoganParterDiShortcuts2} 
(see Theorem~\ref{thm:koganParter1})
and define an $n \times n$ boolean matrix $A' = A + I$, where $A$ is the adjacency matrix of $G \cup H$ (line~\ref{ln:Step2}).
In line~\ref{ln:Step3} we define a new matrix $B^{(1)}$,
which is the adjacency matrix of $G \cup H$ restricted to the rows corresponding to the set $S$ of source vertices.
Lines~\ref{ln:Step4} to~\ref{ln:Step5} constitute the reachability computation step.
We repeatedly perform a rectangular Boolean matrix product, and return the final product as our output.

Let $\mathcal T_{DR}$ and $\mathcal T_{RC}$ denote the time complexity of the diameter reduction step and reachability computation step, respectively.
Then the overall time complexity of our algorithm is $\mathcal T_{DR} + \mathcal T_{RC}$.
\subsection{Analysis for Dense Graphs}\label{sec:basicDense}
In this section, we analyze our two-step algorithm on general $n$-vertex, $m$-edge
digraphs, where $m$ may be as large as $n(n-1)$.

Let $D = n^{\delta}$ for some $0 < \delta \le 1/2$. 
In the diameter reduction step we compute a $D$-shortcut of $G$.
 By Theorem~\ref{thm:koganParter1}, this can be done in time $ \mathcal T_{DR} = \tilde O(m \cdot n/D^2 + n^{3/2}) = \tilde O(n^{3-2\delta})$.
In the reachability computation step we perform $D =  n^{\delta}$ iterations, each of which computes a rectangular matrix product of 
an $n^{\sigma} \times n$ matrix by an $n \times n$ matrix.
Each matrix product requires $O(n^{\omega(\sigma)})$ time.
It follows that $\mathcal T_{RC} = O(n^{ \omega(\sigma) + \delta } )$.
The overall time complexity of the algorithm is therefore $\mathcal T_{DR} +  \mathcal T_{RC} = \tilde O(n^{3-2\delta} + n^{ \omega(\sigma) + \delta })$.
Let $g(\sigma)$ denote the exponent of $n$ (as a function of $\sigma$) in the overall time complexity of our algorithm.
It follows that
\begin{equation}\label{eq:deltaNoRecursion1}
		g(\sigma) = \min_{\delta} \max \{3 -2\delta, \omega(\sigma) + \delta \}.
\end{equation}
Since $3-2\delta$ decreases when $\delta$ grows and $\omega(\sigma) + \delta$ grows,
the minimum is achieved when the following equation holds
\begin{align}\label{eq:deltaNoRecursion2}
	\begin{aligned}
		3-2\delta &= \omega(\sigma) + \delta \text{, i.e.,} \\
		\delta &= 1- \frac{1}{3} \cdot \omega(\sigma).		
	\end{aligned}
\end{align}
The parameter $D$ is therefore set as $D = n^{\delta}$, with $\delta$ given by Equation~\eqref{eq:deltaNoRecursion2}.
It follows from~\eqref{eq:deltaNoRecursion1} and~\eqref{eq:deltaNoRecursion2} that 
\begin{equation}\label{eq:g0Sigma}
g(\sigma) = 1 + \frac{2}{3} \cdot \omega(\sigma).
\end{equation}
The correctness of the basic two-step algorithm (Algorithm~\ref{alg:nonrecursiveDireach}) follows directly from the definition of a $D$-shortcut.
Indeed, it is easy to verify (by an induction on the for-loop index $k$) that after $k$ iterations of the for loop (Lines~\ref{ln:Step4} -\ref{ln:Step5}), the matrix $B^{(k + 1)}$
records $S \times V$ $(k+1)$-hop reachabilities in $G \cup H$, i.e., $B^{(k + 1)}[s,v] = 1$, iff $v$ is reachable from $s$ via a path with at most $k + 1$ hops, for any $s \in S, v \in V$.
Thus, the output matrix $B^{(D)}$ (after $D-1$ iterations of the for loop) records $S \times V$ $D$-hop reachabilities in $G \cup H$.
The latter are, however, by definition of a $D$-shortcut, the same as (unbounded-hop) reachabilities in the original graph $G$.

To compare our algorithm with existing solutions, recall that the time complexity $\mathcal{T}_{Naive}$ of the naive algorithm is given by $ \mathcal{T}_{Naive} = O(m \cdot |S|) = O( n^{2 +\sigma})$. Hence in the range $\omega -2 = 0.371552 \le \sigma \le 1$, the naive algorithm is no better than computing the transitive closure of the input digraph. The latter requires $\tilde O(n^{\omega})$ time (where $\omega = \omega(1))$.
For our algorithm to do better than the naive algorithm in this range we require
\begin{equation*}
             1 + \frac{2}{3} \cdot \omega(\sigma) < \omega, \text{~~i.e.,~~} \omega(\sigma) < \frac{3}{2}(\omega -1).
\end{equation*}
Let $\sigma^* \approx 0.535$ be the value such that $\omega(\sigma) = \frac{3}{2}(\omega -1)$.
Since $\omega(\cdot)$ is a monotone function, 
our algorithm computes $S \times V$ reachability in $o(n^\omega)$ time for $1 \le |S| \le n^{\sigma^*}$ sources, as opposed to 
naive algorithm that does so for $1 \le |S| \le n^{\omega -2}$.

Denote by $\tilde \sigma$ the threshold value such that 
\begin{equation}\label{eq:sigmaThresholdTS}
 g(\tilde \sigma) = 1 + \frac{2}{3} \omega(\tilde \sigma)  = 2 + \tilde \sigma.
\end{equation}

For $\tilde \sigma  \le \sigma < \omega -2$, $g(\sigma) = 1 + \frac{2}{3} \omega(\sigma) < 2 + \sigma$.
Observe that by convexity of $\omega(\sigma)$, it suffices to verify the inequality for the endpoints $\sigma = \sigma^*$ and $\sigma = \omega -2$.
Therefore, in the range $n^{\tilde \sigma} \le |S| < n^{\sigma^*}$, our algorithm outperforms the state-of-the-art solutions.

For all $\tilde \sigma < \sigma <1$, we have 
\begin{equation}\label{eq:sigmaThresholdTS1}
g(\sigma) < 2 + \sigma.
\end{equation}
As Inequality~\eqref{eq:sigmaThresholdTS1} holds for $\sigma = 0.3336$, it follows that $1/3 < \tilde \sigma < 0.3336$.
(Inequality~\eqref{eq:sigmaThresholdTS1} does not hold for $\sigma = 1/3$.)
See also Section~\ref{sec:basicSparse} (between Lemma~\ref{lem:thresh} and Lemma~\ref{lem:sparseSigmaZero}) for more details.

We conclude with the following theorem:
\begin{theorem}\label{thm: noRecursionDense}
Let $G = (V, E)$ be an $n$-vertex directed unweighted graph with $\Theta(n^2)$ edges.
Fix $S \subset V$, $|S| = n^{\sigma}$, for $\tilde \sigma \le \sigma \le \sigma^* \approx 0.535, \tilde \sigma \approx 0.3336$.
Then, Algorithm~DiReach($G,S, D$) (Algorithm~\ref{alg:nonrecursiveDireach}) with an appropriate choice of $D$ 
is a randomized algorithm that computes with high probability $S \times V$ reachability 
on $G$ in $\tilde O(n^{1 + \frac{2}{3} \omega(\sigma)})$ time.
This algorithm outperforms the state-of-the-art solutions for $S \times V$ reachability 
in this range of parameters.
\end{theorem}
For example, we compute $S \times V$ reachability for $S = n^{\sigma}$, $\sigma = 0.4$ 
in time $\tilde O(n^{g(0.4)}) = \tilde O(n^{2.33969})$, improving upon the state-of-the-art bound of $O(n^{2.4})$ 
 for naive algorithm and the state-of-the-art bound of $\tilde O(n^{2.371552})$ obtained by fast square matrix multiplication.
 In Table~\ref{tab:tableComparisonDense}, we present values of $g(\sigma)$ corresponding to some specific values of $\sigma$ in the range $(\tilde \sigma, 0.53]$
 and show how they compare to the exponent of $n$ in $S \times V$-naiveReach (see Definition~\ref{def:naiveReach})  and $S \times V$-squareReach (see Definition~\ref{def:squareReach}).
  \begin{table}[!h]
  \captionsetup{font=scriptsize}
  \begin{center}  
  {\small 
    \begin{tabular}{c|c|c|c} 
     $\sigma$ &$g(\sigma)$ &Exponent of $n$ in $\mathcal{T}_{Naive}$ & Exponent of $n$ in $\mathcal{T}_{Square} $\\
      \hline
                0.335 & 2.333565 & \textbf{2.335} & 2.371552 \\
                0.34 & 2.3337 &  \textbf{2.34} & 2.371552\\
                0.35 & 2.334241 &  \textbf{2.35} &2.371552\\
                0.36 &2.33533 &  \textbf{2.36} & 2.371552\\
                0.37& 2.336422 & \textbf{2.37} & 2.371552\\
                0.38 &2.33751& 2.38 & \textbf{2.371552}\\
                0.39 &2.3386 & 2.39& \textbf{2.371552}\\
                0.40 &2.33969 & 2.40 &\textbf{2.371552}\\
                0.41& 2.34159 & 2.41&\textbf{2.371552}\\
                0.42 &2.34349 & 2.42 & \textbf{2.371552}\\
                0.43 & 2.34539 &2.43 &\textbf{2.371552}\\
                0.44 & 2.34729 &2.44 & \textbf{2.371552}\\
                0.45 &2.34919 & 2.45&  \textbf{2.371552}\\
                0.46 & 2.35175 & 2.46 & \textbf{2.371552}\\
                0.47 & 2.35431& 2.47 &  \textbf{2.371552}\\
                0.48& 2.35687 & 2.48 & \textbf{2.371552}\\
                0.49 &2.359435 & 2.49 &  \textbf{2.371552}\\
                0.50 & 2.3621996 & 2.5 & \textbf{2.371552}\\
                0.51& 2.365081& 2.51 & \textbf{2.371552}\\
                0.52 & 2.368166 & 2.52 & \textbf{2.371552}\\
                0.53 & 2.371252 & 2.53 & \textbf{2.371552}                
        \end{tabular}  
        }               
  \end{center}
  \caption{Comparison of our algorithm from Theorem~\ref{thm: noRecursionDense} with the $S \times V$-naiveReach (BFS-based)  and the $S \times V$-squareReach (TC-based) methods in terms of exponent of $n$ in their respective time complexities.
 Previous state-of-the-art exponents (for corresponding exponents of $n$ in $S$) are marked by bold font.
 Our exponent (given in the column titled $g(\sigma)$) is better than the previous state-of-the-art in the entire range presented in this table (i.e., $0.335 \le \sigma \le 0.53$).
 }    
   \label{tab:tableComparisonDense}  
\end{table}
\subsection{Analysis for Sparser Graphs }\label{sec:basicSparse}
In this section we argue that when our input $n$-vertex $m$-edge digraph $G$ has $m = \Theta(n^{\mu})$ edges, for
some $\mu <2$, then our bounds from Theorem~\ref{thm: noRecursionDense} can be further improved.
Let $m = \Theta(n^{\mu})$ be the size of the edge set of $G$, where $0 \le \mu < 2$.
As before,  let $D = n^{\delta}$, for $0 < \delta \le 1/2$.
Applying Theorem~\ref{thm:koganParter1} for $m = \Theta(n^{\mu})$
and $D = n^{\delta}$, it follows that $\mathcal T_{DR} = \tilde O(n^{1 + \mu -2\delta} + n^{3/2})$.
Recall that the running time of the reachability computation step is
$\mathcal T_{RC} = O(n^{\omega(\sigma) + \delta})$.
The second term of $\mathcal T_{DR}$ will always be dominated by $\mathcal T_{RC}$.
The overall time complexity of the algorithm is therefore 
$\mathcal T_{DR} +  \mathcal T_{RC} = \tilde O(n^{1 + \mu -2\delta} + n^{ \omega(\sigma) + \delta })$.
Let $g^{(\mu)}(\sigma)$ denote the exponent of $n$ (as a function of $\sigma$) in the overall time complexity of our algorithm.
It follows that
\begin{equation}\label{eq:deltaNoRecursion1Sparse}
		g^{(\mu)}(\sigma) = \min_{\delta} \max \{1 + \mu -2\delta, \omega(\sigma) + \delta \}.
\end{equation}
This expression is minimized when
\begin{align}\label{eq:deltaNoRecursion2Sparse}
	\begin{aligned}
		1 + \mu-2\delta &= \omega(\sigma) + \delta \text{, i.e.,}\\
		\delta &=  \frac{1 + \mu - \omega(\sigma)}{3}.		
	\end{aligned}
\end{align} 
It follows from~\eqref{eq:deltaNoRecursion1Sparse} and~\eqref{eq:deltaNoRecursion2Sparse} that
\begin{equation}\label{eq:g0mu}
    g^{(\mu)}(\sigma) = \frac{1 + \mu + 2\omega(\sigma)}{3}.
\end{equation}
Note that for $\mu = 2$, we have $g^{(\mu)} (\sigma) = g(\sigma)$.
Also, when $\mu <2$, the exponent $g^{(\mu)}(\sigma)$ is smaller than $\omega$ for a wider range of $\sigma$ than
 the exponent $g(\sigma)$.
 On the other hand, the sparsity of the input digraph also improves the time complexity $\mathcal{T}_{Naive} = O(m \cdot n^{\sigma}) = O(n^{\mu + \sigma})$
 of the naive algorithm. 
 It is possible that for a given combination of $\sigma$ and $\mu$, $\mathcal{T}_{Naive}$ 
 may be better than the time complexity, $\tilde O(n^{\frac{1 + \mu + 2\omega(\sigma)}{3} })$, of our algorithm.
For our algorithm to outperform the naive algorithm, we need
\begin{equation*}
		\frac{1 + \mu + 2\omega(\sigma)}{3} < \mu + \sigma.
\end{equation*}
This condition is equivalent to
\begin{equation}\label{eq:muLower}
		\mu >    \omega(\sigma)  + \frac{1}{2} - \frac{3}{2} \sigma.		
\end{equation}
 By substituting values of $\omega(\sigma)$ from Table~\ref{tab:tableLeGall} for specific values of $\sigma$, we get corresponding lower bounds on $\mu$.
In Table~\ref{tab:table0}, we present these lower bounds for various values of $\sigma$.
\begin{table}[h!]
 \captionsetup{font=scriptsize}
  \begin{center}
   {\small 
    \begin{tabular}{l|c} 
      \textbf{$\sigma$} & \textbf{$\mu $}\\
      \hline
                0 & $>2.5$\\
                0.34 &$> 1.9906$\\
                0.4 & $>1.91$\\
                0.5 &$> 1.793$\\
                0.6& $ > 1.693$ \\
                0.7 &$>  $ 1.603\\
                0.8 & $>1.521$\\
                0.9 & $> 1.444$\\
                1 & $> \omega - 1$
  \end{tabular}
  }
  \end{center}
   \caption{Values of $\mu$ for which $g(\sigma) < \mu + \sigma$.
   The condition $\mu > 2.5$ in the first row means that for small values of $\sigma$ state-of-the-art 
   solutions outperform our algorithm for all densities of input graph.}
    \label{tab:table0}
\end{table}

For our algorithm to work in $o(n^{\omega})$ time, we need 
\begin{equation}\label{eq:muUpper}
		\frac{1 + \mu + 2\omega(\sigma)}{3} < \omega \text{, i.e.,~}
		\mu < 3 \omega - 2\omega(\sigma) - 1.
\end{equation}
For $\sigma =1$, Inequality~\eqref{eq:muUpper} implies that $\mu < \omega -1$, matching the lower bound $\mu > \omega -1$ implied by~\eqref{eq:muLower}. Indeed, our algorithm improves existing bounds only for $\sigma <1$.
For every $\sigma < 1$, we get some non-trivial intervals $I_{\sigma}$ of $\mu$ for which $g^{(\mu)}(\sigma) < \min \{\mu + \sigma, \omega\}$.
For example, for $\sigma = 0.5$, $\mu > 1.793$ (see Table~\ref{tab:table0}) ensures that $g^{(\mu)}(\sigma) < \mu + \sigma$.
For inequality $g^{(\mu)}(\sigma)  < \omega$ to hold, substituting $\sigma = 0.5$ in~\eqref{eq:muUpper} implies that
$\mu < 2.028668$. It follows that the inequality always holds for all $\mu$, i.e., the condition on $\mu$ for $\sigma =0.5$ is that $\mu > 1.793$.
We computed these values of intervals for various values of $\sigma$.
Specifically, for a number of values of $\sigma$, using the best known upper bounds on $\omega(\sigma)$, 
we numerically computed the intervals $I_{\sigma} = (\mu_1, \mu_2)$, such that for graphs with $m = \Theta(n^{\mu})$ edges, $\mu_1 < \mu < \mu_2$,
our algorithm improves the existing bounds on the $S \times V$ reachability problem.
In Table~\ref{tab:table1}, we present these intervals for various values of $\sigma$.
\begin{table}[h!]
  \captionsetup{font=scriptsize}
  \begin{center}
  {\small 
    \begin{tabular}{l|c} 
      \textbf{$\sigma$} & \textbf{Interval $I_{\sigma} = (\mu_0, \mu_1)$ }\\
      \hline
      0.335&$ \mu > 1.99785$\\
      0.34 &$\mu > 1.9911$\\
           0.4 & $\mu > 1.91$ \\
      	0.5 & $\mu > 1.793$ \\
      	0.55 & $(1.74, 1.982)$\\
      	0.6 & $(1.693, 1.93)$\\
      	0.7 & $(1.603, 1.809)$\\
      	0.8 & $(1.521, 1.673)$\\
     	 0.9 & $(1.4442, 1.526)$\\
     	 0.99 & $(1.3787, 1.3872)$\\
           1 & $(\omega-1, \omega-1) = \phi$
  \end{tabular}
  }
  \end{center}
    \caption{Intervals of values of $\mu$ for which $g^{(\mu)}(\sigma) < \min \{\mu + \sigma, \omega \}. $
                 These intervals are not empty for any $\sigma$, $\tilde \sigma \le \sigma <1$, where $\tilde \sigma <0.3336$ is a universal constant.
          }
    \label{tab:table1}
\end{table}
These results suggest that for all $\sigma, \tilde \sigma < \sigma < 1$, there is a non-empty interval of $\mu$ 
such that for digraphs with $m = \Theta(n^\mu)$ edges, our algorithm improves over existing bounds for $S \times V$ reachability
for $|S| = n^{\sigma}$. We will soon prove that this is indeed the case.

In Table~\ref{tab:tableComparisonMu}, we present a comparison of our algorithm with the previous state-of-the-art algorithms for graphs with $m = \Theta(n^\mu)$ edges 
for various combinations of $\mu$ and $\sigma$.
  \begin{table}[!ht]
  \captionsetup{font=scriptsize}
  \begin{center}   
  {\small
    \begin{tabular}{c|c|ccc} 
     $\mu$ & $\sigma$ &$g^{(\mu)}(\sigma)$ &Exponent of $n$ in $\mathcal{T}_{Naive}$ & Exponent of $n$ in $\mathcal{T}_{Square} $\\
      \hline
                 1.95 & 0.375 & 2.32 & \textbf{2.325} & 2.371552\\
                 &0.4 &  2.323 & \textbf{2.35} & 2.371552\\
                 & 0.45 & 2.3325 & 2.4 & \textbf{2.371552}\\
                 & 0.50 & 2.345 & 2.45 & \textbf{2.371552}\\ \hline
                1.9 & 0.45 & 2.3159 & \textbf{2.35} & 2.371552\\ 
                 & 0.5 & 2.3287 & 2.4 & \textbf{2.371552}\\
                 & 0.55& 2.344 & 2.45& \textbf{2.371552}\\
                 & 0.6 & 2.362 & 2.5 & \textbf{2.371552}\\  \hline
                1.75 & 0.55  &2.294 & \textbf{2.3} & 2.371552\\
                 & 0.6 & 2.312 & \textbf{2.35} &2.371552\\
                & 0.65 & 2.331 &2.4 &\textbf{2.371552}\\
                 & 0.7 & 2.352 & 2.45 &\textbf{2.371552}\\ \hline
                 1.525& 0.8 &2.323 & \textbf{2.325}&2.371552\\
                  & 0.85 & 2.346 & 2.375 & \textbf{2.371552}\\
                 & 0.9 & 2.3711 & 2.425 & \textbf{2.371552} \\ \hline                   
        \end{tabular} 
        }                
  \end{center}
  \caption{Comparison of our algorithm for graphs with $m = \Theta(n^{\mu})$ edges
  with the $S \times V$-naiveReach  and $S \times V$-squareReach methods in terms of exponent of $n$ in their respective time complexities. For each pair $(\mu, \sigma)$ that we list in the table, we have $\mu \in I_{\sigma}$.
 We mark the previous state-of-the-art bounds by bold font.}    
   \label{tab:tableComparisonMu}  
\end{table}

Recall that $\frac{1}{2} + \omega(\sigma) - \frac{3}{2} \sigma$ is the lower bound on the value of $\mu$ implied by Inequality~\eqref{eq:muLower},
i.e., for $\mu > 1/2 + \omega(\sigma) - (3/2) \cdot \sigma$, our algorithm is better than the $S \times V$-naiveReach method.
The threshold  $ \tilde \sigma$ is the value such that $\frac{1}{2} + \omega(\tilde \sigma) - \frac{3}{2} \tilde \sigma =2$, i.e., $\omega( \tilde \sigma) = \frac{3}{2} (1 +  \tilde \sigma)$ holds.
(Recall that we defined it in Section~\ref{sec:basicDense} as the value that satisfies $1 + \frac{2}{3} \omega( \tilde \sigma) = 2+  \tilde \sigma$.
These two conditions are obviously equivalent.)
\begin{lemma}\label{lem:thresh}
For every $\tilde \sigma < \sigma <1  $, we have 
\begin{align}
\begin{aligned}\label{eq:threshold}
\frac{1}{2} + \omega(\sigma) - \frac{3}{2} \sigma < 2 \text{, i.e., } \omega(\sigma) < \frac{3}{2}(1 + \sigma).
\end{aligned}
\end{align}
\end{lemma}
\begin{proof}
Note that for $\sigma = \tilde \sigma$, we have $\omega(\sigma) = \frac{3}{2} (1 + \sigma)$ and for $\sigma =1$, $\omega =\omega(1) < \frac{3}{2}(1 + \sigma) =3$. The Inequality~\eqref{eq:threshold}
now follows by convexity of the function $\omega(\sigma)$.
\end{proof}
To evaluate the value of $\tilde \sigma$ we note that for $\sigma= 0.34$,
we have $\omega(\sigma) = 2.0006$ (see Table~\ref{tab:tableLeGall})
and $\frac{3}{2}(1 + \sigma) = 2.01$.
For $\sigma =1$, we have $\omega(\sigma) = 2.371552$ and $\frac{3}{2}(1 + \sigma) = 3$.
Thus $\omega(\sigma) < \frac{3}{2} ( 1 + \sigma)$ holds for $\sigma = 0.34$ and for $\sigma =1$.
By convexity of the function $\omega(\sigma)$, it follows that this inequality holds for all $0.34 \le \sigma \le 1$.
However, for $\sigma = 1/3$, $\omega(\sigma) > \omega(0.33)>  2.0001 > 2$, whereas $\frac{3}{2} ( 1 + \sigma) =2$.
Therefore, there exists a constant $\tilde \sigma$, $1/3 < \tilde \sigma < 0.34$ (satisfying $\omega(\tilde \sigma) = \frac{3}{2} ( 1 + \tilde \sigma)$),
such that the inequality $\omega(\sigma) < \frac{3}{2}  (1 + \sigma)$
(and thus $g(\sigma) > \omega(\sigma) + \frac{1-\sigma}{2}$) holds for $\tilde \sigma < \sigma \le 1$.
In fact, one can verify that $1/3 < \tilde \sigma < 0.3336$, as $\omega(0.3336) < 2.00028$, while 
$\frac{3}{2} ( 1+ 0.3336) = 2.0004$.

In the following lemma we prove that for $\tilde \sigma < \sigma < 1$, the intervals $I_{\sigma}$ are non-empty.
\begin{lemma}\label{lem:sparseSigmaZero}
There exists a threshold $\tilde \sigma < 0.3336$ such that for every
$\tilde \sigma < \sigma < 1$, we have 
$\frac{1}{2} + \omega(\sigma) - \frac{3}{2} \sigma < 3\omega - 2\omega(\sigma) -1$,
 and there is equality for $\sigma =1$.
 \end{lemma}
\begin{proof}
For $\sigma = \tilde \sigma$, we have $\frac{1}{2} +  \omega(\sigma) - \frac{3}{2}\sigma =2$ and $3\omega - 2\omega(\sigma) -1 >2$
(as $\tilde \sigma < 0.34$, and $\omega(\cdot)$ is a monotonically increasing function).
For $\tilde \sigma < \sigma < 1$, we argue that $\frac{1}{2} +  \omega(\sigma) - \frac{3}{2}\sigma < 3\omega - 2 \omega(\sigma) -1$, i.e., 
$\frac{3}{2}  + 3\omega(\sigma) - \frac{3}{2} \sigma < 3 \omega$.
It is sufficient to prove that
\begin{align*}
	\begin{aligned}
		\frac{1}{2}(1-\sigma) < \omega(1) - \omega(\sigma),
	\end{aligned}
\end{align*}
 or equivalently
 \begin{align}\label{eq:slopeOmega}
	\begin{aligned}
		\frac{\omega(1) - \omega(\sigma) } {1-\sigma} > \frac{1}{2}.
	\end{aligned}
\end{align}
We prove that the Inequality~\eqref{eq:slopeOmega} holds in the following claim.
\begin{claim}\label{clm:slope}
For $\tilde \sigma < \sigma <1$, we have
$\frac{\omega(1) - \omega(\sigma) } {1-\sigma} > \frac{1}{2}$.
\end{claim}
\begin{proof}
Let $\sigma_1 = \alpha =  0.321334$ and $\sigma_2 = 1$.
According to Table~\ref{tab:tableLeGall}, $\omega(\sigma_1) = 2$ and $\omega(\sigma_2) = \omega =  2.371552$.
It follows that 
\begin{align}\label{eq:slope}
	\begin{aligned}
			\frac{\omega(\sigma_2) - \omega(\sigma_1) } {1-\sigma_1} = 0.5475 > 1/2.
	\end{aligned}
\end{align}
Note that $\frac{\omega(1) - \omega(\sigma_1)}{1 - \sigma_1}$ is the slope of the line connecting points $(\sigma_1, \omega(\sigma_1))$ and $(1, \omega(1))$.
By convexity of the function $\omega(\cdot)$, this slope is smaller than the respective slope $\frac{\omega(1) - \omega(\sigma)}{1 - \sigma}$, for any $\sigma_1 < \sigma <1$.
Hence the latter slope is greater than $1/2$ too.
See Figure~\ref{fig:plot} for an illustration.
\begin{figure}[ht!]
 \captionsetup{font=scriptsize}
	\centering
	\fbox{
          \includegraphics[scale=0.15]{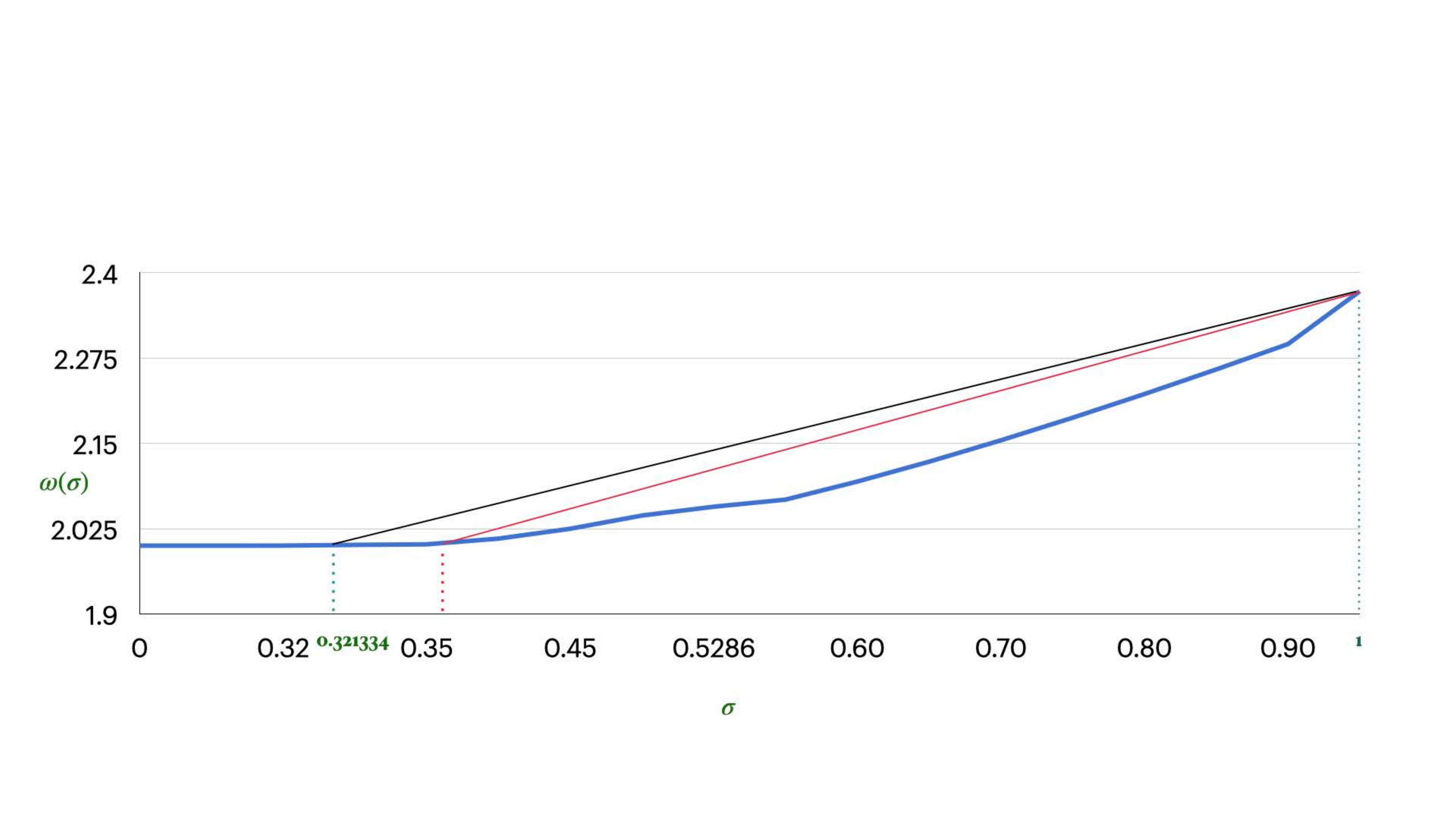}
	}
	\caption{The segment connecting the points $(0.321334, \omega(0.321334))$ and $(1,\omega(1))$ lies above the plot of $\omega(\sigma)$. Its slope is however smaller than the slope of the segment connecting the points $(\sigma, \omega(\sigma))$ and $(1,\omega(1))$ for any $\tilde \sigma < \sigma < 1$. The latter segment is illustrated by a red line.}
	\label{fig:plot}
\end{figure}

\end{proof}
Hence, the assertion of the lemma holds for $\tilde \sigma < \sigma < 1$.
For $\sigma =1$, we have $\frac{1}{2} + \omega(\sigma) - \frac{3}{2}\sigma = 3\omega - 2\omega(\sigma) -1 = \omega -1$.
Hence equality holds for $\sigma =1$.
\end{proof} 
We proved the following theorem.
\begin{theorem}\label{thm:sparseMain}
There is a threshold value $1/3 < \tilde \sigma < 0.3336$, such that for every $\tilde \sigma < \sigma < 1$, 
there exists a non-empty interval $I_{\sigma} \subseteq [1,2]$ that satisfies the following property:\\
For every $\mu \in I_{\sigma}$ and any $n$-vertex $m$-edge digraph with $m = n^{\mu}$ edges,
Algorithm~\ref{alg:nonrecursiveDireach} computes $S \times V$ reachability for $|S| = n^{\sigma}$
faster than the current state-of-the-art algorithms for the problem. 
(The latter are $S \times V$-naiveReach method (see Definition~\ref{def:naiveReach}) and $S \times V$-squareReach method (see Definition~\ref{def:squareReach}).)
See Table~\ref{tab:table1} for sample values of $I_{\sigma}$ and Table~\ref{tab:tableComparisonMu}
for sample exponents of our and state-of-the-art algorithms.
The exponent of our algorithm is given by $\frac{1 + \mu + 2\omega(\sigma)}{3}$.
\end{theorem}

%% file: BoundedSeparators.tex
\section{Graphs with Sublinear Recursive Separators}\label{Sec:BoundedSep}
In this section, following the classical work of Cohen~\cite{Cohen93},
we consider $S\times V$ reachability and $(1+\epsilon)$-approximate shortest distances in directed graphs that admit \emph{recursive separators} of sublinear size.
\subsection{Definitions}\label{sec:separatorDef}
An $n$-vertex undirected graph $G = (V,E)$ admits a
\emph{balanced separator} of size $n^{\rho}$, for a parameter $0 < \rho < 1$,
if there exists a subset $S \subseteq V$ of $n^{\rho}$ vertices such that $V\setminus S$
splits into two disjoint subsets $V_1$, $V_2$, $V_1 \cup V_2 = V\setminus S$,
with at most $\frac{2}{3} n$ vertices each, so that there are no edges in $E$ that cross between $V_1$ and $V_2$. The constant $2/3$ is called the \emph{ratio} of the separator.

A family $\mathcal F$ of undirected graphs admits a \emph{recursive $k^{\rho}$-separator},
for a parameter $0 < \rho <1$, if for every graph $G = (V,E) \in \mathcal F$ and every subset 
$U \subseteq V$, the induced subgraph $G[U] = (U, E[U])$ admits a balanced separator of size
$O(|U|^{\rho})$.

A \emph{directed} $n$-vertex graph $G = (V,E)$ is said to have a balanced separator to size $n^{\rho}$ 
if its undirected skeleton\footnote{\emph{Undirected skeleton} of a directed graph $G = (V,E)$ is the undirected graph $G'= (V,E')$,
$E' = \{ (u,v) \mid \langle u,v \rangle \in E \text{~or~} \langle v,u \rangle \in E  \}$
}
has a separator of this size.
A family of \emph{directed} graphs $\mathcal F$ has a recursive $k^{\rho}$-separator if for every
graph $G = (V,E) \in \mathcal F$ and every subset $U \subseteq V$, the undirected skeleton of induced subgraph $G[U]$ has a balanced separator of size $O(|U|^{\rho})$.

Given an $n$-vertex (directed or undirected) graph $G = (V,E)$ that admits a recursive $k^{\rho}$-separator, the separator decomposition tree $T_G$ of $G$ is the following binary tree:
Each node $t$ of $T_G$ corresponds to a vertex set $\mathcal V(t) \subseteq V$, and it is labelled by three vertex sets.
First, by the vertex set $\mathcal V(t)$ itself, second by a \emph{separator set} $\mathcal S(t)$ of $\mathcal V(t)$ and finally, the 
\emph{boundary set} $\mathcal B(t) \subseteq \mathcal V(t)$. (We will soon elaborate on the meaning of these subsets.) We also denote by $G(t) = (\mathcal V(t), E[\mathcal V(t)])$, the subgraph of $G$ induced by $\mathcal V(t)$.

Let $r$ be the root of a separator decomposition tree $T_G$ of $G$.
The vertex set $\mathcal V(r)$ of the root node $r$ is set as $\mathcal V(r) = V$, i.e., it contains all the vertices of the underlying graph $G$.
Its separator set $\mathcal S = \mathcal S(r)$ is a balanced separator of size at most $O(n^{\rho})$ of $G$, and its boundary set $\mathcal B(r)$ is an empty set.

Given a node $t$ of $T_G$ with its separator and boundary sets $\mathcal S(t)$ and $\mathcal B(t)$, respectively, 
its two children are defined as follows:\footnote{A node $t$ has children in $T_G$ if its vertex set $\mathcal V(t)$ is sufficiently large,
i.e.,  $|\mathcal V(t)|$ is at least some sufficiently large constant. Otherwise, $t$ is a leaf of $T_G$.
Specifically, the threshold size $\tau$
of leaves should be sufficiently large to ensure that for all values $\tau' \ge \tau$,
the separator size $O(\tau'^{\rho})$ is smaller than $0.01 \tau'$. (The leading constant $0.01$
is chosen arbitrarily.)
}
Let $U = \mathcal V(t)$.
Let $U_1 \cup U_2 = U \setminus \mathcal S(t)$, $U_1 \cap U_2 = \phi$, be the two subsets of $U \setminus \mathcal S(t)$
of size at most $\frac{2}{3}|U|$ each, so that $E \cap (U_1 \times U_2) = \phi$.
(The existence of $U_1$, $U_2$ is guaranteed by the definition of a recursive balanced separator.) Then the two children of $t$ in $T_G$ are nodes $t_1$ and $t_2$ with vertex sets $\mathcal V(t_i) = U_i \cup (\mathcal S(t) \cap \hat {\Gamma}(U_i))$, for $i \in \{1,2 \}$, 
 where $\hat{\Gamma}(U_i)$ is the set of vertices of $U_i$ and their neighbours.
The separator sets $\mathcal S(t_i)$ (unless $t_i$ is a leaf) are defined in the natural way,
 as balanced separators of the respective vertex sets $\mathcal V(t_i)$.
 Finally, their boundary sets $\mathcal B(t_i)$, $i \in \{ 1, 2 \}$,
 (again, assuming that $t_i$ is not a leaf) are given by 
 \begin{equation}\label{eq:boundarySet}
       \mathcal B(t_i) =  (\mathcal S(t) \cup (\mathcal B(t)) \cap \mathcal V(t_i),
 \end{equation}
i.e., it contains all the vertices from the separator of its parent $t$ and those boundary vertices of
$t$ that ended up in the vertex set $\mathcal V(t_i)$.
The \emph{level} $\ell(t)$ of a node $t \in T_G$ is defined in a natural way
as the depth of the subtree of $T_G$ rooted at $t$.

Note that if a vertex $u \in \mathcal S(t)$ is not adjacent to any vertex of $U_1 \cup U_2$, then it can be added to the smaller set among $U_1$, $U_2$ (or one of them, if they are of equal size). Hence we can assume w.l.o.g. that every separator vertex $ u \in \mathcal S(t)$ has at least one neighbor either in $U_1$ or in $U_2$.
Moreover, if $u$ has neighbours only in $U_1$ but not in $U_2$, then it can be added to $U_1$. The resulting partition $U'_1 = U_1 \cup \{ u\}$, $U'_2 = U_2 $, $\mathcal S'(t) = \mathcal S(t) \setminus \{u\}$ still satisfies that there are no edges that cross from $U'_1$ to $U'_2$.
Let $U''_1$, $U''_2$ be the sets obtained by the process of iteratively removing from the separator all
vertices $u$ which have edges to just one side of the partition, and adding them to that side.
Note that $|U''_1|$, $|U''_2| \le \frac{2}{3}n + O(n^{\rho}) \le (\frac{2}{3} + 0.01) n$.
Hence this process results in a separator $\mathcal S''(t)$ that satisfies 
the property that each of its vertices is adjacent to both parts $U''_1$, $U''_2$ of the partition,
and the partition is still sufficiently well-balanced. (It is immaterial whether the leading constant in the
definition of separator is $2/3$ or $2/3 + 0.01$.) We will therefore assume w.l.o.g.
that our separators have this property. As a result, since $\mathcal S(t) \subseteq \hat{\Gamma} (U_i)$, for $i \in \{1,2 \}$,
the definitions of $\mathcal V(t_i)$, $\mathcal B(t_i)$ can be simplified to $\mathcal V(t_i) = U_i \cup \mathcal S(t)$ and $\mathcal B(t_i) = \mathcal S(t) \cup (\mathcal B(t) \cap \mathcal V(t_i))$.

\subsection{A Separator-Based Reachability Algorithm}\label{sec:SeparatorDAvail}
In this subsection, we devise another $S \times V$-reachability algorithm for graphs that admit $k^{\rho}$-recursive separators.
Like the algorithm of~\cite{Cohen93}, this algorithm assumes that a separator decomposition of the 
input graph $G$ is provided to the algorithm as input. Unlike our algorithm in Section~\ref{Sec:BTwoStepTheAlgo}, and similarly to~\cite{Cohen93},
this algorithm can be efficiently implemented in the \textsf{PRAM} setting.
Moreover, we will also extend it to $S \times V$ {\em $(1 + \epsilon)$-approximate distance computation}.

The algorithm itself combines Cohen's approach with fast rectangular matrix multiplication in augmented graphs. We start with a short overview of Cohen's algorithm, focusing for now on the reachability problem. The algorithm of Cohen~\cite{Cohen93} computes an $O(\log n)$-shortcut 
$G' = (V, H)$ of the input graph 
$G = (V,E)$. The shortcut and the augmented graph contain $O(n^{2\rho} + n\log n) = O(n^{2\rho})$ edges
(recall that we consider the setting where $\rho > 1/2$),
that is, $|E \cup H| = O(n^{2\rho})$.
Then, it computes $S \times V$ reachabilities in $\tilde G = G \cup G' = (V, E \cup H)$
via $O(\log n)$-bounded Bellman-Ford.
The latter provides (unrestricted) $S \times V$ reachabilities in the original graph $G$, as they are identical to the respective $O(\log n)$-bounded reachabilities in $\tilde G$.

The shortcut is computed by following reachabilities:
for every node $t \in T_G$, 
we need reachabilities in $G(t)$ between all pairs $( \mathcal B(t) \times \mathcal B(t) ) \cup ( \mathcal S(t) \times \mathcal S(t))$.
For each pair $(u,v) \in (\mathcal B(t) \times \mathcal B(t) ) \cup ( \mathcal S(t) \times \mathcal S(t))$, we insert the arc $\langle u,v \rangle$
into the shortcut $H$ iff $v$ is reachable from $u$ in $G(t)$.
(When building a hopset, the construction is the same
except the edge $e = \langle u,v \rangle$ is assigned weight $w(e) = d_{G(t)}(u,v)$.)
This computation requires polylogarithmic time and work 
$O({\sum}_{t \in T_G} (|\mathcal S(t)| + |\mathcal B(t)|)^{\omega})$ in the \textsf{PRAM} model.
(We also have $\sum_{t \in T_G} |\mathcal S(t)|^{\omega}, \sum_{t \in T_G} |\mathcal B(t)|^{\omega} \le \tilde O(n^{\omega \rho})$.
See Section $4.1$ and the proof of Theorem $5.1$ of~\cite{Cohen93}.
See also Lemma~\ref{lem:SeparatorBoundarySizeBounds} below for a proof sketch of these statements.)
It therefore follows that 
\begin{align}\label{eq:timeFromBasics}
    \begin{aligned}
\underset{t \in T_G}{\sum} (|\mathcal S(t)| + |\mathcal B(t)|)^{\omega} &\le \underset{t \in T_G, |\mathcal S(t)| > | \mathcal B(t)|}{\sum}
(2 |\mathcal S(t)|)^{\omega} + \underset{t \in T_G,  |\mathcal B(t)| \ge |\mathcal S(t)|}{\sum} (2 |\mathcal B(t)|)^{\omega}\\
       &\le  2^{\omega} \cdot \left(\underset{t \in T_G}{\sum} (| \mathcal S(t)|)^{\omega} + 
       \underset{t \in T_G} {\sum}(|\mathcal B(t)|)^{\omega} \right) = 2^{\omega} \cdot \tilde O(n^{\omega \rho }) = \tilde O(n^{\omega \rho}).
   \end{aligned}
\end{align}
Moreover, ~\cite{Cohen93} showed that if one sets weights $w(u,v)$ for every edge $\langle u,v \rangle$ in the shortcut by 
$w(u,v) = d_{G(t)}(u,v)$, then this shortcut is, in fact, an exact $O(\log n)$-hopset as well.

Once the shortcut is in place, running in parallel an $O(\log n)$-bounded Bellman-Ford from every 
source $s \in S$ requires additional polylogarithmic time and work $O(|S| \cdot (|E| +  |H|)) = O(n^{\sigma + 2\rho})$. (It is shown in~\cite{Cohen93} that $|H| = O(n^{2\rho})$; see the third bullet of Theorem 5.1 of~\cite{Cohen93}.)
Thus, the overall time of the algorithm is polylogarithmic in $n$, while its work complexity is
$\tilde O(n^{\omega \rho} + n^{\sigma + 2\rho})$.
As was mentioned above, in the centralized setting the complexity $O(n^{\sigma + 2\rho})$ is trivial,
and thus Cohen's algorithm does not provide non-trivial bounds for $S \times V$-direachability (or shortest paths) in the centralized setting.
Our algorithm builds the shortcut by Cohen's procedure, but then it uses the shortcut in the same way as we used Kogan-Parter's shortcut in Sections~\ref{sec:basicDense} and~\ref{sec:basicSparse}.
Specifically, we build the adjacency matrix $A$ of the augmented graph $\tilde G = G \cup G' = (V, E \cup H)$, and the rectangular $n^{\sigma } \times n$ matrix $B$ formed of the $n^{\sigma}$ rows of $A$ that
correspond to the sources from $S$.
We then iteratively multiply $B$ by $A$ for $O(\log n)$ times (like in Algorithm~\ref{alg:nonrecursiveDireach}), and obtain the $O(\log n)$-bounded $S \times V$ reachabilities in 
$\tilde G$.
As those are the same as unrestricted $S \times V$ reachabilities in $G$, it follows that the algorithm solves the problem of computing $S \times V$ reachabilities in $G$.

The running time of the second stage of the algorithm is $O(\log n) \cdot n^{\omega(\sigma)} = \tilde O(n^{\omega(\sigma)})$.  (As each rectangular matrix product requires $n^{\omega(\sigma)}$
time, and the algorithm computes $O(\log n)$ such products.
We note that for $\sigma \le \alpha$, even though $\omega(\sigma) = 2$, the running time is $n^{2 + o(1)}$. However, for larger $\sigma$,
the time is $n^{\omega(\sigma)}$ rather than $n^{\omega(\sigma) + o(1)}$.)
Hence, the overall centralized running time of the algorithm is $n^{\omega \rho} + n^{\omega(\sigma)}$. 
Moreover, each such matrix product can also be computed in polylogarithmic time using
$n^{\omega(\sigma) }$ work in the \textsf{PRAM} model (see, e.g.,~\cite{ElkinN22}).
Hence our algorithm requires polylogarithmic time and work $n^{\omega \rho} + n^{\omega(\sigma)}$.

Our exponent $F(\sigma, \rho) = \max \{\omega \rho, \omega(\sigma) \}$
improves upon Cohen's exponent $C(\sigma, \rho) = \max \{\omega \rho, \sigma + 2\rho \}$,
whenever the following two conditions hold:
\begin{enumerate}
\item \label{itm:1}  $\sigma + 2\rho > \omega \rho, \text{~i.e.,~} \rho < \frac{\sigma}{\omega -2}$, and
\item \label{itm:2}  $\omega(\sigma) < \sigma + 2\rho$.
\end{enumerate}
Condition~\ref{itm:1} is a restriction only as long as $\sigma < \omega -2 \approx 0.371552$.
For $\sigma = \alpha$ (which is only slightly larger than the minimum value $\tau$ for which our algorithm starts providing improvements for some values of $\rho$), we have $\rho < \frac{\sigma}{\omega -2} < 0.867 $
and the condition~\ref{itm:2} implies $\rho > 1 - \frac{\alpha}{2} = 0.839$.
The threshold value $\tau$ is obtained by solving $1 - \frac{\sigma}{2} = \frac{\sigma}{\omega - 2}$. It is given by
\begin{equation}\label{eq:threshImprovementSep}
\tau = \frac{2(\omega - 2)}{\omega} \approx 0.31334.
\end{equation}
See Table~\ref{tab:rhoInterval2} for these intervals $J_{\sigma} = (\frac{\omega(\sigma) - \sigma}{2}, \frac{\sigma }{\omega - 2})$ of values of $\rho$ for which our \textsf{PRAM} algorithm improves upon the  previous state-of-the-art bound~\cite{Cohen93}. 
See also Table~\ref{tab:comparisonRho2} for comparison between our exponent $F(\sigma, \rho)$ and the exponent of Cohen's algorithm~\cite{Cohen93} $C(\sigma, \rho)$ within the ranges $J_{\sigma}$, 
for various values of $\sigma$.
\begin{theorem}\label{thm: separator1}
Given a directed $n$-vertex graph $G = (V,E)$ that admits a $k^{\rho}$-recursive separator along with its
separator decomposition, our algorithm computes $S \times V$-direachability in it 
(for a set $S$ of $n^{\sigma}$ sources)  in polylogarithmic \textsf{PRAM} time using $O(n^{\omega \rho } + n^{\omega(\sigma)}) $ work.

For every $\sigma > \tau = \frac{2(\omega - 2)}{\omega}$,
there is a non-empty interval $J_{\sigma} = \left(\frac{\omega(\sigma) - \sigma}{2}, \frac{\sigma }{\omega - 2}\right)$ such that for every $\rho \in J_{\sigma}$, our algorithm outperforms the algorithm of~\cite{Cohen93}.
\end{theorem}
\begin{remark}
We note that $n$-vertex graphs that admit an $O(k^{\rho})$-recursive separator for $\rho < 1$ have $O(n)$ edges.
(See~\cite{LiptonRoseTarjan79}, Corollary 12.)
Thus, $S \times V$ reachabilities and distances,
for $|S| = n^{\sigma}$, can be computed in them in \emph{centralized} $\tilde O(n^{1 + \sigma})$ time.
\end{remark}

\begin{table}[h!]
  \captionsetup{font=scriptsize}
  \begin{center}
  {\small 
    \begin{tabular}{l|c |c} 
      \textbf{$\sigma$} & \textbf{Interval $J_{\sigma} = (\frac{\omega(\sigma) - \sigma}{2}, \frac{\sigma}{\omega -2})$ } & \textbf{Interval $\hat J_{\sigma} = (\omega(\sigma) - \sigma - 1, \frac{2\sigma}{\omega -2} -1)$ } \\
      \hline
      $\alpha = 0.321324$ & $ (0.8386, 0.8674)$ & $$(0.6672, 0.7348)$$\\
       0.37 & (0.8173, 0.9958) & $$(0.6346, 0.9916)$$\\
      0.4 & $\rho > 0.8047$ & $q > 0.6094$\\
      0.5 & $\rho > 0.7714$ & $q > 0.5428$ \\
      0.6 & $ \rho > 0.7463$ & $q > 0.4926$\\
       0.7 & $\rho > 0.7265 $ & $q > 0.4530$\\
       0.8 & $ \rho > 0.7145 $ & $q > 0.4290$\\
       0.9 & $\rho > 0.6971 $ & $q > 0.3942$ \\
       $\sigma \rightarrow 1$  & $ \rho > \frac{\omega - 1}{2} \approx 0.6858$ & $q > \omega - 2$          
  \end{tabular}
  }
  \end{center}
    \caption{Intervals $J_{\sigma}$ such that for every $\rho \in J_{\sigma}$
             our algorithm improves upon Cohen's algorithm for \textsf{PRAM} $S \times V$-reachability 
             in graphs that admit $k^{\rho}$-recursive separators.
             The last line indicates that as $\sigma$ tends to $1$, the left endpoint of the interval
             tends to $\frac{\omega -1}{2}  \approx 0.685$.
              $J_{\sigma}$ is not empty for every $\sigma > \frac{2(\omega- 2)}{\omega}$.
              The table also provides analogous intervals $\hat J_{\sigma} = (\omega(\sigma) - \sigma - 1, \frac{2\sigma}{\omega -2} -1)$ of values of the exponent $q$ of $k = n^q$, (respectively, $\xi$ of $g = n^{\xi})$ for $k$-NN
              graphs and intersection (and overlap) graphs of $k$-ply neighbourhood systems for dimension 
              $d = 2$, (respectively, of graphs with genus $g$) in which our exponents outperform the state-of-the-art bound of~\cite{Cohen93}. See Sections~\ref{sec:reachShortBoundedGenus} and ~\ref{sec:reachDistGeometric} for further details.
             }
    \label{tab:rhoInterval2}
\end{table}

\begin{table}[h!]
  \captionsetup{font=scriptsize}
  \begin{center}
  {\small 
    \begin{tabular}{c|c|c|c|c} 
     $\sigma$ & $\rho$ &$q = \xi =  2\rho - 1$&$F(\sigma, \rho) = F(\sigma, q) = F(\sigma, \xi)$ &$\min\{\omega, C(\sigma, \rho) \} = \min\{\omega, C(\sigma, q) \}$\\
      \hline
                 $\alpha = 0.321314$ & 0.85 &0.7 & 2.0145 &$ \alpha + 1.7 = 2.021314$\\ \hline
                 0.4 & 0.85 & 0.7 & 2.0145 & 2.1\\
                     &0.9 &0.8 & 2.1344 & 2.2\\
                     &0.95 & 0.9 &  2.2530&2.3\\ \hline  
                0.5 & 0.8 &0.6 & 2.0429 & 2.1\\
            \rowcolor{Gray}
                    & 0.85 &0.7 & 2.0429 & 2.2\\
            \rowcolor{Gray}
                    &0.9 &0.8 & 2.1344 & 2.3\\
                    &0.95 &0.9 & 2.2530 & $ \omega = 2.371552$\\ \hline
                0.6 & 0.75 &0.5 & 2.0926 & 2.1 \\
                    &0.8 & 0.6& 2.0926 & 2.2\\
            \rowcolor{Gray}
                    &0.85 & 0.7 &2.0926 &2.3\\
            \rowcolor{Gray}
                    &0.9 &0.8 &2.1344 & $\omega$\\
                    & 0.95 &0.9 &2.2530& $ \omega$\\ \hline
                0.7 &0.75 &0.5 &2.153 & 2.2\\
                    &0.8 &0.6   & 2.153 & 2.3\\
            \rowcolor{Gray}
                    &0.85 &0.7 & 2.153 & $ \omega$\\
            \rowcolor{Gray}
                    &0.9 &0.8  & 2.153 & $ \omega$ \\
                    &0.95 &0.9 & 2.2530 & $ \omega$\\ \hline
                0.8 & 0.75 &0.5  & 2.2209 & 2.3\\
            \rowcolor{Gray}
                    & 0.8 &0.6  & 2.2209 & $ \omega$\\
            \rowcolor{Gray}
                    & 0.85 &0.7 & 2.2209 & $ \omega$\\
            \rowcolor{Gray}
                     & 0.9  &0.8 & 2.2209 & $ \omega$\\
                     & 0.95 &0.9 & 2.2530 & $ \omega$\\ \hline
                0.9 & 0.7 &0.4 & 2.2942 & 2.3 \\
                    & 0.75 &0.5 & 2.2942 & $ \omega$ \\
                    &0.8 &0.6 & 2.2942 & $ \omega$\\
                    & 0.85 &0.7 & 2.2942 & $ \omega$\\
                    & 0.9 &0.8 & 2.2942 & $ \omega$\\
                    & 0.95 &0.9 & 2.2942 & $ \omega$              
        \end{tabular}
  }
  \end{center}
    \caption{Sample exponents of the work complexities of our algorithm and of the previous
              state-of-the-art~\cite{Cohen93} on graphs that admit $k^{\rho}$-recursive separators,
              for pairs $(\sigma, \rho)$ with $\rho \in J_{\sigma}$.
              These exponents are the same as the ones achieved for $2$-dimensional $k$-NN 
              graphs and intersection and overlap graphs of $k$-ply neighbourhood systems, 
              with $k = n^q$ and for graphs with genus $g = n^{\xi}$. See Sections~\ref{sec:reachShortBoundedGenus} and~\ref{sec:reachDistGeometric}.              
              Highlighted rows indicate particularly substantial improvements.
          }
    \label{tab:comparisonRho2}
\end{table}

\subsection{Graphs with a Flat Bound on the Size of the Recursive Separator}\label{sec:flatBound}
In this section we analyze the running time of our algorithm for building 
shortcuts on graphs $G = (V,E)$, $|V| = n$, for which every subgraph admits 
a bounded separator of size at most $c_{sep} \cdot n^{\rho}$, $0 < \rho < 1$, and
$c_{sep} > 0$ is a universal constant.

In Appendix~\ref{sec:treeDecompAppedix} we show that at the cost of a slight increase in the size of
recursive separator, one can build a tree decomposition $T_G$ that has the property 
that for every non-leaf node $t \in T_G$, its separator {\em has ratio $1/2$}.
Specifically, the size of the recursive separator grows by a logarithmic factor, i.e., it becomes $O(n^{\rho} \cdot \log n)$. This argument extends Lipton-Tarjan's proof (\cite{LiptonRoseTarjan79}, Corollary 3)
that \emph{applies} for planar graphs.
(In fact, we also show there
 that this is the case (without the increase by a factor of $O(\log  n)$) 
 for all graphs that admit $k^{\rho}$-recursive separators, for any $0 < \rho < 1$.)

 We say that a separator $S$ of a graph $G = (V,E)$  that
 partitions the graph into two parts $A$, $B$ is \emph{doubly-incident}
 if every vertex $v \in S$ has at least one neighbor in $A$ and at least 
 one neighbor in $B$. 
 By Corollary~\ref{cor:doublyIncident} (in Appendix~\ref{sec:treeDecompAppedix}), 
 $n$-vertex graphs $G = (V,E)$ with a flat bound of $c_{sep} \cdot n^{\rho}$
 on their recursive separators admit also a doubly-incident separator os size 
 $O(n^{\rho} \cdot \log n)$, that splits the graphs into parts of size at most
 $n/2 + O(n^{\rho} \cdot \log n)$ each.

Denote $C = c'_{sep} \cdot n^{\rho} \cdot \log n$, for a universal constant $c'_{sep}$, the upper bound on the size of each separator set $S(t)$, for every node $t \in T_G$. 
(This constant may be different from the constant $c_{sep}$ 
defined above.)
Let $r_T$ be the root node of $T_G$.
Let $t_1$, $t_2$ be its two children.
Then $|\mathcal V(t_1)|$, $|\mathcal V(t_2)| \le \frac{n}{2} + C$.

\begin{lemma}\label{lem:faltSize}
Let $t \in T_G$ be a node on level $i$ of $T_G$, i.e., its distance from
$r_T$ in $T_G$ is $i$. Then,
$$ |\mathcal V(t)| \le \frac{n}{2^i} + (2 -\frac{1}{2^i} ) \cdot C \le \frac{n}{2^i} + 2C. $$
\end{lemma}
\begin{proof}
The proof is by induction on $i$.
The induction base $(i = 0)$ holds as $|\mathcal V(r_T)| \le n$.

Consider a node $t$ of level $i \ge 1$, and let $t'$ be its parent.
By the inductive hypothesis,
\begin{align*}
    |\mathcal V(t')| \le \frac{n}{2^{i-1}} + ( 2 - \frac{1}{2^{i -1}}) \cdot C.
\end{align*}
  Hence,
\begin{align*}
    |\mathcal V(t)| \le \frac{|\mathcal V(t')|}{2} + C \le \frac{n}{2^i} + (1- \frac{1}{2^i}) \cdot C + C
     = \frac{n}{2^i} + (2 - \frac{1}{2^i}) \cdot C.
\end{align*}
\end{proof}

We continue the decomposition up until level $k \approx \log_2 (\frac{n^{1-\rho}}{\log n})$ on which the nodes
$t$ satisfy 
\begin{align*}
    |\mathcal V(t) | \le n^{\rho} \cdot \log n + 2\cdot C = n^{\rho}\cdot \log n + 2 c'_{sep} \cdot n^{\rho} \cdot \log n = (2c'_{sep} + 1) \cdot n^{\rho} \cdot \log n.
\end{align*}
Denote by $T_k$ the set of leaves of $T_G$, i.e., nodes on level $k$.
It follows that the leaves of $T_G$ contribute to the centralized running time of the algorithm that creates the shortcut at most 
\begin{align*}
       \tilde O \left (\sum_{t \in T_k} |\mathcal V(t)|^{\omega} \right) \le \tilde O( 2^k \cdot ((2 c'_{sep} + 1)n^{\rho} )^{\omega} )
                                                  \le  \tilde O( n^{1 - \rho} \cdot O(n^{\omega \rho}) ) = \tilde O(n^{\rho \omega + 1 - \rho}).
\end{align*}
Other nodes $t \in T_G$ contribute together $O(\sum_{t \in T_G\setminus T_k} |\mathcal S(t)|^{\omega})$ and $O(\sum_{t \in T_G\setminus T_k} |\mathcal B(t)|^{\omega})$. The former is bounded by:
\begin{align*}
    \sum_{t \in T_G\setminus T_k} |\mathcal S(t)|^{\omega} \le O(2^k \cdot C^{\omega}) \le  \tilde O(n^{1 - \rho} \cdot n^{\rho \omega})
                                                   = \tilde O(n^{\rho \omega + 1 - \rho}).
\end{align*}
Also, for every internal (i.e., non-leaf) node $t \in T_G\setminus T_k$, 
denote by $W_t$ the set of ancestors of $t$ (in $T_G$).
Then,
\begin{align*}
|\mathcal B(t)| \le \sum_{t' \in W_t} |\mathcal S(t)| \le k \cdot C \le c'_{sep} \cdot \log^2 n \cdot n^{\rho}.
\end{align*}
Thus, 
\begin{align}\label{eq:FlatBound1}
    \begin{aligned}
          \sum_{t \in T_G\setminus T_k} |\mathcal B(t)|^{\omega}  \le 2^k \cdot (c'_{sep} \cdot \log^2 n \cdot n^{\rho})^{\omega}
           \le \tilde O(n^{\rho\omega + 1 -\rho}).
    \end{aligned}
\end{align}
The following corollary follows:
\begin{cor}\label{cor:flatTime}
Let $1/2 < \rho < 1$.
Consider $n$-vertex graphs
that admit a flat bound of $O(n^{\rho})$ on their recursive separator.
The centralized running time of our algorithm for creating a shortcut for these graphs 
is $\tilde O(n^{1 - \rho + \rho \omega})$.
In parallel setting, assuming that the tree decomposition is provided, the same bound applies to the work complexity of our algorithm,
while the running time is polylogarithmic in $n$.
\end{cor}

Bernstein et al.~\cite{BernsteinGS21} devised a near-linear time randomized centralized  $O(\log^3 n)$-approximation algorithm for the problem  of computing a tree decomposition. Specifically, given an $n$-vertex $m$-edge graph 
$G$ of treewidth $t(G)$, their (centralized) algorithm computes a tree decomposition with width $O(t(G) \cdot \log^3 n)$ in time $m^{1 + o(1)}$. A previous result of Chuzhoy and Saranurak~\cite{ChuzhoyS21} provides the same approximation in deterministic $n^{2 + o(1)}$
time. In the sequel we will use the latter result.

We note that a graph $G$ with treewidth $t(G)$ has a balanced separator of size $s(G) \le t(G) + 1$, and conversely, a graph with balanced separator of size $s(G)$ has treewidth $O(s(G) \log n)$~\cite{BODLAENDERGHK95}.
(In fact, for large values of $t(G)$ and $s(G)$ the ratio between them is at most $O(1)$.)
Hence in the centralized setting a tree decomposition with treewidth larger by at most $\log^{O(1)} n$ than the  optimal one (or than the size of optimal recursive separator) can be computed in deterministic $n^{2 + o(1)}$ time.
Thus, in the centralized setting, our algorithm does not need to accept a tree decomposition as input.
Rather, it can compute \emph{from scratch} a near-optimal tree decomposition (i,e, optimal up to a factor $O(\log^3 n)$) within desired
time bounds.
This is, however not the case, in general, for our parallel results.

\subsection{Reachabilities and Approximate Distances in Graphs with a Flat Bound on the size of the Separator}\label{sec:reachDistFlatBound}
In this section we analyze our algorithm for computing reachabilities/approximate distances in graphs with a flat bound of $O(n^{\rho})$ on the size of the recursive separator for some constant $0 < \rho < 1$.
Equivalently, one can think of graphs with treewidth $O(n^{\rho})$.
The algorithm is based on our algorithm for building shortcuts in these graphs, described in Section~\ref{sec:flatBound}.

In centralized setting, using the approximation algorithm of~\cite{ChuzhoyS21} we can compute a tree decomposition $T_G$ with each bag size $\tilde O(n^{\rho})$ within deterministic $n^{2 + o(1)}$ time.
In parallel setting we assume that this tree decomposition is in place before the algorithm starts working.
Given the tree decomposition, by Corollary~\ref{cor:flatTime}, our algorithm computes the required shortcut within $\tilde O(n^{\omega \rho + 1 - \rho})$ centralized time or in polylogarithmic time and 
$\tilde O(n^{\omega \rho + 1 - \rho})$ processors.
As we will argue in Section~\ref{sec:ApproxDistance} (see Theorem~\ref{thm:approxDistances}),
in the case of hopsets both bounds need to be multiplied by $O_{\epsilon} (\log W) \cdot \log^{O(1)} n$.
Given the shortcut/hopset our centralized algorithm then computes $S \times V$ reachabilities/approximate distances within additional time $n^{\omega(\sigma)}$, $|S| = n^{\sigma}$, i.e., 
the exponent $F'(\sigma, \rho)$ of our algorithm is given by 
\begin{align*}
    F'(\sigma, \rho) = \max \{ \omega \rho + 1 - \rho, \omega(\sigma) \}.
\end{align*}
 (In parallel setting this is the exponent of our number of processors, while the time is polylogarithmic in $n$.) The running time of the naive centralized algorithm on such graphs is
$\min \{ n^{\omega}, \tilde O(n^{\mu + \sigma})\}$,
where $m = n^{\mu}$ is the number of edges of the input graph.
We first  consider
the densest case, when $\mu = 1 + \rho$, i.e., the running time becomes 
$\min \{ n^{\omega} , \tilde O(n^{1 + \rho + \sigma})  \}$,
i.e., its exponent is $N'(\sigma, \rho) = \min \{\omega, 1 + \rho + \sigma \}$.

In the case of parallel algorithm, once shortcut is in place, Cohen's algorithm~\cite{Cohen93}
spends polylogarithmic time and $\tilde O(n^{1 + \rho + \sigma})$ processors to compute the 
reachabilities (even if $\mu < 1 + \rho$, because it operates on top of the shortcut, and the size of the latter is $\tilde O(n^{1 + \rho})$ in this case).
Hence its exponent of the number of processors is
\begin{align*}
    C'(\sigma, \rho) = \min \{\omega, \max\{ \omega \rho + 1 -\rho, 1 + \rho + \sigma\} \}.
\end{align*}
Note that $\omega(\sigma) < \omega$ for every $\sigma < 1$, and also 
$\omega \rho + 1 - \rho < \omega \rho + \omega \cdot ( 1- \rho) = \omega$,
for every $\rho < 1$.
Thus, for $F'(\sigma, \rho)$ to improve $N'(\sigma, \rho)$ (and, as a result, also to improve $C'(\sigma, \rho)$ too),
we need 
\begin{enumerate}[label=(\arabic*)]
    \item \label{condition1} $1 + \rho + \sigma > \omega \rho + 1 - \rho$, and
    \item \label{condition2} $1 + \rho + \sigma > \omega(\sigma)$.
\end{enumerate}
 Condition~\ref{condition1} implies that 
 \begin{align}\label{eq:cond1Implies}
     \rho < \frac{\sigma}{\omega -2}.
 \end{align}
 The condition is non-trivial only for $\sigma < \omega -2$.
 Condition~\ref{condition2}  implies that
 \begin{align}
     \rho > \omega(\sigma) - (1 + \sigma).
 \end{align}
 For $\sigma \le \alpha$, we get
 \begin{align}\label{q:cond2Implies}
     \rho > 1 - \sigma.
 \end{align}
 For the upper bound~\eqref{eq:cond1Implies} to be greater than the lower bound~\eqref{q:cond2Implies} we need
 \begin{align}
     \frac{\sigma}{\omega -2} > 1 - \sigma, \text{~i.e.,~}  \sigma > \frac{\omega - 2}{\omega -1} \approx 0.2701.
 \end{align}
Hence for $\frac{\omega - 2}{\omega -1} < \sigma \le \alpha$, the interval of values of 
$\rho$ in which $F'(\sigma, \rho) <  N'(\sigma, \rho)$ is $(1- \sigma, \frac{\sigma}{\omega -2})$, and it is non-empty for all 
$\sigma$ in this range.
 For $\sigma \ge \omega -2$, the interval of values of $\rho$ in which $F'(\sigma, \rho) < N'(\sigma, \rho) $
 is $(\omega(\sigma) - 1 - \sigma, 1)$.
 As $\omega(\sigma) - 1 - \sigma < 1$ (i.e., $\omega(\sigma) < 2 + \sigma$)
 for all $0 < \sigma < 1$, this interval is always non-empty.

 For $\alpha < \sigma < \omega -2$, the condition $\omega(\sigma) - 1 - \sigma < \frac{\sigma}{\omega -2}$
 is equivalent to $\omega(\sigma) < 1 + \sigma \cdot \frac{\omega - 1}{\omega -2}$.
 This holds for $\sigma > \frac{\omega -2}{\omega -1}$.
 Therefore, for each $\sigma > \frac{\omega - 2}{\omega -1}$, there is a range of $\rho$ such that 
 \begin{align}
     F'(\sigma, \rho) < N'(\sigma, \rho), \text{~i.e.,}
 \end{align}
our algorithm outperforms the naive algorithm (and a parallel variant of our algorithm outperforms Cohen's algorithm~\cite{Cohen93}; both parallel algorithms assume that the tree decomposition is in place).
We denote these ranges by $K'_{\sigma} = (\omega(\sigma) - (1+\sigma), \min\{1, \frac{\sigma}{\omega -2} \})$.

   Another interesting threshold $\kappa_\sigma$ is the value such that for all $\rho \le \kappa_\sigma$, the exponent $F'(\sigma,\rho)$ of our algorithm is equal to $\omega(\sigma)$, i.e., it is the best possible.
   This happens whenever $\omega(\sigma) \ge \omega \rho + 1 - \rho$, i.e., 
   $\rho \le \frac{\omega(\sigma)}{\omega - 1} = \kappa_\sigma$. 
   We denote by $K''_\sigma$ the intersection of $K'_\sigma$ with the interval $\rho \le \kappa_\sigma$, i.e., the range in which our algorithm achieves the best possible bound and improves the state-of-the-art.
Table~\ref{tab:tableRhoInterval}  provides numerical values for the ranges $K'_\sigma$ and $K''_\sigma$.

In Table~\ref{tab:FlatBoundCompare} we provide some sample values of the  exponents $F'(\sigma,\rho)$ and $N'(\sigma,\rho)$.
\begin{table}[h!]
  \captionsetup{font=scriptsize}
  \begin{center}
  {\small 
    \begin{tabular}{c|c|c} 
      \textbf{$\sigma$} & $K'_{\sigma}$ & $K''_\sigma$\\
      \hline
            0.3 & (0.7,0.8074)   & (0.7,0.7291)\\
            0.4 & (0.6095,1)  & (0.6095,0.7360)\\
      	  0.5 & (0.5430,1) & (0.5430,0.7604)\\
      	  0.6 &  (0.4926,1) & (0.4926,0.7966) \\
      	  0.7 & (0.4530,1) & (0.4530,0.8407) \\
      	  0.8 & (0.4209,1) & (0.4209, 0.8902\\
     	    0.9 & (0.3942,1) & (0.3942,0.9436) \\
          $\sigma \rightarrow 1$ & $(\omega -2,1)$ & $(\omega-2,1)$
  \end{tabular}
  }
  \end{center}
    \caption{Intervals $K'_{\sigma} = (\omega(\sigma) - (1 + \sigma), \min \{1 , \frac{\sigma}{\omega -2} \})$
              for various values of $\sigma$ are provided. 
              For $\rho \in K'_{\sigma}$ our $S \times V$ reachability algorithm for graphs with treewidth $n^{\rho}$ outperforms 
              the (\emph{naive}) state-of-the-art solutions. For $\rho \in K''_\sigma$, our exponent not only improves the state-of-the-art, but it is actually the best possible, as it is equal to $\omega(\sigma)$. 
              When $\sigma$ tends to $1$, $\omega(\sigma) - (1 + \sigma)$ tends to $\omega -2$ and $\kappa_\sigma$ (the right endpoint of $K''_\sigma$) tends to 1.
          }
    \label{tab:tableRhoInterval}
\end{table}

\begin{table}[h!]
  \captionsetup{font=scriptsize}
  \begin{center}
  {\small 
    \begin{tabular}{l|c|c|c|c} 
      $\sigma$ & $K'_{\sigma}$ & $\rho$ & $F'(\sigma, \rho)$ & $N'(\sigma, \rho)$\\
      \hline
           0.3 & (0.7, 0.81) & 0.75 & 2.0287 & 2.05\\ \hline
           0.4 & $\rho > 0.60959$  & 0.7 & $\omega(0.4) =  2.0095$ & 2.1 \\
               && 0.8 & 2.0972 & 2.2\\
               && 0.9 & 2.2344 & 2.3\\ \hline
      	0.5 & $ \rho > 0.5430$ & 0.6 & $\omega(0.5) =  2.0429$ & 2.1 \\
         \rowcolor{Gray}
                      && 0.7 & $\omega(0.5)$ & 2.2\\
          \rowcolor{Gray}
                      && 0.8 & 2.0972 & 2.3\\
                      && 0.9 & 2.2344 & $\omega$ \\ \hline
      	0.6 & $ \rho > 0.4926$ & 0.5 & $\omega(0.6)  =  2.0926$ & 2.1\\
                            && 0.6 & $\omega(0.6) $ & 2.2\\
         \rowcolor{Gray}
                            && 0.7 & $\omega(0.6) $ & 2.3 \\
          \rowcolor{Gray}
                            && 0.8 & 2.0972 & $\omega$\\
                            && 0.9 & 2.2344 & $\omega$\\ \hline
      	0.7 & $\rho > 0.4530$ & 0.5 & $\omega(0.7)= 2.1530$ & 2.2\\
                             && 0.6 & $\omega(0.7)$ & 2.3 \\
         \rowcolor{Gray}
                             && 0.7 & $\omega(0.7)$ & $\omega$\\
          \rowcolor{Gray}
                             && 0.8 & $\omega(0.7)$ & $\omega$\\
                             && 0.9 & 2.2344 & $\omega$\\ \hline
      	0.8 & $ \rho > 0.4209$ & 0.5 & $\omega(0.8) = 2.2209$ & 2.3\\
          \rowcolor{Gray}
                              && 0.6 & $\omega(0.8)$ & $\omega$\\
           \rowcolor{Gray}  
                              && 0.7 & $\omega(0.8)$ & $\omega$\\
            \rowcolor{Gray}
                              && 0.8 & $\omega(0.8)$ & $\omega$\\
                              && 0.9 & 2.2344 & $\omega$\\  \hline       
     	0.9 & $ \rho >0.3942$ & 0.4 & $\omega(0.9) =  2.2942$ & 2.3 \\
                              && 0.5 & $\omega(0.9)$ & $\omega$\\
                              && 0.6 & $\omega(0.9)$  & $\omega$\\  
                              && 0.7 & $\omega(0.9)$  & $\omega$\\    
                              && 0.8 & $\omega(0.9)$  & $\omega$\\   
                              && 0.9 & $\omega(0.9)$  & $\omega$\\    
                              && 0.95 & 2.23030 & $\omega$
  \end{tabular}
  }
  \end{center}
    \caption{Comparison between the exponent of our $S\times V$-reachability (and approximate distance computation algorithm) 
              on graphs with a flat bound of $n^{\rho}$ on their treewidth with the (naive) state-of-the-art solutions.
              Here we consider the densest regime, i.e., when the number of edges of the input graph is $m = \Theta(n^{1 + \rho})$.        
              }
    \label{tab:FlatBoundCompare}
\end{table}
\begin{theorem}\label{thm:treewidthDistance}
    Let $\epsilon > 0$, $W \ge 1$ be parameters.
    Consider a directed $n$-vertex $m$-edge graph $G = (V,E)$ with non-negative edge weights
    with aspect ratio at most $W$, $m = n^{\mu}$, $1 < \mu < 2$,
    and assume that the underlying graph of $G$ has treewidth $O(n^{\rho})$,
    for a parameter $0 < \rho < 1$.
    Let $S \subseteq V$, $|S| = n^{\sigma}$, $0 < \sigma <1$, be a set of sources.
    The centralized deterministic running time of our algorithm for $S \times V$ reachability is 
    $\tilde O(n^{F(\sigma, \rho)})$, with $F(\sigma, \rho) = \max \{\omega \rho + 1 - \rho, \omega(\sigma) \}$, and it works \emph{from scratch}.
     In the parallel setting it has polylogarithmic (deterministic) running time and work as above.
     For $S \times V$ $(1 + \epsilon)$-approximate distance computation the centralized running time and the parallel work complexities need to be multiplied by $O_{\epsilon}(\log W)\cdot \log ^{O(1)} n$.

     In the densest regime (i.e., when $\mu = 1 + \rho$) our exponent is smaller than the exponent of the state-of-the-art naive $S \times V$ reachability algorithm ($N'(\sigma, \rho) = \min \{\omega, 1 + \rho + \sigma \}$)
     in a wide range of parameters $\sigma$ and $\rho$. Specifically, for any $\frac{\omega -2}{\omega -1} < \sigma < 1$,
     there is an non-empty range $K'_{\sigma}$, such that for every $\rho \in K'_{\sigma}$,
     our algorithm outperforms the previous bound (i.e., $F'(\sigma, \rho) < N'(\sigma, \rho)$),
     see Tables~\ref{tab:tableRhoInterval} and~\ref{tab:FlatBoundCompare}.

     In the parallel setting our algorithm requires a tree decomposition with all bags of size $\tilde O(n^{\rho})$
     to be provided to it as input, and so does the previous state-of-the-art algorithm (due to~\cite{Cohen93}).
     As both our exponent and that of~\cite{Cohen93} do not improve when $\mu < 1 + \rho$,
     in the parallel setting this comparison applies to all graphs with treewidth $O(n^{\rho})$, and 
     not only in the densest regime.    
    \end{theorem}
In Appendix~\ref{app:comparison} we also consider
 the more general scenario in which our input $n$-vertex $m$-edge graph $G$ has treewidth $O(n^{\rho})$, 
but $m = n^{\mu}$, for some $0 < \rho < 1$, $1 < \mu \le 1 + \rho$.

\section{Approximate Distance Computation}\label{sec:ApproxDistance}
\subsection{Outline}\label{sec:ApproxDistanceOutline}
In this section, we extend our reachability algorithm from Section~\ref{sec:SeparatorDAvail}
to the problem of computing $S \times V$ $(1+\epsilon)$-approximate 
distances in directed graphs with non-negative edge weights.
The graphs are assumed to admit $k^{\rho}$-recursive separators
and moreover, their separator decomposition tree $T_G$ is given as a part of the input.
The algorithm is similar to Cohen's~\cite{Cohen93} algorithm for \emph{exact}
distance computation. 

We start with outlining Cohen's algorithm for exact $S \times V$ distance computation. Cohen's algorithm starts with computing an {\em exact} $O(\log n)$-hopset $H$ with $O(n^{2\rho} + n) = O(n^{2\rho})$ edges (recall that we consider the regime $\rho \ge 1/2$).
It then employs the hopset for computing $S\times V$ distances via Bellman-Ford algorithm on the augmented  graph $\tilde G = (V, E\cup H)$. 
The computation of Cohen's hopset relies on exact distance computations in all nodes $t$ of the separator
decomposition tree $T_G$. To make the algorithm more efficient, we replace these \emph{exact} computations  by $(1 + \epsilon)$-approximate ones.
We argue below that this results in a $(1 + \epsilon)^{O(\log n)}$-approximate hopset with the same size and hopbound as for Cohen's hopset. Our algorithm is also different from Cohen's algorithm in the second stage. Instead of running Bellman-Ford on graph $\tilde G$ (as Cohen's algorithm does), we iteratively compute approximate distance products between the rectangular sub-matrix of $\tilde G$ with its rows corresponding to the sources in $S$ and the (weighted) square adjacency matrix of $\tilde G$. We start with describing the first stage of the algorithm and then proceed to the second one.
\subsection{Computing the Approximate Hopset}\label{sec:approxHopset}
As was discussed above, our algorithm in this section is closely related
to Cohen's~\cite{Cohen93} algorithm (see Algorithm 4.1 of~\cite{Cohen93}) for building an exact hopset.
Let $\epsilon >0$ be a fixed parameter.
The algorithm proceeds bottom-up, starting from the leaves of $T_G$,
and dealing with a node $t \in T_G$ only after both its children $t_1$ and $t_2$ have been taken care of.
For a leaf $t$, all pairwise approximate distances in $\mathcal V(t)$
can be computed in 
 $O(1)$ time (as $|\mathcal S(t)|, |\mathcal B(t)| \le |\mathcal V(t)| = O(1)$).

In the general case, before the algorithm starts to handle an internal node
$t \in T_G$, it has already computed approximate distance estimates 
$\delta_1(u,v) = \delta_{t_1}(u,v)$ (respectively, $\delta_2(u,v) = \delta_{t_2} (u,v)$) of $d_{G(t_1)} (u,v)$ (respectively, $d_{G(t_2)} (u,v)$), for every vertex pair $(u,v) \in (\mathcal S(t_1) \times \mathcal  S(t_1)) \cup (\mathcal B(t_1) \times \mathcal B(t_1)) $ (respectively, 
$(u,v) \in (\mathcal S(t_2) \times \mathcal S(t_2)) \cup (\mathcal B(t_2) \times \mathcal B(t_2))$).
The objective is now to compute approximate distance estimates 
$\delta_t(u,v)$ for all $(u,v) \in (\mathcal S(t) \times \mathcal S(t)) \cup (\mathcal B(t) \times \mathcal B(t))$.
When this computation is over for all nodes $ t\in T_G$, the hopset 
is created by inserting all edges of $E_t = \{(u,v) \in (\mathcal S(t) \times \mathcal S(t)) \cup (\mathcal B(t) \times \mathcal B(t)) \}$, for every $t \in T_G$ into the hopset $H$, with corresponding weights $w(u,v) = \delta_t(u,v)$.

Next, we describe the computation in a node $t$, given the distance 
estimates for edge sets $E_1 = E_{t_1}, E_2 = E_{t_2}$ of its children $t_1$ and $t_2$, respectively.
This computation starts (\textsf{\textbf{Step 1}}) by defining the graph $G_S = (\mathcal S(t), \mathcal S(t) \times \mathcal S(t))$, with weights $w_S(u,v)$ (for every vertex-pair $(u,v) \in \mathcal S(t) \times \mathcal S(t)$) defined by
\begin{equation}\label{eq:weights}
    w_S(u,v) = \min \{\delta_1(u,v), \delta_2(u,v) \}.
\end{equation}
Observe that the set $\mathcal S(t)$ is contained in in both boundary sets 
$\mathcal B(t_1)$ and $\mathcal B(t_2)$ of children $t_1$ and $t_2$, respectively, of $t$,
and thus, both values $\delta_1(u,v)$ and $\delta_2(u,v)$ are available
at this stage.

The algorithm now computes $(1 + \epsilon)$-approximate all-pairs 
shortest paths (henceforth, $(1 + \epsilon)$-APASP) in $G_S$.
Denote by $\{\delta_S(u,v)~|~ (u,v) \in \mathcal S(t) \times \mathcal S(t) \}$, 
the resulting  estimates that are produced by this step. 
For pairs $(u,v) \in \mathcal S(t) \times \mathcal S(t) $, the algorithm returns these estimates, i.e., sets $\delta_t(u,v) = \delta_S(u,v)$.

To compute (\emph{approximate}) distances for pairs in $\mathcal B(t) \times \mathcal B(t)$, the algorithm proceeds to the next step (\textsf{\textbf{Step 2}}).
The algorithm now constructs the graph $G_{BS} = (\mathcal B(t) \cup \mathcal S(t), E_{BS})$,
$E_{BS} = (\mathcal B(t) \times \mathcal S(t)) \cup (\mathcal S(t) \times \mathcal S(t)) \cup (\mathcal S(t) \times \mathcal  B(t))$, with weights $w_{BS} : E_{BS} \rightarrow R$ defined as follows:
For $e = (u,v) \in \mathcal S(t) \times \mathcal S(t)$, we set 
\begin{equation}\label{eq:weight2}
     w_{BS}(u,v) = \delta_S(u,v),
\end{equation}
which was computed in the first step of the algorithm.

For $e = (u,v) \in \mathcal B(t) \times \mathcal S(t)$, observe that $v \in \mathcal S(t) \subseteq \mathcal  B(t_1) \subseteq \mathcal V(t_1)$, and also $v \in \mathcal S(t) \subseteq \mathcal B(t_2) \subseteq \mathcal V(t_2)$. The endpoint $ u \in \mathcal B(t)$ belongs to either $\mathcal V(t_1)$ or $\mathcal V(t_2)$, but not to both of them (unless $u \in \mathcal S(t)$, but this case was 
already taken care of). 
Let $i \in \{ 1,2\}$ be the index such that $u,v \in \mathcal V(t_i)$.
Observe also that $u \in \mathcal B(t_i)$ 
(as $u \in \mathcal B(t) \cap \mathcal V(t_i)$, see Equation~\eqref{eq:boundarySet}).

Since $u,v \in \mathcal B(t_i)$, we have the estimate $\delta_i(u,v)$ available.
We therefore set
\begin{equation}\label{eq:weight3}
w_{BS}(u,v) = \delta_i(u,v).
\end{equation}
The weights $w_{BS}(v,u)$ for edges $(v,u) \in \mathcal S(t) \times \mathcal B(t)$
are defined in the same way.

Now the algorithm computes $(1 + \epsilon)$-approximate all pairs
distance estimates $\delta_{BS}(u,v)$ in this graph $G_{BS}$.\footnote{As ~\cite{Cohen93} shows, it is sufficient to compute only $3$-hop-bounded
distances to and from every $u \in \mathcal B(t)$.
To simplify presentation, we however prefer to compute here all pairwise 
distances without restriction on the hopbound. In terms of work complexity,
Cohen's approach may save a logarithmic factor here.}
The final distance estimates for vertex-pairs $(u,v) \in \mathcal B(t) \times \mathcal B(t)$
is then produced in the following way:
If $(u,v) \in \mathcal B(t_i) \times \mathcal B(t_i)$, for $i \in \{1,2\}$,
then 
\begin{equation}\label{eq:finalEstimates}
    \delta_t(u,v) = \min \{\delta_i(u,v), \delta_{BS} (u,v)  \}.
\end{equation}
Otherwise, (for $(u,v) \in (\mathcal B(t_1) \times \mathcal B(t_2)) \cup (\mathcal B(t_2) \times \mathcal B(t_1))$, the algorithm returns 
\[
       \delta_t(u,v) = \delta_{BS}(u,v).
\]
This completes the description of the algorithm.

\subsection{Analysis of the Approximate Hopset}\label{sec:ApproxHopset}
$(1 + \epsilon)$-Approximate APSP in graphs $G_S$ and $G_{BS}$ are
computed via Zwick'a algorithm~\cite{APSPDirectedZwick}.
Its running time in $k$-vertex graphs with edge weights at most $W$
is $O(k^{\omega} \cdot \log k \cdot \log \frac{W}{\epsilon} \cdot \frac{1}{\epsilon})$.
The edge weights need to be non-negative, and $W$ is the ratio between the largest and the smallest non-zero edge weight. We henceforth refer to it as the {\em aspect ratio w.r.t. non-zero weights} or simply \emph{aspect ratio}.
In parallel setting it requires polylogarithmic time and 
$O(k^{\omega} \cdot \log k \cdot \log \frac{W}{\epsilon} \cdot \frac{1}{\epsilon})$ work~\cite{ElkinN22}.
In the context of our construction of hopsets, the overall running time is therefore
\begin{equation}\label{eq:appHopsetTime1}
    \sum_{t \in T_G} O(|\mathcal S(t)| + |\mathcal B(t)|)^{\omega}\cdot \log n \cdot \log \frac{W}{\epsilon} \cdot \frac{1}{\epsilon}.    
\end{equation}
where $W = O(M n)$. (Here $M$ is the maximum weight in the original graph $G$, assuming that the minimum edge weight is $1$.
We will soon argue that all weights in graphs $\{ G_S(t), G_{BS}(t)~|~ t \in T_G\}$ are at most $O(Mn)$.)
The sum 
$\sum_{t \in T_G}  (|\mathcal S(t)| + |\mathcal B(t)|)^{\omega}$ is, by Equation~\eqref{eq:timeFromBasics}, bounded by 
\[
       \sum_{t \in T_G} (|\mathcal S(t)| + |\mathcal B(t)|)^{\omega} \le O(1) \cdot \left(\sum_{t \in T_G}|\mathcal S(t)|^{\omega} + \sum_{t \in T_G}|\mathcal B(t)|^{\omega}\right).
\]
Both these sums $\sum_{t \in T_G}|\mathcal S(t)|^{\omega}$, $\sum_{t \in T_G}|\mathcal B(t)|^{\omega}$ are
bounded by at most $O(n^{\omega \rho} \cdot \log n)$ (whenever $\omega \rho > 1$).
(See~\cite{Cohen93}, Section 5.)
Hence the overall running time is 
\begin{equation}\label{eq:overallTime}
   O \big(\frac{1}{\epsilon} \cdot (\log^2 n) (\log n + \log M/\epsilon) \cdot n^{\omega \rho}\big).
\end{equation}

Next, we argue that the algorithm provides a $((1 + \epsilon)^{O(\log n)}, O(\log n))$-hopset of the input graph $G$. For a given node $ t\in T_G$, recall that $\ell(t)$ denotes the depth of the subtree of $T_G$ rooted at the node $t$. Recall also that $G(t)$ is the subgraph of $G$ induced by vertices of the node $t$,
$d_{G(t)}$ stands for the distance function in this graph, while $\delta_t$ is the distance estimate computed by our algorithm.
\begin{lemma}
For every node $t \in T_G$, and every vertex pair $(u,v) \in (\mathcal S(t) \times \mathcal S(t)) \cup (\mathcal B(t) \times \mathcal B(t))$,
the estimate $\delta_{t}(u,v)$ computed by our algorithm satisfies
\[
      (1+ \epsilon)^{3(\ell(t) + 1)} \cdot d_{G(t)}(u,v) \ge \delta_{t}(u,v) \ge d_{G(t)}(u,v).
\]
\end{lemma}
\begin{proof}
 The algorithms only uses upper bounds on the lengths of certain paths, and thus the right-hand inequality follows. We prove the left-hand inequality by induction on $\ell(t)$.
 
  \noindent \textsf{\textbf{Base:}} When $\ell(t) = 0$ (i.e., $t$ is a leaf), the algorithm uses Zwick's algorithm~\cite{APSPDirectedZwick} to compute $(1+ \epsilon)$-approximate 
 APSP in $G(t)$. Hence, 
 \[
     (1 + \epsilon) \cdot  d_{G(t)}(u,v) \ge \delta_t(u,v) \ge d_{G(t)}(u,v),
 \]
 as desired.

\noindent \textsf{\textbf{Step:}} Consider first a pair $(u,v) \in \mathcal S(t) \times \mathcal S(t)$, for $t$ with $\ell(t) \ge 1$.
 We assume inductively that for $i \in \{1,2 \}$, the computed distance estimates for pairs
 $(x,y) \in (\mathcal S(t_i) \times \mathcal S(t_i)) \cup (\mathcal B(t_i) \times \mathcal B(t_i))$ satisfy
 \[
       \delta_{t_i}(x,y) = \delta_i(x,y) \le (1+\epsilon)^{3(\ell(t_i) + 1)} \cdot d_{G(t_i)}(x,y).
 \]
 First, suppose that the shortest path $P(u,v) = (u = u_1, u_2, \ldots, u_h = v)$ in $G(t)$
 satisfies that $u,v \in  \mathcal S(t)$, and all other vertices $u_2,\ldots, u_{h -1} \in \mathcal V(t) \setminus \mathcal S(t) = (\mathcal V(t_1) \cup \mathcal V(t_2)) \setminus \mathcal S(t)$.
  
  Note that $u,v \in \mathcal S(t)$ implies that $u,v \in \mathcal B(t_1) \cap \mathcal B(t_2)$.
 Moreover, either $V(P) \subseteq \mathcal V(t_1)$ or $V(P) \subseteq \mathcal V(t_2)$.
 (Indeed, let $u = u_1 \in \mathcal S(t) \subseteq \mathcal V(t_1)$.
 If $u_2 \in \mathcal V(t_1)$, then if there exists an index $i < h$ such that
 $u_i \in \mathcal V(t_1)$ and $u_{i + 1} \in \mathcal V(t_2)$, then $\{u_i, u_{i+1} \} \cap \mathcal S(t) \neq \phi$.
 By the assumption above, this can be the case only if $i + 1 =h$, and then $\mathcal V(P) \subseteq \mathcal V(t_1)$.
 If $u_2 \in \mathcal V(t_2)$, then similarly, $V(P) \subseteq \mathcal V(t_2)$.)

 Suppose w.l.o.g. that $V(P) \subseteq \mathcal V(t_1)$.
 As $u, v \in \mathcal B(t_1)$, the weight $w(P)$ of the path $P$ is equal to 
 $d_{G(t_1)}(u,v)$.
 Also, by the inductive hypothesis,
 \[
     w_S(u,v) = \delta_1(u,v) \le (1 + \epsilon)^{3(\ell(t_1) +1)} \cdot d_{G(t_1)}(u,v)
    = (1 + \epsilon)^{3\ell(t)} \cdot d_{G(t_1)}(u,v).
\]
 Hence, 
 \begin{equation}\label{eq:deltasInduction}
                \delta_S(u,v) \le (1 + \epsilon) \cdot w_S(u,v) \le (1 + \epsilon)^{3\ell(t) + 1} \cdot d_{G(t_1)}(u,v).
\end{equation}
Since the shortest path $P = P(u,v)$ in $G(t)$ is contained in $\mathcal V(t_1)$, we have
$\delta_{G(t_1)}(u,v) = \delta_{G(t)}(u,v)$, and thus 
\[
       \delta_S(u,v) \le (1+ \epsilon)^{3\ell(t) + 1} \cdot d_{G(t)}(u,v).
\]
As the algorithm returns $\delta_S(u,v)$ for a vertex pair $(u,v) \in \mathcal S(v) \times \mathcal S(v)$,
we are done.

Now consider the case that the path $P$ also has some internal vertices in $\mathcal S(t)$.
(We are still considering the case that $(u,v) \in \mathcal S(t) \times \mathcal S(t)$.)
Let $(i_1 = 1 < i_2 < i_3 < \ldots < i_q =h)$ be the indices (w.r.t. path $P$) such that
$u_{i_j} \in \mathcal S(t)$,
for all $j \in [q]$.
By the above argument, it follows that
\[
w_S(u_{i_j}, u_{i_{j + 1}}) \le (1 + \epsilon)^{3\ell(t) + 1} \cdot 
\min \{ d_{G(t_1)}(u_{i_j}, u_{i_{j + 1}}), d_{G(t_2)}(u_{i_j}, u_{i_{j + 1}})  \}.
\]
Also, the exact distance from $u$ to $v$ in $G(t)$ is 
\[
       d_{G(t)}(u,v) = w(P) = \sum_{j =1}^{h -1} \min \{ d_{G(t_1)}(u_{i_j}, u_{i_{j + 1}}), d_{G(t_2)}(u_{i_j}, u_{i_{j + 1}})  \}.
\]
Hence
\begin{align}\label{eq:sum}
    \begin{aligned}
        \delta_S(u,v) &\le (1 + \epsilon) \sum_{j =1}^{h-1} w_S(u_{i_j}, u_{i_{j + 1}})
                      \le (1 + \epsilon)^{3\ell(t) +2}\sum_{j = 1}^{h-1} \min \{ d_{G(t_1)}(u_{i_j}, u_{i_{j + 1}}), d_{G(t_2)} (u_{i_j}, u_{i_{j + 1}})\}\\
                      &= (1 + \epsilon)^{3\ell(t) +2} w(P) = (1 + \epsilon)^{3\ell(t) +2} \cdot d_{G(t)}(u,v).
    \end{aligned}
\end{align}
This proves the claim for pairs $(u,v) \in \mathcal S(t) \times \mathcal S(t)$.

We now turn our attention to pairs $(u,v) \in \mathcal B(t) \times \mathcal B(t)$.
If no vertex of $P = P(u,v)$ belongs to $\mathcal S(t)$, then either $V(P) \subseteq \mathcal V(t_1)$ or
$V(P) \subseteq \mathcal V(t_2)$. 
Assume w.l.o.g. that $V(P) \subseteq \mathcal V(t_1)$.
In addition, as $u,v \in \mathcal B(t)$, $u$ (respectively, $v$) belongs to a separator of some ancestor $t_u$ (respectively, $t_v$) of $t$ (in $T(G))$,
which is also an ancestor of $t_1$. Also, as $u,v \in \mathcal V(t_1)$, they therefore belong to $ \mathcal B(t_1)$.
It follows that
\[w(P) = d_{G(t)}(u,v) = d_{G(t_1)}(u,v),\]
and 
 \[ 
        \delta_1(u,v) \le (1 + \epsilon)^{3(\ell(t_1) + 1)} d_{G(t_1)}(u,v) 
         \le (1 + \epsilon)^{3\ell(t)} d_{G(t)}(u,v).
 \]
 By Equation~\eqref{eq:finalEstimates} , $\delta_t(u,v)$ is set in this case as $\min\{ \delta_1(u,v), \delta_{BS}(u,v)\}$.
 We conclude that 
 \begin{equation}\label{eq:distEstimateBT2}
        \delta_t(u,v) \le (1 + \epsilon)^{3\ell(t)} d_{G(t)}(u,v),
 \end{equation}
 as desired.
 
 Hence we are left with the case that $V(P) \cap \mathcal S(t) \neq \phi$.
 Let $i_L$ and $i_R$ be the leftmost and the rightmost indices of 
 vertices on the path $P$ such that $u_{i_L}, u_{i_R} \in \mathcal S(t)$.
 Observe that $u = u_1 \in \mathcal B(t)$, $u_{i_L} \in \mathcal S(t)$,
 $u_{i_R} \in \mathcal S(t)$, $u_h = v \in \mathcal B(t)$, and thus $(u_1, u_{i_L}) \in \mathcal B(t) \times \mathcal S(t)$,
 $(u_{i_L}, u_{i_R}) \in \mathcal S(t) \times \mathcal S(t)$, $(u_{i_R}, v) \in \mathcal S(t) \times \mathcal B(t)$.
Hence, all these three pairs are edges of the hopset $H$.
We have already seen that (by Equation~\eqref{eq:sum}) 
\[
          \delta_S(u_{i_L} , u_{i_R}) \le (1 + \epsilon)^{3\ell(t) + 2} d_{G(t)} (u_{i_L}, u_{i_R}),
\]
and $w_{BS}(u_{i_L}, u_{i_R})$ is set as $w_{BS}(u_{i_L}, u_{i_R}) = \delta_S(u_{i_L}, u_{i_R})$ (see Equation~\eqref{eq:weight2}).

The subpath $(u = u_1, u_2, \ldots, u_{i_L}), u \in \mathcal B(t), u_{i_L} \in \mathcal S(t)$,
$u_1, u_2, \ldots u_{i_{L -1}} \in \mathcal V(t) \setminus \mathcal S(t)$, is contained in either $\mathcal V(t_1)$ or $\mathcal V(t_2)$.
(Because to contain vertices from both these sets, it must traverse a separator vertex.)
Assuming w.l.o.g. that it is contained in $\mathcal V(t_1)$, we conclude that 
$u, u_{i_L} \in \mathcal B(t_1)$, and thus, by inductive hypothesis, the estimate 
\[
    \delta_1(u, u_{i_L}) \le (1 + \epsilon)^{3(\ell(t_1) + 1)} \cdot d_{G(t_1)}(u, u_{i_L})=
      (1 + \epsilon)^{3(\ell(t_1) + 1)} d_{G(t)} (u, u_{i_L}) = (1 + \epsilon)^{3\ell(t)} d_{G(t)}(u, u_{i_L})
\]
was computed for it while processing $t_1$.

Similarly, we have also computed a $(1 + \epsilon)^{3 \ell(t)}$-approximate estimate for the distance $d_{G(t)}(u_{i_R}, v)$ while processing the children $t_1$, $t_2$ of $t$.

By Equation~\eqref{eq:weight3}, we thus conclude that 
\[
    w_{BS}(u, u_{i_L}) \le (1 + \epsilon)^{3\ell(t)} \cdot  d_{G(t)}(u, u_{i_L}),
\]
\[
    w_{BS}(u_{i_R}, v ) \le (1 + \epsilon)^{3\ell(t)} \cdot d_{G(t)} (u_{i_R}, v).
\]
As the distance in $G_{BS}$ between $u$ and $v$ is at most 
\[
     w_{BS} (u, u_{i_L}) + w_{BS}(u_{i_L}, u_{i_R}) + w_{BS}(u_{i_R}, v),
\]
the computed distance estimate $\delta_{BS}(u,v)$ is therefore bounded by
\begin{align*}
    \begin{aligned}
          \delta_{BS}(u,v) &\le     (1 + \epsilon) \left( w_{BS}(u, u_{i_L}) + w_{BS}(u_{i_L}, u_{i_R}) + w_{BS}(u_{i_R}, v) \right)\\
         & \le  ( 1 + \epsilon) \left( (1 + \epsilon)^{3 \ell(t)} d_{G(t)}(u, u_{i_L}) + (1 + \epsilon)^{3 \ell(t) +2} d_{G(t)}(u_{i_L}, u_{i_R})   + (1 + \epsilon)^{3\ell(t)} d_{G(t)} (u_{i_R}, v)\right)\\
         & \le (1 + \epsilon)^{3(\ell(t) + 1)} \left( d_{G(t)} (u, u_{i_L}) + d_{G(t)} (u_{i_L}, u_{i_R}) +  d_{G(t)} (u_{i_R},v) \right)\\
         & = (1 + \epsilon)^{3(\ell(t) + 1)} d_{G(t)}(u,v).
   \end{aligned}
\end{align*}
\end{proof}
As $\ell(t) = O(\log n)$, we conclude that for any hopset edge $e = (u,v) \in H$, the weight that the algorithm computes for
it is a $(1 + \epsilon)^{O(\log n)}$-approximation of $d_{G(t)}(u,v)$.
\cite{Cohen93} showed that when all these distances are computed exactly, then $H$ is an exact $O(\log n)$-hopset 
of $G$. It follows that our algorithm computes a $(1 + \epsilon)^{O(\log n)}$-approximate hopset  $H$ with hopbound $O(\log n)$
of $G$. For a parameter $\epsilon' > 0$, we set $\epsilon = \frac{\epsilon'}{O(\log n)}$.
The resulting approximation of distances becomes $(1 + \epsilon')$.

As the distances in the original graph are at most $Mn$, it follows that all the edge weights that our algorithm sets 
are at most $(1 + \epsilon') Mn$ as well.
Thus, the parameter $W$ in our invocation of the algorithm of~\cite{APSPDirectedZwick} is $O(Mn)$.

Our hopset contains the same edges as the shortcut of~\cite{Cohen93} does.
Thus, the analysis of its size that shows that $|H| = O(n^{2\rho} + n \cdot \log n) $
from~\cite{Cohen93} (see Theorem $5.1$ therein) is applicable here.
To summarize:
\begin{theorem}\label{thm:approxDistances}
Consider an $n$-vertex directed weighted graph $G = (V,E)$ with non-negative edge weights.
Suppose that the aspect ratio w.r.t.  non-zero weights is at most $M$, for a parameter $M \ge 1$.
Suppose also that the underlying undirected unweighted skeleton of $G$ admits a $k^{\rho}$-recursive separator,
for a parameter $0< \rho < 1$.
Given a parameter $\epsilon > 0$, 
the algorithm, given a separator decomposition $T_G$ of the graph $G$,
computes a $(1 + \epsilon, O(\log n))$-hopset
$G' = (V, H)$ of $G$ with $O(n^{2\rho} + n \log n)$ edges in the centralized time 
\[
     T(n, \epsilon, \rho, M) = O \left (\frac{1}{\epsilon} \log^3 n (\log n + \log (M/\epsilon))   \cdot n^{\omega \rho}\right),
\]
or in parallel polylogarithmic time and within $ T(n, \epsilon, \rho, M)$ work (up to a factor of $\log^{O(1)} n$).
\end{theorem}
As was mentioned above, in the centralized setting a tree decomposition with 
near-optimal tree-width can be computed in near-linear time.
Thus, $(1 + \epsilon, O(\log n))$-hopset $G'$ with $\tilde O(n^{2\rho} + n)$
size can be computed from scratch within $\tilde O_{\epsilon}(\log M \cdot n^{\omega \rho}) + m^{1 + o(1)}$ time.
\subsection{Approximate Rectangular Distance Products}\label{sec:distanceProducts}
As was discussed in Section~\ref{sec:ApproxDistanceOutline}, after the first stage of the algorithm
(computing an approximate $(1+ \epsilon, O(\log n))$-hopset) is completed, the algorithm
proceeds to the second stage. In the second stage the algorithm iteratively computes approximate 
distance products between a rectangular matrix $B$ of dimensions $n^{\sigma} \times n$,
and the square $n \times n$-matrix $A$.
The matrix $A$ is the (\textsf{weighted}) adjacency matrix (with self-loops) of the augmented graph $\tilde G = (V, E \cup H)$
, i.e, of the original graph $G = (V,E)$ union with the hopset graph $G' = (V,H)$.
The self-loops are $0$ values on the diagonal, while non-diagonal entries $A[i,j]$ store the weight of the edge from vertex $i$ to vertex $j$ (for all $i, j \in [n], i \neq j$) in $\tilde G$ assuming that $\langle i, j\rangle \in E \cup H$.
If there is no such edge, the entry contains $\infty$.
The matrix $B = B^{(0)}$ is initialized as the rectangular submatrix of the weighted adjacency matrix
$A$ that contains $|S| = n^{\sigma}$ rows of $A$ that correspond to the set of sources $S \subseteq V$. 

After iteration $i$, $i = 0,1,\ldots, O(\log n)$, the matrix $B$ is updated by 
$B^{(i +1)} = B^{(i)} \star_{\xi}A $, where $\star_{\xi}$ stands for $(1 + \xi)$-approximate distance product.
Here $\xi > 0$ is a parameter that will be determined in the sequel.
\begin{definition}\label{def:min-plus}
Given two $n$-dimensional vectors $\vec b = (b_1, b_2, \ldots, b_n)$,
and $\vec a = (a_1, a_2,\ldots a_n)$, their \textsf{distance} (or \textsf{min-plus})
\emph{product} is
\[
       \vec b \star \vec a = \underset{1 \le i \le n}{\min } \{ b_i + a_i\}.
\]

For two matrices $B$ and $A$ of dimensions $n^{\sigma} \times n$ and $n \times n$, respectively, $ 0\le \sigma \le 1$, their \textsf{distance product} is the $n^{\sigma} \times n$ matrix $C = B \star A$, where for every 
$1 \le i \le n^{\sigma}$ and $1 \le j \le n$, $C[i,j] = \vec b_i \star \vec a_j$, where $\vec b_i$ (respectively, $\vec a_j$) is the $i$-th row (respectively, the $j$-th column) of the matrix $B$ (respectively, $A$).

Finally, for a parameter $\xi > 0$, a \emph{$(1 + \xi)$-approximate distance product}
$\vec b \star_{\xi} \vec a$ of two $n$-vectors $\vec b$ and $\vec a$ is a value $c$ such that 
$ \vec b \star  \vec a \le c \le (1 + \xi) \cdot(\vec b \star \vec a)$.
A \emph{$(1 + \xi)$-approximate distance product } $B  \star_{\xi} A$ of two matrices $B$ and $A$
as above is a matrix $C$ that satisfies for any $1 \le i \le n^{\sigma} $, $1 \le j \le n$, that 
$C[i,j]$ is a $(1 + \xi)$-approximate distance product of the vectors $\vec b_i$
and $\vec a_j$, i.e., the $i$-th row of $B$ and the $j$-th column of $A$.
\end{definition}
It is easy to see 
that $B \star A^h$, for an integer parameter $h$ (and exact distance product $\star$), provides one with
exact $S \times V$ $( h + 1)$-bounded distances in $\tilde G$. 
Since $\tilde G = G \cup G'$ and $G' = (V, H)$ is a $(1 + \epsilon)$-approximate $O(\log n)$-hopset of $G$,
for some fixed parameter $\epsilon > 0$,
it follows that for $h = O(\log n)$, the matrix $B \star A^h$ provides $(1 + \epsilon)$-approximate $S \times V$ distances in $G$.

As our algorithm replaces each exact distance product in $B \star A^h$ by a $(1 + \xi)$-approximate one, we conclude that it computes $(1 + \epsilon) \cdot (1 + \xi)^{O(\log n)}$-approximate $S \times V$ distances in $G$.
Zwick showed~\cite{APSPDirectedZwick}  that $(1 + \xi)$-approximate distance product of two matrices $B$ and $A$ as above with all entries being non-negative and maximum non-infinity entries bounded by a parameter $W$, 
can be computed in $\frac{1}{\xi} \cdot \log n \cdot \log (\frac{W}{\xi}) \cdot n^{\omega(\sigma) + o(1)} $ time, i.e, by $O(\frac{1}{\xi} \log n \cdot \log(W/\xi))$ ordinary matrix products of matrices 
with the same dimensions $n^{\sigma} \times n$ and $n \times n$, respectively.
(In fact,~\cite{APSPDirectedZwick} showed it for square matrices.
See~\cite{ElkinN22}, Section~\ref{sec:prelims}, for its extension to rectangular matrices.)

For the sake of completeness, we provide details of this argument (due to~\cite{APSPDirectedZwick, ElkinN22}) in 
Appendix~\ref{sec:Argument1}.
By Theorem~\ref{thm:distanceProductsAppdx} in Appendix~\ref{sec:Argument1}, the second stage of the algorithm computes $(1 + \epsilon)$-approximate
$S \times V$-distances within $O(\log W) \cdot \tilde O (1/\xi)\cdot n^{\omega(\sigma) + o(1)}$
centralized time (or in parallel polylogarithmic time and work bounded by this expression)\footnote{Here $\tilde O(\cdot )$ hides factors polylogarithmic in $n$}.
Combining this with Theorem~\ref{thm:approxDistances}, we derive the following theorem:
\begin{theorem}\label{thm:DistancesFinal}
Consider an $n$-vertex directed weighted graph $G = (V,E)$
with non-negative edge weights.
Suppose that the aspect ratio (w.r.t. non-zero edge weights)
is at most $W$, for a parameter $W \ge 1$, and that the underlying undirected unweighted skeleton of $G$ admits
a $k^{\rho}$-separator for a parameter $0  < \rho < 1$.
Given a parameter $\epsilon > 0$, 
  $(1 + \epsilon)$-approximate $S \times V$ distances are computed by our algorithm in 
   in polylogarithmic time and work at most $T'(n, \epsilon, \rho, W)$, where
\[
            T'(n, \epsilon, \rho, W) = O(\log W) \cdot \tilde O(1/\epsilon) \cdot (n^{\omega \rho} + n^{\omega(\sigma)}).       
\]
\end{theorem}
Theorem~\ref{thm:DistancesFinal} shows that results from Section~\ref{sec:SeparatorDAvail}
(see Theorem~\ref{thm: separator1}, Tables~\ref{tab:rhoInterval2} and~\ref{tab:comparisonRho2})
apply not only to the $S \times V$ reachability problem, but also to the $S \times V$
$(1 + \epsilon)$-approximate distance computation problem on graphs with aspect ratio 
(w.r.t. non-zero edge weights) at most polynomial in $n$ (and $\epsilon > \frac{1}{\log ^{O(1)} n}$).
In other words, for any $\sigma > \tau  = \frac{2(\omega - 2)}{\omega}$, 
(see Equation~\eqref{eq:threshImprovementSep}), there is a non-empty interval $J_{\sigma}$ such that for any $\rho \in J_{\sigma}$, the exponent of the running time of our parallel
$S \times V$ $(1 + \epsilon)$-approximate distance computation algorithm is smaller that the respective exponent 
of the state-of-the-art parallel reachability algorithm (due to~\cite{Cohen93}).

Cohen~\cite{Cohen93} also provided a (slower) parallel algorithm for {\em exact} $S \times V$ distance computation
on graphs with {\em arbitrary} edge weights that admit $k^{\rho}$-recursive separators.
This algorithm runs in polylogarithmic in $n$ time, but its work complexity is 
$\tilde O (n^{3 \rho} + n^{2\rho + \sigma})$, i.e., the exponent of its work complexity is given by
\[
       C_{Dist}(\sigma, \rho) = \max \{3\rho, 2\rho + \sigma \}.
\]
Another currently existing solution for $S \times V$ $(1 + \epsilon)$-approximate 
distance computation problem on graphs with non-negative edge weights with aspect ratio (w.r.t.  non-zero weights) 
at most $W$ (for a parameter $W$) is the \textsf{APSP} algorithm due to~\cite{APSPDirectedZwick}.
Its parallel version has polylogarithmic in $n$ running time, while its work complexity is 
$O(\log W) \cdot \tilde O(1/\epsilon) \cdot \tilde O(n^{\omega})$. (Unlike our algorithm or the algorithm of~\cite{Cohen93},
this algorithm does not require a tree decomposition to be provided as a part of the input.)
Therefore, the state-of-the-art exponent of the work complexity of polylogarithmic time parallel algorithms
for the $S \times V$ $(1 + \epsilon)$-approximate distance computation problem on graphs with non-negative edge 
weights that admit $k^{\rho}$-recursive separators (and are provided the respective tree decomposition as a part of their input) is 
\[
       C'_{Dist} (\sigma, \rho) = \min \{\omega, C_{Dist}(\sigma, \rho)\} = 
       \min \{ \omega, \max\{3\rho, 2\rho + \sigma \} \}.
\]
This exponent is inferior to the exponent $F(\sigma, \rho) = \max \{ \omega \rho, \omega(\sigma)\} $
of our algorithm in wide ranges of parameters $\sigma$ and $\rho$.
Indeed $F(\sigma, \rho) < \omega$ for all $\sigma < 1$ and $\rho < 1$.
Also, as $\omega \rho < 3\rho$, it follows that if 
$\omega(\sigma) < \max\{ 3\rho, 2\rho + \sigma\}$,
then, $F(\sigma, \rho) < C'_{Dist}(\sigma, \rho)$.
This holds for all $\rho > \min \{\frac{\omega(\sigma)}{3}, \frac{\omega(\sigma) - \sigma}{2} \}$.
Denote the right-hand-side expression by $\rho_{thr}(\sigma)$.
It is easy to verify that $\rho_{thr}(\sigma) = \frac{\omega(\sigma)}{3}$ for
$\sigma \le 0.725\ldots$, and $\rho_{thr}(\sigma) = \frac{\omega(\sigma) - \sigma}{2}$
for larger values of $\sigma$.

To summarize, for every value of $\sigma$, $0 \le \sigma < 1$, there is an interval 
$\hat J_{\sigma} = (\rho_{thr}(\sigma), 1)$ in which 
$F(\sigma, \rho) < C'_{Dist}(\sigma, \rho)$, i.e., our 
parallel algorithm outperforms the state-of-the-art solutions for this problem.
In Table~\ref{tab:hatJSigma} we provide the values of $\rho_{thr}(\sigma)$ for various values of $\sigma$, along with 
sample values of $F(\sigma, \rho)$ and $C'_{Dist}(\sigma, \rho)$ for various values of $\sigma$ and $\rho$ such that
$\rho \in \hat J_{\sigma}$.

\begin{table}[h!]
  \captionsetup{font=scriptsize}
  \begin{center}
  {\small 
    \begin{tabular}{l|c |c |c |c} 
      \textbf{$\sigma$} & $\rho_{thr}(\sigma)$ & $\rho$ & $F(\sigma, \rho)$ & $C'_{Dist}(\sigma, \rho)$\\
      \hline
      $\le \alpha$ & 2/3 & 0.7 & 2 & 2.1\\
      \rowcolor{Gray}
                   && 0.8 & 2 & $\omega$\\
       \rowcolor{Gray}
                   && 0.9 & $\omega \cdot 0.9 = 2.1344$ &$\omega$ \\ \hline
       0.4 & 0.6698 & 0.7 & $\omega (0.4) = 2.0095$ & 2.1\\
       \rowcolor{Gray}
                    && 0.8 & $\omega(0.4)$ & $\omega$\\
        \rowcolor{Gray}
                    && 0.9 & $\omega \cdot 0.9$ & $\omega$\\ \hline
        0.5 & 0.6810 & 0.7 & $\omega(0.5) = 2.0430$ & 2.1\\
        \rowcolor{Gray}
                     && 0.8 & $\omega(0.5)$ & $\omega$\\
        \rowcolor{Gray}
                     && 0.9 & $\omega \cdot 0.9$ & $\omega$\\ \hline
        0.6 & 0.6973 & 0.7 & $\omega(0.6) = 2.0926$ & 2.1\\
        \rowcolor{Gray}
                     && 0.8 & $\omega(0.6)$ & $\omega$\\
        \rowcolor{Gray}             
                     && 0.9 & $\omega \cdot 0.9$ & $\omega$\\ \hline
        \rowcolor{Gray}             
         0.7 & 0.7177 & 0.8 & $\omega(0.7) = 2.1530$ &$\omega$ \\
         \rowcolor{Gray}
                     && 0.9 & $\omega(0.7)$ & $\omega$\\ \hline
          0.8 & 0.7105 & 0.8 & $\omega(0.8) = 2.2209$ &$\omega$ \\
                     && 0.9 & $\omega(0.8)$ & $\omega$\\ \hline
          0.9 & 0.671 & 0.7 & $\omega(0.9) = 2.2942$ &2.3 \\
                     && 0.8 & $\omega(0.9)$ & $\omega$\\ 
                     && 0.9 & $\omega(0.9)$ & $\omega$\\ \hline
          1  & $\frac{\omega - 1}{2}$                                         
  \end{tabular}
  }
  \end{center}
    \caption{The table provides the values of intervals $\hat J_{\sigma} = (\rho_{thr}, 1 )$
               such that for all $\rho \in \hat J_{\sigma}$, the exponent $F(\sigma, \rho)$
               of our algorithm is better than the current state-of-the-art for the work complexity of $S\times V$
               $(1 + \epsilon)$-approximate distance computation problem on graphs that admit recursive $k^{\rho}$-separators with 
               non-negative edge weights with polynomial in $n$ aspect ratio.
               Our bounds, as well as bounds due to~\cite{Cohen93}, assume that a tree
               decomposition is provided to the algorithm as input.
               On the other hand, the algorithm of~\cite{APSPDirectedZwick} that provides
               the exponent $\omega$ works \emph{from scratch}.
               Highlighted rows indicate particularly significant gaps.
             }
    \label{tab:hatJSigma}
\end{table}

\section{Building Shortcuts/Hopsets in Specific Graph Families with Small Separators }\label{sec:Shorcuts-GraphFamilies}
In this section we analyze the running time of our algorithm on a number of graph families that admit small recursive separators.
In particular, we consider graphs with bounded \emph{genus}, graphs that \emph{exclude a fixed minor} (Section~\ref{sec:excludedMinor}), \emph{geometric} graphs (Section~\ref{sec:GeometricGraphs}), 
and finite element graphs (Section~\ref{sec:finiteElemnt}).
Recall that in Sections~\ref{sec:flatBound} and~\ref{sec:reachDistFlatBound} we have analyzed our algorithm 
on graphs that admit a flat bound of $O(n^{\rho})$ on the treewidth.

\subsection{Graphs with Bounded Genus and Graphs with Fixed Excluded Minor}\label{sec:excludedMinor}
In this section we analyze the running time of our algorithm for building shortcuts/hopsets on graphs with bounded genus $g = g(n)$ and graphs with a fixed excluded minor $K_h$, $h = h(n)$.
Recall that (see Equation~\eqref{eq:timeFromBasics}) the running time of shortcut/hopset construction is bounded by 
\begin{align*} O (\sum_{t \in T_G} |\mathcal S(t)|^{\omega}) + O(\sum_{t \in T_G} |\mathcal B(t)|^{\omega}).\end{align*}
We next sketch Cohen's analysis~\cite{Cohen93} of this expression.
Later we adapt it to the case of graphs with bounded genus and to graphs with fixed excluded  minor.
Following~\cite{Cohen93}, for every index $i =0,1,2,\ldots$, we denote by $T_i$ the subset of nodes $t \in T_G$
with
\begin{equation}\label{eq:boundedGenusCohen1}
(1 - \lambda)^{i + 1} n < |\mathcal V(t)| \le (1-\lambda)^i n.
\end{equation}
(We denote by $1/2 < \lambda < 1$ the ratio of the separator, i.e., a constant such that if $\tilde t$
is a child of $t$, then $|\mathcal V(\tilde t)| \le \lambda \cdot |\mathcal V (t)|$.)
Observe that $T_i$ is a forest, in which every tree has depth $O(1)$.
\begin{lemma}\label{lem:SeparatorBoundarySizeBounds}
Let $\rho \ge 1/2$, and assume that for any node $t \in T$, 
$|\mathcal S(t)| \le O(|\mathcal V(t)|^{\rho})$. Then,
    \begin{align}
    \sum_{t \in T_G} |\mathcal S(t)| = O(n^{\omega \rho}), \sum_{t \in T_G} |\mathcal B(t)| = O(n^{\omega \rho} \cdot \log n).
    \end{align}
\end{lemma}
A sketch of Cohen's proof is provided in Appendix~\ref{sec:ProofSketchCohen}.
(See~\cite{Cohen93} for more details.)

Consider now graphs with genus at most $g = g(n)$.
These graphs admit balanced recursive separators of size $O(\sqrt {ng})$ (see~\cite{GILBERTHT84}).
In this case we continue the recursive decomposition up until the number of vertices in a 
node becomes smaller than $c' \cdot g$, for a sufficiently large constant $c'$.
(Note that after that point the upper bound on the size of the separator is no longer meaningful.)
Denote by $\ell = \lceil \log_{\frac{1}{1 - \lambda}} \frac{n}{g} \rceil$ the maximum index such that
$T_j \neq \phi$. In the sequel, the constants hidden by the $O$-notation may depend on $\lambda$ and $\rho$,
which are also constant.

Similarly to Equation~\eqref{eq:boundedGenusCohen3-2}, for any index $j$, $0 \le j \le \ell$,
\begin{align}\label{eq:sjForBoundedGenus1}
    s_j = \sum_{t \in T_j} |S(t)| = \frac{1}{1 - \lambda} \cdot O\left(\frac{n_j}{(1 - \lambda)^j \cdot n} \cdot \sqrt{(1 - \lambda)^j \cdot n \cdot g}\right).
\end{align}
Also, 
\begin{align*}
    \begin{aligned}
    n_{j + 1} \le n_j + s_j &\le n_j \cdot  \left(1 + O_{\lambda}\left(\sqrt{\frac{g}{n}} \cdot (1 - \lambda)^{-j/2}\right)\right)\\
                             &= n_j \cdot  \left(1 + O\left(\sqrt{\frac{g}{n}} \cdot (1 - \lambda)^{-j/2}\right)\right).
    \end{aligned}
\end{align*}
Thus,

\begin{align*}
    \begin{aligned}
         n_j &\le n_0 \cdot \Pi_{i=0}^{j-1} \left( 1 + O\left(\sqrt{\frac{g}{n}} \cdot (1 - \lambda)^{-i/2}\right) \right)\\
             & \le n \cdot e^{(\sum_{i = 0}^{j-1} ((1 - \lambda)^{-1/2})^i) \sqrt{g/n}}.
    \end{aligned}
\end{align*}
As $j \le \ell = \lceil \log_{(1-\lambda)^{-1}} n/g \rceil$, we conclude that
\begin{align*}
        \sum_{i = 0}^{j -1} ((1 - \lambda)^{-1/2})^i = O\left(\sqrt{\frac{n}{g}}\right),    
\end{align*}
and thus, 
\begin{align*}
    n_j \le n \cdot e^{O(1)} = O(n),
\end{align*}
for every index $0 \le j \le \ell$.
By Inequality~\eqref{eq:sjForBoundedGenus1}, it follows that
\begin{align}\label{eq:sjForBoundedGenus2}
    \begin{aligned}
         s_j &= O\left(\frac{n_j}{(1-\lambda)^j \cdot n} \cdot \sqrt{(1 - \lambda)^j ng}\right) \\
             & = O\left(\frac{\sqrt {ng}}{(1 - \lambda)^{j/2}}\right).
    \end{aligned}
\end{align}
(The constant $c'$ such that the recursion continues until $n' < c' \cdot g$ satisfies 
$n'/2 + c \sqrt{n' g} \le \lambda n'$, where $c$ is the leading constant in the upper bound on the size of
separator for graphs with genus $g$.)

On the bottom level $\ell$ of the recursion the running time is
\begin{align}\label{eq:bottomLevelBoundedGenus}
    \sum_{ t \in T_{k}} O(g^{\omega}) \le O(g^{\omega}) \cdot \frac{n_{k}}{g} = O(n \cdot g^{\omega -1}).
\end{align}
Note that $n \cdot g^{\omega -1} \le (n \cdot g)^{\omega/2}$ in the range $g \le n$. (For $g > n$, the bound $O(\sqrt{ng})$
on the size of the separator becomes meaningless.)
Hence the right-hand-side of~\eqref{eq:bottomLevelBoundedGenus} is at most $O\left( (ng)^{\omega/2}\right)$.

For levels $0 \le i < \ell$, by Inequality~\eqref{eq:sjForBoundedGenus2},
\begin{align*}
    \begin{aligned}
        \sum_{t \in T_i } O(|\mathcal S(t)|^{\omega}) &=  
                                             \frac{O(s_i)}{\sqrt{(1 - \lambda)^i \cdot ng}} \cdot (1 - \lambda)^{\frac{i\omega}{2}} (ng)^{\frac{\omega}{2}}\\
                                            & = O\left(\frac{(ng)^{\omega/2}}{(1 - \lambda)^{i/2}}\right) \cdot (1 - \lambda)^{\frac{i}{2} (\omega - 1) } \\
                                            & = O((n g)^{\omega/2}) \cdot (1 - \lambda)^{\frac{i}{2} (\omega -2)}\\
                                            & = O((ng)^{\omega/2}) \cdot(1 - \lambda)^{i(\frac{\omega}{2} - 1)}
    \end{aligned}
\end{align*}
As $\omega > 2$, we have $(1- \lambda)^{\frac{\omega}{2} - 1} < 1$. Thus,
\begin{align*}
    \begin{aligned}
     \sum_{t \in T_G} O(|\mathcal S(t)|^{\omega}) &= \sum_{i = 0}^{\ell} \sum_{t \in T_i} O(| \mathcal S(t)|^{\omega}) \\
                                         &= O((ng)^{\omega/2}) \sum_{i = 0}^{\ell} (1 - \lambda)^{i \cdot(\frac{\omega }{2} - 1)}\\
                                         & = O((ng)^{\omega/2}).
     \end{aligned}
\end{align*}
(If $\omega = 2$, the upper bound grows by only a factor of $O(\log n)$.)
Now we analyze 
\begin{align}\label{eq:boundedGenusNew2}
      \sum_{t \in T_G} |\mathcal B(t)|^{\omega} = \sum_{j = 0}^{\ell} \sum_{t \in T_j} |\mathcal B(t)|^{\omega}
\end{align}
Fix a level $j$, $1 \le j \le \ell$, and analyze separately 
\begin{align}\label{eq:boundedGenusNew3}
        \sum_{t \in T_j} |\mathcal B(t)|^{\omega}.
\end{align}
To maximize the sum~\eqref{eq:boundedGenusNew3}, each of the $O((1-\lambda)^{-1})$ $T_0$-nodes
should contribute its entire separator to the boundary set of $O((1 - \lambda)^{-1})$ nodes of $T_j$.
Moreover, for a pair of nodes $t_0 \in T_0$, $t_j \in T_j$, such that $t_0$ is an ancestor of $t_j$,
to maximize $|\mathcal B(t_j)|$, this boundary set should contain all separator sets along the $t_0 - t_j$ path in $T_G$. The sum of sizes of separators on such a path is at most
\begin{align*}
    \begin{aligned}
            O(1-\lambda)^{-1} \cdot c \left( \sqrt{ng} + \sqrt{n (1 - \lambda)g} + \sqrt{n (1 - \lambda)^2g} + \ldots + O(g) \right) 
                     = O_{\lambda} (\sqrt {ng}) = O(\sqrt {ng}).
    \end{aligned}
\end{align*}
(Recall that $\lambda > 1/2$ is a universal constant.)
Thus, the overall contribution of the boundary sets of these 
$O((1 - \lambda)^{-1})$ nodes $t_j$ of $T_j$ (whose boundary set contains all
separator sets along the corresponding $t_0 - t_j$ path) to the sum~\eqref{eq:boundedGenusNew2} is 
\begin{align*}
    O((1 - \lambda)^{-1}) \cdot (O(\sqrt{ng}))^{\omega} = O((ng)^{\omega/2}).
\end{align*}
Similarly, each of the $O((1 - \lambda)^{-2})$ nodes of $T_1$ may contribute its entire separator set and all separator sets along
a descending path to some node of $T_j$ to the boundary of that latter node.
This provides an overall contribution of at most
\begin{align*}
    O((1 - \lambda)^{-2}) \cdot O\left(\sqrt{ng (1 - \lambda)}\right)^{\omega} = \frac{1}{1 - \lambda} \cdot O\left((ng)^{\omega/2}\right) (1 - \lambda)^{\frac{\omega}{2} -1}.
\end{align*}
Generally, for any $i$, $0 \le i < j$, 
the $O\left( (1 - \lambda)^{-( i +1)} \right)$
nodes of $T_i$ contribute (along with their descending $T_i - T_j$ paths) to boundary sets of $O\left( (1 - \lambda)^{-(i +1)} \right)$
$T_j$-nodes. Each of them may have a boundary set of size $O(\sqrt{ng (1 - \lambda)^i})$.
Their overall contribution to the sum~\eqref{eq:boundedGenusNew3} is therefore 
\begin{align*}
    \begin{aligned}
    O((1 - \lambda)^{-(i + 1)}) (ng (1 - \lambda)^i)^{\omega/2} &= \frac{1}{1 - \lambda} \cdot O(ng)^{\omega/2} (1 - \lambda)^{i(\omega/2 - 1)} \\
                                                        & = O_{\lambda}\left( (ng)^{\omega/2} \cdot (1 - \lambda)^{i(\omega/2 - 1)} \right) \\
                                                        & = O\left((ng)^{\omega/2} \cdot (1 - \lambda)^{i(\omega/2 - 1)}\right).   
    \end{aligned}
\end{align*}
Summing up over all levels $i$ we obtain (for $\omega > 2)$
\begin{align*}
    \sum_{t \in T_j} |\mathcal B(t)|^{\omega} = \sum_{i < j}O\left((ng)^{\omega/2} \cdot (1 - \lambda)^{i(\omega/2 -1)}\right) = O((ng)^{\omega/2}).
\end{align*}
and 
\begin{align*}
    \sum_{t \in T_G} |\mathcal B(t)|^{\omega} = O(\log n \cdot (ng)^{\omega/2}).
\end{align*}
(For $\omega = 2$ there is an additional factor of $O(\log n)$.)
\begin{cor}\label{cor:BoundedGenusTime}
  For parameters $\epsilon > 0$ and $W \ge 1$,
  the running time of our algorithm for building $O(\log n)$-shortcuts (respectively, $(1 + \epsilon, O(\log n ))$-hopsets) in the centralized model
  of computing for $n$-vertex $g$-genus directed graphs with non-negative edge weights and aspect ratio $W$ (w.r.t.  non-zero weights), $g = g(n)$, is
  $T(n, g) =  \tilde O((ng)^{\omega/2})$ (respectively, $ T(n, g, W, \epsilon ) = \tilde O_{\epsilon}(\log W \cdot (ng)^{\omega/2})$).
  In the \textsf{PRAM} model, the algorithm works in polylogarithmic in $n$ time using $T(n,g)$ (respectively, $T(n,g, W, \epsilon)$) processors. In this  setting, the algorithm assumes that the tree decomposition with the required properties is provided to the algorithm as a part of the input.
\end{cor}

For graphs with excluded minor $K_h$, the size of the separator is $ O(\sqrt{n \log n} \cdot h) = \tilde O(\sqrt n \cdot h)$~\cite{PlotkinRaoSmith}. In this case we continue the decomposition up until the nodes have $n' \le c' h^2$ vertices, for a sufficiently large constant $c'$. Analogous considerations lead to the bound of 
\begin{align*}
    \tilde O\left( n \cdot h^{2(\omega-1)} + (\sqrt n h )^{\omega} \right) = \tilde O\left ( (\sqrt n h)^{\omega} \right),
\end{align*}
as long as $h < \sqrt n$.
(For $h \ge \sqrt n$ the upper bound $\tilde O(\sqrt n h)$ on the size of the separator becomes meaningless.)
We get the following corollary.
\begin{cor}\label{cor:excludedMinor}
 For parameters $\epsilon > 0$ and $W \ge 1$, and a function $h = h(n)$,
    the running time of our algorithm for building shortcuts (respectively, hopsets) in $n$-vertex directed graphs
     with non-negative weights and aspect ratio $W$
    that exclude minor $K_h$ is $\tilde O\left( (\sqrt n h)^{\omega} \right)$
    (respectively, $O_{\epsilon }\left( \log W (\sqrt n h)^{\omega}\right)$).
     In the \textsf{PRAM} model the expressions above reflect the work complexity, and the time is polylogarithmic in $n$. 
    In the centralized setting a tree decomposition with 
    the properties required by the algorithm can be computed within time
    $\tilde O((\sqrt n h)^{\omega}) + m^{1 + o(1)}$,
    while in the PRAM setting it is not known if this is the case.
\end{cor}
Recall that the tree decompositions required in Corollaries~\ref{cor:BoundedGenusTime} and~\ref{cor:excludedMinor} can be computed via algorithms from~\cite{BernsteinGS21, ChuzhoyS21}.
See Section~\ref{sec:flatBound} (after Corollary~\ref{cor:flatTime}) for more details.

\subsection{Geometric Graphs}\label{sec:GeometricGraphs}
There is a long line of classical work~\cite{MillerThurston90, teng1991points, MillerVavasis91, MillerTengVavasis91, MillerTTV97, EppsteinMillerTeng95, davies2025} about sublinear separators for geometric graphs. 
Here we will focus on the \emph{$k$-nearest-neighbours} ($k$-NN) 
graphs, \emph{$k$-ply neighborhood systems} (i.e., their intersection graphs) 
and \emph{$r$-overlap graphs} of $k$-ply neighborhood systems in $\Re ^d$.
In both cases the dimension $d$ should be viewed as a (possibly very large) constant, while $k$ is a  parameter that can depend on $n$. (We do require $k = n^{1 - \xi}$ for some constant $\xi > 0$, though.
The overlap parameter $r$ is a universal constant too.)

Teng~\cite{teng1991points} devised an $O(n^2)$-time 
deterministic centralized algorithm for building balanced
recursive separators of size $O(k^{1/d} n^{1- 1/d})$ in such graphs.
Miller et al.~\cite{MillerTTV97} devised a
randomized parallel $\polylog(n)$ time and near-linear work algorithm for constructing such  separators.
Eppstein et al.~\cite{EppsteinMillerTeng95} devised a linear-time 
centralized algorithm for computing such separators, though for our purposes,
the quadratic time of the algorithm of Teng~\cite{teng1991points}
is sufficient. Alternatively, one can also use the algorithm of
Chuzhoy and Saranurak~\cite{ChuzhoyS21} for computing $O(\log^3 n)$-approximate
tree decomposition in $n^{2 + o(1)}$ deterministic time.
Existential bounds on tree-width from~\cite{teng1991points}
guarantee that the resulting tree decomposition will have tree-width $\tilde O(k^{1/d} \cdot n^{1 - 1/d})$.
\begin{definition}\label{def:k-NN}
 Given a set $\mathcal P = \{ p_1, p_2, \ldots, p_n\}$   of $n$
 points in $\Re^d$, and a parameter $k$, 
 the \emph{$k$-NN graph} of $\mathcal P$ is $NN_k(\mathcal P) = (\mathcal P, E)$  with vertex set $\mathcal P$. The edge set $E$ contains all pairs of 
 distinct points $p, p' \in \mathcal P$ for which either $p'$ is among $k$ closest points to $p$ or vice versa.
\end{definition}
\begin{definition}\label{def:k-ply}
    A \emph{$d$-dimensional (Euclidean) neighborhood system}
    $\Phi = \{B_1, B_2, \ldots, B_n  \}$ is a collection of balls in $\Re^d$.
    Let $\{p_1, p_2, \ldots, p_n \}$ denote the set of their respective centers.
     For each point $p \in \Re^d$ the \emph{ply} of $p$ in $\Phi$ is the number
    of balls $B_i$ that contain $p$.
    The system $\Phi$ is said to be \emph{$k$-ply} for a parameter $k$,
    if for any point $p \in \Re^d$, its ply in $\Phi$ is at most $k$.
\end{definition}
Another closely related graph family is intersection graphs of 
\emph{$k$-neighborhood systems}.
\begin{definition}\label{k-intersection}
    A neighborhood system $\Phi = \{B_1, B_2, \ldots, B_n \}$ centered at points $P = \{p_1, p_2, \ldots, p_n \}$
    is a \emph{$k$-neighborhood system} if the interior of each ball $B_i$ contains no more than $k$ points of $P$.
\end{definition}
It is well known (see, e.g., Lemma $3.2$ of~\cite{teng1991points}) that ply of any $k$-neighborhood system is $O_d(k)$.
(Recall that we assume throughout that $d = O(1)$.) Therefore, Teng's algorithm~\cite{teng1991points} provides a 
separator of size $O_d(k^{1/d} \cdot n^{1 - 1/d}) = O(k^{1/d} \cdot n^{1 - 1/d})$ for this graph family as well.
We note, however, that neighborhood systems with small ply need not be $k$-neighborhood systems for any bounded $k$. (Consider, e.g., a large ball that contains $n-1$ disjoint smaller balls. The ply of this system is just $2$,
 while it is not a $k$-neighborhood system for any $k < n -1$. See~\cite{teng1991points}, Section $3.4$, for further discussion.)

 It is also easy to see that $k$-NN graphs are special case of intersection graphs of $k$-neighborhood systems, which are,
 in turn, a special case of intersection graphs of $O_d(k)$-ply neighborhood systems.
 (See Section $4.3.2$ in~\cite{teng1991points}.)

 A natural extension of intersection graphs of neighborhood systems $\Phi = \{B_1, B_2, \ldots, B_n\}$ are \emph{$r$-overlap}  graphs, for a parameter $r > 0$. The \emph{$r$-overlap} graph of $\Phi$ contains an edge between (vertices that correspond to) balls $B_i$ and $B_j$ iff an $r$-inflated version of one of them intersects the other.

 Finally, in addition to Euclidean neighborhood systems, one can allow arbitrary linear normed spaces.
 In other words, the system $\Phi$ may well contain balls in an arbitrary $\ell_p$ space, and still as long as
 its ply is at most $k$, its $r$-overlap graph in $d$-dimensions (for constant $r$ and $d$) admits a separator 
 of size $O(k^{1/d} \cdot n^{1 - 1/d})$~\cite{teng1991points, MillerTengVavasis91}.
 Moreover, up to polylogarithmic factors in $n$, such a separator can be computed in 
 deterministic centralized time $n^{2 + o(1)}$~\cite{teng1991points, ChuzhoyS21}, 
 or in randomized polylogarithmic time and work $n^{2 + o(1)}$~\cite{teng1991points, ChuzhoyS21}.

 It is well known (see~\cite{teng1991points}, Lemma $3.6$) that both $k-NN$ graphs in 
$\Re^d$ and intersection graphs of $k$-ply neighborhood systems in
$\Re^d$ have degeneracy $O_d(k)$. (The leading coefficient is exponential in $d$.) 
Thus, the number of edges in these graphs is $m \le O(n \cdot k)$, and 
it can be as large as that.

For dimension $d = 2$, the size of the separator is $O(\sqrt{kn})$, 
and thus, our analysis of the centralized running time and parallel 
number of processors of our algorithms for graphs with bounded genus from
Section~\ref{sec:excludedMinor} is applicable as is to $k-NN$ graphs and 
intersection graphs of $k$-ply systems for dimension $d= 2$.

For dimension $d \ge 3$, the size of the separator is $O(k^{1/d} n^{1 - 1/d})$.
We next outline how the argument from Section~\ref{sec:excludedMinor}
extends to this scenario.

The tree decomposition is continued until the number of vertices in each node becomes at most $c' \cdot k$, for a sufficiently large constant $c'$.
(Indeed, after that point, the bound of the size of the separator is $O(k)$,
i.e., it is trivial.)
The expression for $s_j$ (see Equation~\eqref{eq:boundedGenusCohen3-2} 
in the proof of Lemma~\ref{lem:SeparatorBoundarySizeBounds} and recall that
 the separator ratio $\lambda$ is a fixed constant, $1/2 < \lambda < 1$) becomes
\begin{align}\label{eq:geometric-sj}
    \begin{aligned}
            s_j &= \frac{1}{1 - \lambda} \cdot  O\left(\frac{n_j}{(1 - \lambda)^j \cdot n} \cdot ((1 - \lambda)^j\cdot n)^{\frac{d-1}{d}} \cdot k^{1/d} \right) \\
            & = O_{\lambda}\left(\frac{n_j}{((1 - \lambda)^j \cdot n)^{1/d}}    \cdot k^{1/d} \right)
              = O\left(\frac{n_j}{((1 - \lambda)^j \cdot n)^{1/d}}    \cdot k^{1/d} \right).
    \end{aligned}
 \end{align}
 Thus 
 \begin{align*}
        \begin{aligned}
             n_{j + 1} \le n_j + s_j = n_j \left( 1 + O\left( \left(\frac{k}{n}\right)^{1/d} \right)  \left(1-\lambda\right)^{-j/d} \right),
        \end{aligned}     
 \end{align*}
 and so,
 \begin{align*}
        \begin{aligned}
             n_{j} \le n \cdot \exp \{\sum_{i=0}^{j-1} (1-\lambda)^{-i/d} \cdot O\left( \left(\frac{k}{n}\right)^{1/d} \right) \}
        \end{aligned}     
 \end{align*}
As $j \le \ell = \lceil \log_{(1 - \lambda)^{-1}} n/k \rceil$, we have 
\begin{align*}
    \sum_{i = 0}^{j -1} (1 - \lambda)^{-i/d} = O((n/k)^{1/d}),
\end{align*}
i.e., $n_j = O(n)$ (for every index $1 \le j \le \ell$).
Thus, by~\eqref{eq:geometric-sj},
\begin{align*}
    s_j = O\left(n^{1 - 1/d} k ^{1/d} \cdot \frac{1}{(1 - \lambda)^{j/d}} \right).
\end{align*}
On the bottom level of the recursion we still spend
\begin{align}\label{eq:bottomLevelGeometric}
    \sum_{t \in T_\ell} O(k^{\omega}) = O(k^{\omega}) \cdot \frac{n_{\ell}}{k} =
     O(n k ^{\omega -1)})
\end{align}
centralized time (or parallel work). It holds that
 $n k^{\omega - 1} \le (n^{1 - 1/d} k^{1/d})^{\omega} = n^{\omega - \omega/d} k^{\omega/d}$,
 for any $k \le n$.
Hence the right-hand-side of~\eqref{eq:bottomLevelGeometric} is at most $O\left(\left(n^{1 - 1/d} k^{1/d}\right)^{\omega}\right)$.

\noindent For levels $0 \le i < \ell$, 
\begin{align*}
    \begin{aligned}
         \sum_{t \in T_i} O(|\mathcal S(t)|^{\omega}) &= \frac{1}{(1 - \lambda)^{1 - 1/d}} \cdot \frac{O(s_i)}{( (1 - \lambda)^i n)^{1 - 1/d} \cdot k^{1/d}} \cdot ( ((1- \lambda)^i n)^{1-1/d}k^{1/d})^{\omega}\\         
          & = \frac{O_{\lambda}(n^{1 - 1/d} \cdot k^{1/d})/(1 - \lambda)^{i/d}} 
                   {(1 - \lambda)^{i \cdot \frac{d-1}{d} } \cdot n^{\frac{d-1}{d}} \cdot k^{1/d}}
         \cdot (1- \lambda)^{\frac{i(d-1)}{d} \omega} \cdot n^{\omega \frac{d-1}{d}} \cdot k^{\omega/d} \\
         &= O(n^{\omega \frac{d-1}{d}} \cdot k^{\omega /d}) \cdot(1- \lambda)^{i(\frac{d-1}{d} \omega - 1)}.
    \end{aligned}
\end{align*}
Thus, 
\begin{align*}
    \begin{aligned}
        \sum_{t \in T_G} O(|\mathcal S(t)|^{\omega}) = \sum_{i} \sum_{t \in T_i} O(|S(t)|^{\omega}) =
            O(n^{\omega \frac{d-1}{d}} \cdot k^{\frac{\omega}{d}}).
    \end{aligned}
\end{align*}
The analysis of $\sum_{t \in T_G} O(|\mathcal B(t)|^{\omega})$ also proceeds analogously to Section~\ref{sec:excludedMinor}.
For any level $i$, $0 \le i < \ell$, the sum of sizes of separators along a path
descending from a $T_i$-node is 
\begin{align*}
    \begin{aligned}
     O_{\lambda}(1) \cdot \left (n^{\frac{d-1}{d}} k^{\frac{1}{d}} (1 - \lambda)^{i \cdot \frac{d-1}{d}} + n^{\frac{d-1}{d}} k^{\frac{1}{d}}(1 - \lambda)^{( i + 1) \cdot \frac{d-1}{d}} 
    + \ldots + O(k)\right)  &= O_{\lambda}(n^{\frac{d-1}{d}} \cdot k^{\frac{1}{d}} (1 - \lambda)^{i \cdot \frac{d-1}{d}}) \\
    &= O\left(n^{\frac{d-1}{d}} \cdot k^{\frac{1}{d}} \cdot (1 - \lambda)^{i \cdot \frac{d-1}{d}}\right). 
    \end{aligned}
\end{align*}
(Recall that $1/2 < \lambda < 1$ is a fixed constant.)
Also, 
\begin{align*}
    \sum_{t \in T_G} |\mathcal B(t)|^{\omega} = \sum_{j = 0}^{\ell} \sum_{t \in T_j} |\mathcal B(t)|^{\omega}.
\end{align*}
Fix a level $j$, $0 \le j \le \ell$.
The contribution of $O((1 - \lambda)^{-1})$ nodes of $T_0$ to $\sum_{t \in T_j} |\mathcal B(t)|^{\omega}$
is at most $O((1 - \lambda)^{-1}) \cdot O (n^{\frac{d-1}{d}} k^{\frac{1}{d}})^{\omega} = O_{\lambda}(n^{\omega \frac{d-1}{d}} \cdot k^{\frac{\omega}{d}}) = O(n^{\omega \frac{d-1}{d}} \cdot k^{\frac{\omega}{d}})$.

The contribution of $O((1 - \lambda)^{-2})$ nodes of $T_1$ to $\sum_{t \in T_j} |\mathcal B(t)|^{\omega}$
is, similarly, at most 
\begin{align*}
    O((1 - \lambda)^{-2})\cdot O( (n (1-\lambda))^{\frac{d-1}{d}}\cdot k^{\frac{1}{d}})^{\omega} = 
      ( 1- \lambda)^{-1} \cdot O( n^{\omega \frac{d-1}{d}}\cdot k^{\frac{\omega}{d}}) \cdot (1-\lambda)^{\omega \frac{d-1}{d} -1}.
\end{align*}
Generally, for any $i \le j$, the contribution of $O((1 - \lambda)^{-(i +1)})$ nodes of $T_i$ to 
$\sum_{t \in T_j} |\mathcal B(t)|^{\omega}$ is at most
\begin{align*}
     O((1 - \lambda)^{-(i +1)}) \cdot O((n(1 -\lambda)^i)^{\frac{d-1}{d}} \cdot k^{\frac{1}{d}})^{\omega} = 
       O((1 - \lambda)^{-1}) \cdot O(n^{\omega\frac{d-1}{d}} \cdot k^{\frac{\omega}{d}}) \cdot (1- \lambda)^{i(\omega \frac{d-1}{d} - 1)}.
\end{align*}
Summing up over all all these contributions we conclude that (for $\omega > 2$)
\begin{align*}
        \begin{aligned}
     \sum_{t \in T_j} |\mathcal B(t)|^{\omega} &= (1 - \lambda)^{-1} \cdot O(n^{\omega \frac{d-1}{d}} \cdot k^{\frac{1}{d}})
                        \cdot \sum_{i \le j} (1 - \lambda)^{i(\omega \cdot \frac{d-1}{d}-1)}\\
                                                & = O_{\lambda}(n^{\omega \cdot \frac{d-1}{d}} \cdot k^{1/d}) 
                                                  = O(n^{\omega \cdot \frac{d-1}{d}} \cdot k^{1/d})
         \end{aligned}
\end{align*}
and 
\begin{align*}  
    \sum_{t \in T_G} |\mathcal B(t)|^{\omega} = O(\log{n} \cdot  n^{\omega \frac{d-1}{d}} \cdot k^{\frac{1}{d}}). 
\end{align*}
(For $\omega = 2$ the bound grows by another factor of $O(\log n)$.
  The latter, is however, anyway immaterial in this context, as we hide polylogarithmic factors by $\tilde O$-notation.)
Therefore, we obtain the following corollary:
\begin{cor}\label{cor:knn}
    Let $d = 2, 3, \ldots$, $\epsilon > 0$ and $W \ge 1$.
    Consider a directed $n$-vertex $m$-edge graph whose underlying undirected graph is
    an $n$-vertex $m$-edge $k-NN$ graph in $\Re^d$
    (or intersection graph of a $k$-ply neighborhood system or of a $k$-neighborhood system in $\Re^d$),
    for $k \le n^{1- \xi}$, for some $\xi > 0$, with non-negative edges weights and aspect ratio $W$. Our algorithm computes its $O(\log n)$-shortcut (or $(1 + \epsilon, O(\log n))$-hopset) from scratch within $\tilde O(n^{\omega \frac{d-1}{d}} \cdot k^{\frac{\omega}{d}} + n^{2 + o(1)})$ deterministic centralized time or in polylogarithmic parallel randomized time using $\tilde O(n^{\omega \frac{d-1}{d}} \cdot k^{\frac{\omega}{d}} + n^{2 + o(1)})$ processors.
     For $(1+ \epsilon)$-hopset the centralized running time/parallel work complexity needs to be multiplied by $O_{\epsilon}(\log W) \cdot \log^{O(1)} n$.
\end{cor}
Observe that in centralized setting, the separator decomposition $T_G$ can be computed using the 
algorithm of Eppstein et al.~\cite{EppsteinMillerTeng95} in $\tilde O(m)$ time or by the algorithm of~\cite{ChuzhoyS21}, 
while in parallel setting 
we use a randomized algorithm of Miller et al.~\cite{MillerTengVavasis91}
(see Theorem $2.5$ in~\cite{MillerTengVavasis91}).
In either case, the decomposition is computed within claimed resource bounds.

See also Appendix~\ref{sec:finiteElemnt} for an analysis of finite element graphs. The result is summarized in the following corollary:
\begin{cor}\label{cor:finiteElement}
    Given a directed weighted graph with non-negative edge weights and aspect ratio $W$, such that its undirected unweighted skeleton is an
    $n$-vertex $k$-finite element graph, for a parameter $k = k(n) \ge 4$, along with a planar
    embedding of its skeleton, our algorithm constructs an $O(\log n)$-shortcut for it in randomized polylogarithmic time and using $(\sqrt n k)^{\omega} + n^{1 + \epsilon} = (\sqrt n k)^{\omega}$ processors. For a $((1 + \epsilon), O(\log n))$-hopset, for a parameter $\epsilon > 0$, the centralized running time and the number of processors are multiplied by $O_{\epsilon} (\log W) \cdot \log^{O(1)} n$.
\end{cor}

\section{Reachabilities and Approximate Distances in Restricted Graph Families}\label{sec:reachShortFamilies}
 In this section we employ the algorithms for building shortcuts/hopsets in various graph families (developed in Section~\ref{sec:Shorcuts-GraphFamilies}). Based on them we devise algorithms for computing multi-source reachabilities and approximate distances in directed weighted graphs whose underlying undirected graphs are from one of these families.
 
\subsection{Reachabilities and Approximate Distances in Graphs with Bounded Genus}\label{sec:reachShortBoundedGenus}
Graphs with bounded genus $g \le n$
 have $O(n + g) = O(n)$ edges~\cite{GILBERTHT84}.
 As our algorithms require $\Omega(n^2)$ time, they do not improve upon naive shortest path 
 computation the centralized setting for this graph family.
We therefore analyze them only in the parallel setting.
Assume that we are given a tree decomposition $T_G$ in advance. (To the best of our knowledge,
there is currently no known efficient \textsf{PRAM} algorithm for computing recursive separators in graphs with bounded genus.) By Corollary~\ref{cor:BoundedGenusTime}, our $S \times V$ reachability algorithm for $|S| = n^{\sigma}$, $0 < \sigma < 1$, 
requires $(ng)^{\omega/2} + n^{\omega(\sigma)}$ work and polylogarithmic time.
(For distance computation, the work is multiplied by $O(\log W) \cdot \tilde O(1/\epsilon)\cdot \polylog(n)$.)
There are two known solutions for this problem.
First, a parallel version of Zwick's algorithm~\cite{APSPDirectedZwick}
requires $O_{\epsilon}(\log W)\cdot n^{\omega}$ work and polylogarithmic time.
(For reachabilities, it is $n^{\omega}$ work.)
This algorithm computes, however, \emph{all pairs} approximate distances/reachabilities,
and it does not assume that a separator decomposition $T_G$ is in place.

The second existing solution is Cohen's algorithm~\cite{Cohen93},
that like our algorithm, assumes that decomposition $T_G$ is in place, 
and computes $S \times V$ reachabilities/approximate distances.
Cohen's algorithm computes the same shortcut/hopset as our algorithm and then runs 
Bellman-Ford explorations from $|S| = n^{\sigma}$ sources.
The overall number of processors used in that algorithm is 
\begin{align*}
    C(n, g, \sigma) = (ng)^{\omega/2  } + O((ng + (n + g)) \cdot n^{\sigma}) =  (ng)^{\frac{\omega}{2}} + O(n^{1 + \sigma} \cdot g).
\end{align*}
(Observe that the size of the shortcut/hopset is $(O(\sqrt {ng}))^2 = O(ng)$,
and the number of edges in the input graph is $O(n + g)$.)

For our bound
\begin{align*}
    R(n, g, \sigma) = (ng)^{\frac{\omega}{2} } + n^{\omega(\sigma) },
\end{align*}
to improve upon $C(n,g,\sigma)$, we need
\begin{align*}
    (ng)^{\frac{\omega}{2}} \ll n^{1 + \sigma} \cdot g, \text{~~i.e.,~~} 
                          g \ll n^{\frac{2 + 2\sigma - \omega}{\omega -2}} = n^{\frac{2\sigma}{\omega -2} -1}.    
\end{align*}
We write $g = n^{\xi}$, for a parameter $0 < \xi < 1$.
The condition becomes $\xi < \frac{2\sigma}{\omega -2} - 1$.
For this condition to be feasible we need
\begin{align*}
 \frac{2\sigma}{\omega -2} >1, \text{~~i.e.,~~} \sigma > \frac{\omega -2}{2} \approx 0.185\ldots
\end{align*}
In addition, for $R(n, \sigma ,g)$ to be smaller than $C(n, \sigma, g)$,
we also need
\begin{align*}
    n^{1 + \sigma} \cdot g &= n^{1 + \sigma + \xi} \gg n^{\omega(\sigma)}, \text{~~i.e.,~~}
                       \xi > \omega(\sigma) - \sigma -1.
\end{align*}
For every value $\frac{\omega - 2}{2} < \sigma < 1$, we define the interval
\begin{align*}
    J'_{\sigma} = (\omega(\sigma) - \sigma -1, \frac{2 \sigma}{\omega -2} -1),
\end{align*}
such that for every value $\xi \in J'_{\sigma}$, $R(n, g, \sigma) < C(n, g, \sigma)$.
If we write $1 + \xi = 2\rho$
(i.e,. $n \cdot g = n^{1 + \xi} = n^{2\rho}$, and $n^{\rho}$ is the size
of the recursive separator), then we obtain the interval $J_{\sigma} = (\frac{\omega(\sigma) - \sigma}{2}, \frac{\sigma}{\omega -2})$.
(This is exactly the interval of values of $\rho$ for which, for a given $\sigma$,
our algorithm improves upon existing solutions on graphs with balanced recursive separator $n^{\rho}$.
See Theorem~\ref{thm: separator1} and Table~\ref{tab:rhoInterval2}.)

In the range $\frac{\omega -2}{2} < \sigma < \alpha$, we have $\omega(\sigma) = 2$.
We get 
\begin{align*}
    \omega(\sigma) - \sigma - 1 = 1 - \sigma < \xi < \frac{2}{\omega - 2} \cdot  \sigma - 1.
\end{align*}
This is feasible when 
\begin{align*}
    \sigma  > \frac{2\cdot (\omega -2)}{\omega} = 0.312\ldots
\end{align*}
This is the threshold value such that $J'_{\sigma} \neq \phi$, i.e., when there is a non-empty interval of values of the genus-exponent $\xi$ for which our algorithm improves upon the state-of-the-art solution. See Table~\ref{tab:rhoInterval2} for the values of intervals $J'_{\sigma}$ for sample values of $\sigma$.

Next, we compare some specific values of the exponent of our algorithm with the state-of-the-art exponent in the range we improve (see Table~\ref{tab:comparisonRho2}).
Recall that our exponent is 
\begin{align}\label{eq:boundedGenusExponent}
F(\sigma, \xi) = \max\{\frac{\omega(1 + \xi)}{2}, \omega(\sigma) \},
\end{align}
(because our number of processors is $(ng)^{\omega/2 + o(1)} + n^{\omega(\sigma) + o(1)}$.)
The state-of-the-art exponent is $ C'(\sigma, \xi) = \min\{\omega, C(\sigma, \xi)\}$, with 
\begin{align}\label{eq:bgCohen}
    C(\sigma, \xi) = \max\{\frac{\omega(1 + \xi)}{2}, 1 + \sigma + \xi \}.
\end{align}
(The number of processors in Cohen's Algorithm~\cite{Cohen93} is $\min \{ n^{\omega + o(1)}, (n \cdot g)^{\omega/2 + o(1)} + O((n\cdot g) \cdot n^{\sigma})\}$.)
\begin{theorem}\label{thm:approxDistBoundedGenus}
Let $\epsilon > 0$, $W \ge 1$ be parameters.
Consider an $n$-vertex directed weighted graph $G = (V,E)$ of genus $g = n^{\xi}$
and a set $S \subseteq V$ of $|S| = n^{\sigma}$ sources, with non-negative edge weights 
and aspect ratio $W$, and a tree decomposition $T_G$ for $G$ with bags of size $\tilde O(\sqrt {ng})$. Then our algorithm solves $S \times V$ reachability problem 
in parallel polylogarithmic time and $n^{F(\sigma, \xi)}$ work,
where $F(\sigma, \xi)$ is given by Equation~\eqref{eq:boundedGenusExponent}.
For $S \times V$ distance computation problem this work complexity needs to be 
multiplied by $\tilde O _{\epsilon} (\log W) \cdot \log^{O(1)} n$.

Our exponent $F(\sigma, \xi)$ improves the previous state-of-the-art exponent $C'(\sigma, \xi)$ of algorithm of~\cite{Cohen93} (which also requires $T_G$ to be provided as input)
in a wide range of parameters. Specifically, for every $\sigma > \frac{2(\omega - 2)}{\omega}$, there is a non-empty interval $J'_{\sigma} = ( \omega(\sigma) - \sigma - 1, \frac{2 \sigma}{\omega - 2} - 1)$ such that for every $\xi \in J'_{\sigma}$, our 
exponent $F(\sigma, \xi)$ is smaller than $C'(\sigma, \xi)$.
See Tables~\ref{tab:rhoInterval2} and~\ref{tab:comparisonRho2} for further details.
\end{theorem}

\subsection{Reachabilities and Approximate Distances in Graphs with Bounded Excluded Minor}\label{sec:reachDistExcludedMinor}
In this section we show that for $n$-vertex $m$-edge graphs with excluded minor $K_h$
the centralized version of our algorithm provides improved bounds for $S \times V$ reachabilities
and $(1 + \epsilon)$-approximate distance computations, in certain ranges of parameters 
$m = n^{\mu}$, $h = n^{\eta}$ and $|S| = n^{\sigma}$.
Notably, our algorithm for graphs with bounded excluded minor does not need to receive
as input a separator decomposition tree $T_G$. Rather it can compute it within its time bound even if it starts \emph{from scratch}.

It is known~\cite{Thomason82, Kostochka84} that $K_h$-minor-free graphs have degeneracy $O(h \sqrt {\log h})$.
It follows that the number of edges $m$ satisfies $m \le O(n h \sqrt {\log h})$,
i.e., (up to polylogarithmic factors) 
\begin{equation}\label{eq:excludedMinorSize1}
\mu \le 1 + \eta.
\end{equation}
There is a long line of research concerning efficient constructions of balanced separators in $K_h$-minor-free graphs~\cite{AlonSeymourThomas, PlotkinRaoSmith, KawarabayashiR10, WulffNisen11,ReedW09}.
We will use the following result due to Wulff-Nilsen~\cite{WulffNisen11}.
\begin{theorem}[\cite{WulffNisen11}, Corollary 1]
Given a vertex-weighted graph $G$ with $m$ edges and $n$ vertices, given $h \in N$, 
and a constant $\delta > 0$, there exists a deterministic algorithm with running time $O(m + h n^{3/2 + \delta})$
that either produces a $K_h$-minor or finds a separator of size $O(h \sqrt {n \log n})$.
\end{theorem}
For a constant $\eta < 1/2$, we can set $\delta = 1/2 - \eta$, and obtain an algorithm with running time $O(n^2)$ that builds a balanced separator of size $\tilde O(n^{1/2 + \eta})$ in $K_h$-minor-free graphs ($h = n^{\eta}$).
By using this algorithm recursively one obtains the separator decomposition $T_G$ within the same time of $O(n^2 + m \cdot \log n) = \tilde O(n^2)$. Alternatively, we can use here an approximation algorithm from~\cite{BernsteinGS21} or from~\cite{ChuzhoyS21} to build a tree
decomposition with essentially the same (up to a factor of $\log^{O(1)} n$) treewidth.

By Corollary~\ref{cor:excludedMinor}, the running time of our algorithm for computing $S \times V$-reachabilities on $n$-vertex $m$-edge $K_h$-minor-free graphs (with $|S| = n^{\sigma}$, $h = n^{\eta}$) is 
\begin{align*}
    n^{\omega \cdot (\eta + 1/2) + o(1)} + n^{\omega(\sigma) + o(1)}.
\end{align*}
(Note that the size of the separator is $O(n^{\rho}) = O(n^{\eta + 1/2} \cdot \sqrt{\log n})$,
i.e., $\rho = \eta + 1/2 + o(1)$.)
For $(1 + \epsilon)$-approximate distance computation, the running time is the same, up to 
a multiplicative factor of $O_{\epsilon}(\log W) \cdot \polylog n$.

Denote the exponent of this running time by
\begin{equation}\label{eq:gEtaSigma}
g(\eta, \sigma) = \max\{\omega \cdot (\eta + 1/2) , \omega(\sigma) \}.
\end{equation}
The running time of the naive algorithm for computing reachability/exact distances in these graphs is
$ \min \{\tilde O(n^{\mu + \sigma}) , n^{\omega + o(1)} \}, \text{ where } n^{\mu} = m$.
Denote the exponent of this running time by 
\begin{align*}
    N(\mu, \sigma) = \min \{\mu + \sigma, \omega \}.
\end{align*}
The following basic property of the function $\omega(\cdot)$ is provided without proof.
\begin{fact}
       If $\mu \le \omega - 1$, then for any $\sigma \in [0,1]$, we have
         \begin{align*}
              \mu + \sigma \le \omega(\sigma).
         \end{align*}    
\end{fact}
Hence, from now on we only consider the case $\mu > \omega - 1$.
Then there exists a value $\sigma'_{\mu}$ ($0 < \sigma'_{\mu} < 1$) such that
$\omega(\sigma'_{\mu}) = \mu + \sigma'_{\mu}$.

\begin{lemma}
    For all $\sigma'_{\mu} < \sigma \le 1$, $\omega(\sigma) < \mu + \sigma$.
\end{lemma}
\begin{proof}
    For $\sigma = 0$, $\omega(0) = 2 \ge \mu + 0 = \mu$.
    For $\sigma = 1$, $\omega(1) < \mu + 1$.\\
    Both functions are continuous and monotonically non-decreasing.
    Thus, there exists $\sigma'_{\mu}$  such that $\omega(\sigma'_{\mu}) = \mu + \sigma'_\mu$.
    Suppose for contradiction that there exists some $\sigma''$, such that $\sigma'_{\mu} < \sigma'' <1$ such that
    $\omega(\sigma'') \ge \mu + \sigma''$.
    But then there also exists $\sigma^{*}$, $\sigma'' < \sigma^* <1$,
    such that $\omega(\sigma^*) = \mu + \sigma^*$.
    (because $\omega(1) < \mu + 1$.)
    By convexity of the function $\omega(\cdot)$, we have $\omega(\sigma) < \mu + \sigma$
    for all $\sigma'_{\mu} < \sigma < \sigma^*$.
    However, $\sigma'_{\mu} < \sigma'' < \sigma^*$ and 
    $\omega(\sigma'') \ge \mu + \sigma''$, contradiction.
\end{proof}
In addition to the condition $\sigma > \sigma'_{\mu}$, we also require $\omega \cdot (\eta + 1/2) < \mu + \sigma$.
(As otherwise our algorithm is slower than the naive one.)
It follows that 
\[\sigma > \omega(\eta + 1/2) - \mu.\]
We also have $\eta \ge \mu - 1$, as graphs that exclude $K_{n^{\eta}}$-minor have $\tilde O(n^{1 + \eta})$ edges
~\cite{Kostochka84, Thomason82}.
Hence $1/2 > \eta \ge \mu -1 > \omega -2$, and $\omega - 1 < \mu \le 1 + \eta$, and 
\begin{align}\label{eq: muOmegaEta1}
    \sigma > \max\{\sigma'_{\mu}, \omega \cdot (\eta + 1/2) - \mu \}.
\end{align}
Denote by $\sigma_{thr}$ the right-hand-side of Inequality~\eqref{eq: muOmegaEta1}, i.e., 
\begin{align}\label{eq:sigmaThrExdMinor}
    \sigma_{thr} = \max\{\sigma'_{\mu}, \omega \cdot (\eta + 1/2) - \mu \}.
\end{align}
This range of $\sigma$ is not empty for all $\mu, \eta$ as above.
Indeed for $\eta < 1/2$, we have $\mu > \omega - 1 > \omega \cdot(\eta + 1/2) -1$.
Hence,
\begin{align}\label{eq:sigmaThrExdMinor1}
          \begin{aligned}
                 \omega \cdot (\eta + 1/2) - \mu < 1.
          \end{aligned}
\end{align}
In Tables~\ref{tab:sigmaThr1}-\ref{tab:sigmaThr4} in Appendix~\ref{sec:tablesExcludedMinor} we provide some values of $\sigma_{thr}$, $g(\eta, \sigma)$ and $N(\mu, \sigma)$ for some sample values of $\eta, \mu, \sigma$ within these ranges.
\begin{theorem}\label{thm:excludedMinorDistances}
 Suppose that we are given parameters $\epsilon > 0$, $W \ge 1$, 
 and a directed weighted $n$-vertex $m$-edge graph $G = (V,E)$, $m = n^{\mu}$,
 with non-negative edge weights and aspect ratio $W$, that excludes minor $K_h$,
 $h = n^{\eta}$, such that $\omega - 2 < \eta < 1/2$ and $\omega - 1 < \mu \le 1 + \eta$.
 Suppose that we are also given a subset $S \subseteq V$ of sources, $|S| = n^{\sigma}$.
 Suppose that $\sigma > \sigma_{thr}$, where $\sigma_{thr}$ is given by Equation~\eqref{eq:sigmaThrExdMinor}.
 Then our centralized algorithm computes $S \times V$-reachability and $(1 + \epsilon)$-approximate
 $S \times V$ distances in $\tilde O(n^{g(\eta, \sigma)})$ time,
 with $g(\eta, \sigma)$ given by~\eqref{eq:gEtaSigma}, from scratch.
 This exponent improves the state-of-the-art exponent $N(\sigma, \mu) = \min\{\mu + \sigma , \omega\}$
 for this problem in the aforementioned range of parameters.
 See Tables~\ref{tab:sigmaThr1}-\ref{tab:sigmaThr4} for more details.
\end{theorem}

The same bounds apply also to $n$-vertex $m$-edge $k$-finite element graphs.
In this case, $k = n^{\eta}$, $m = n^{\mu}$, $|S| = n^{\sigma}$.

\subsection{Reachabilities and Approximate Distances in Geometric Graphs}\label{sec:reachDistGeometric}
In this section we analyze the performance of our $S \times V$-direachability/$(1 + \epsilon)$-approximate distance computation algorithm on geometric graphs in centralized and parallel settings, and compare it 
with the state-of-the-art solutions. Specifically, we consider here $k$-NN graphs and intersection (and $r$-overlap) graphs of $k$-neighborhood systems, and more generally, $k$-ply neighborhood systems in dimension $d$.
Both $d$ and $r$ are assumed to be universal constants. See Section~\ref{sec:GeometricGraphs} for more details about these graph families.

For $d = 2$, the size of the separator is $O(k^{1/2} \cdot n^{1/2})$.
By Corollary~\ref{cor:knn}, the running time of our algorithm for building the separator decompositions 
is $\tilde O(n^{\omega/2} \cdot k^{\omega/2} + n^{2 + o(1)})$,
i.e., the running time of our $S \times V$-reachability algorithm is $O((n^{\omega/2} \cdot k^{\omega/2} + n^{\omega(\sigma)})$.
Denote $k = n^{q}$, $0 < q <1$.
The exponent $F(\sigma, q)$ of the running time of our algorithm is therefore
\begin{align}\label{eq:GeometricDistanceTime}
        F(\sigma, q) = \max \{\omega \cdot \left(\frac{1 + q}{2}\right), \omega(\sigma) \},
\end{align}
While the exponent of running time of the state-of-the-art algorithm on $m = n^{\mu}$ edge graphs is
\begin{align}\label{eq:GeometricDistanceTimeSOA}
    N(\sigma, \mu) = \min\{\omega, \mu + \sigma \}.
\end{align}
(In the parallel setting, the exponent of the number of processors is 
$C(\sigma, q) = \min \{\omega, 1 + q + \sigma \}$.)

The comparison of our exponent with the state-of-the-art 
parallel ($C(\sigma, q)$) exponent is now the same as in Section~\ref{sec:SeparatorDAvail}, except that $\rho$
has to be replaced by $\frac{1 + q}{2}$.
In other words, $q = 2\rho -1$.
See Tables~\ref{tab:rhoInterval2} and~\ref{tab:comparisonRho2} for details.
The following theorem summarizes the properties of our parallel algorithm for geometric graphs. 
\begin{theorem}\label{thm:geometricTimeDim2}
    Let $\epsilon > 0$, $W \ge 1$ be parameters.
    Given an $n$-vertex directed weighted graph $G = (V,E)$ with non-negative weights and 
    aspect ratio $W$, whose underlying undirected graphs is an $r$-overlap graph ($r = \Theta(1))$)
    of a $k$-ply neighborhood system in dimension $d = 2$, $k = n^q$, and a subset $S \subseteq V$ of $|S| = n^{\sigma}$ sources, our parallel algorithm for $S \times V$ reachability has polylogarithmic running time and work $n^{F(\sigma, q)}$, where $F(\sigma, q)$ is given by~\eqref{eq:GeometricDistanceTime}.
    For $(1 + \epsilon)$-approximate $S \times V$ distance computation the work complexity grows by a factor
    of $O_{\epsilon}(\log W)\cdot \log^{O(1)} n$.

    Our algorithm outperforms the previous state-of-the-art algorithm due to~\cite{Cohen93} in a wide range of
    parameters $\sigma$ and $q$.
    Specifically, for every $\frac{2(\omega - 2)}{2} < \sigma < 1$, there is a non-empty interval 
    $\hat J_{\sigma} = ( \omega(\sigma) - \sigma - 1, \frac{2 \sigma}{\omega -2} - 1)$,
    such that for every $q \in \hat J_{\sigma}$, our exponent $F(\sigma, q)$ is better than the state-of-the-art 
    exponent $C(\sigma, q)$ of~\cite{Cohen93}.
    Both algorithms are randomized (because the parallel algorithm of~\cite{MillerTengVavasis91, teng1991points}
    for computing the tree decomposition is randomized), and work \emph{from scratch}.
    See Tables~\ref{tab:rhoInterval2} and~\ref{tab:comparisonRho2} for more details.
\end{theorem}
Next, we compare the performance of our algorithm with previous state-of-the-art (in two dimensions) in the \emph{centralized} setting.
First note that $\omega(\sigma) < \sigma + \mu$, for every $\sigma > \sigma'_{\mu}$.
(For $\sigma'_{\mu}$ to exist, one needs $\mu > \omega -1$.)
Also, $\omega \cdot \frac{1 + q}{2} < \sigma + \mu$ for $\sigma > \omega \cdot \frac{1 + q}{2} -\mu$.

Let
\begin{align}\label{eq:sigmaThreshGeometric1}
    \sigma'_{thr} = \sigma'_{thr}(\mu, q) = \max \{\sigma'_{\mu}, \omega \cdot \frac{1 + q}{2} - \mu \}.
\end{align}
We therefore have $F(\sigma, q) < N(\sigma, q)$ as long as $\sigma > \sigma'_{thr}(\mu, q) $.
Observe that $\omega \cdot \frac{1+ q}{2} - \mu \le \omega - \mu < 1$, as long as $\mu > \omega -1$.
Thus, for $\mu > \omega -1$, $\sigma'_{thr}(\mu, q) < 1$, i.e., 
the interval $(\sigma'_{thr}(\mu, q),1)$ of values of $\sigma$ for which 
$F(\sigma, q) < N(\sigma,q)$ is not empty for all $\omega - 1 < \mu \le 1 + q $.
(Note that $k$-NN graphs and intersection and overlap graphs of $k$-ply neighborhood systems 
have $m = n^{\mu} \le O(k n) = O(n^{1 + q})$ edges, and thus $\mu \le 1 + q$.)
\begin{theorem}\label{thm:geometricTimeDim2Cetlzd}
    Let $\epsilon$, $W$, $G$, $S$, $\sigma$, $q$, $r$ be as in Theorem~\ref{thm:geometricTimeDim2}. Denote by $m = n^{\mu}$ the number of edges in $G$.
    Our centralized algorithm for $S\times V$-reachability/$(1 + \epsilon)$-approximate 
    distance computation requires $n^{F(\sigma, q)}$ deterministic time 
    (respectively, $O_{\epsilon}(\log W) \cdot n^{F(\sigma, q) }$ time), 
    and works \emph{from scratch}.
    For all $\omega - 1 < \mu \le 1 + q$ there is a value $\sigma'_{thr} = \sigma'_{thr}(\mu, q)$ given by Equation~\eqref{eq:sigmaThreshGeometric1} such that for all $\sigma \in (\sigma'_{thr}, 1)$, $F(\sigma, q) < N(\sigma, \mu) $, i.e., our algorithm outperforms the state-of-the-art solution for this problem. See Tables~\ref{tab:GeometricDim2Comp1} and~\ref{tab:GeometricDim2Comp2} for specific values of the two exponents.
\end{theorem}

\begin{table}[h!]
  \captionsetup{font=scriptsize}
  \begin{center}
  {\small 
    \begin{tabular}{l|c|c|c|c|c} 
      $q$ & $\mu$ & $\sigma'_{thr}(q,\mu)$ & $\sigma$ &  $F(\sigma, \rho)$ & $N(\sigma, \mu)$\\
      \hline
      	       0.4 & 1.4 & 0.875 & 0.9 & $\omega(0.9) = 2.2942$ & 2.3\\ \hline
               0.5 & 1.4 & 0.875 & 0.9 & $\omega(0.9)$ & 2.3\\
                   & 1.5 & 0.585 & 0.6 & $\omega(0.6) = 2.0926$ & 2.1\\
                                &&& 0.7 & $\omega(0.7) = 2.1530$ & 2.2\\
                                &&& 0.8 & $\omega(0.8) = 2.2209$ & 2.3\\
                                &&& 0.9 & $\omega(0.9)$ & $\omega$\\ \hline
             0.6 & 1.4 & 0.875 & 0.9 & $\omega(0.9)$ & 2.3\\
                 & 1.5 & 0.585 & 0.6 & $\omega(0.6)$ & 2.1 \\
                              &&& 0.7 & $\omega(0.7)$ & 2.2 \\
                              &&& 0.8 & $\omega(0.8)$ & 2.3 \\
                              &&& 0.9 & $\omega(0.9)$ & $\omega$ \\
                & 1.6 & 0.435 & 0.5 & $\omega(0.5) = 2.0430$ & 2.1 \\
                                 &&& 0.6 & $\omega(0.6)$ & 2.2 \\
                        \rowcolor{Gray}
                                  &&& 0.7 & $\omega(0.7)$ & 2.3 \\
                        \rowcolor{Gray}
                                   &&& 0.8 & $\omega(0.8)$ & $\omega$ \\
                                    &&& 0.9 & $\omega(0.9)$ & $\omega$\\ \hline
              0.7 & 1.4 & 0.875 & 0.9 & $\omega(0.9)$ & 2.3\\
                  & 1.5 & 0.585 & 0.6 & $\omega(0.6)$ & 2.1\\
                               &&& 0.7 & $\omega(0.7)$ & 2.2 \\
                               &&& 0.8 & $\omega(0.8)$ & 2.3 \\
            \rowcolor{Gray}
                                 &&& 0.9 & $\omega(0.9)$ & $\omega$ \\
                  & 1.6 & 0.435 & 0.5 & $\omega(0.5)$ & 2.1\\
                                &&& 0.6 & $\omega(0.6)$ & 2.2 \\
                        \rowcolor{Gray}
                                &&& 0.7 & $\omega(0.7)$ & 2.3 \\
                        \rowcolor{Gray}
                                &&& 0.8 & $\omega(0.8)$ & $\omega$ \\
                                &&& 0.9 & $\omega(0.9)$ & $\omega$ \\
                 & 1.7 & 0.315 & 0.4 & $\omega \cdot 0.85 = 2.0158$ & 2.1\\
                    \rowcolor{Gray}
                               &&& 0.5 & $\omega(0.5)$ & 2.2 \\
                    \rowcolor{Gray}
                               &&& 0.6 & $\omega(0.6)$ & 2.3 \\
                    \rowcolor{Gray}
                               &&& 0.7 & $\omega(0.7)$ & $\omega$ \\
                    \rowcolor{Gray}
                               &&& 0.8 & $\omega(0.8)$ & $\omega$ \\
                               &&& 0.9 & $\omega(0.9)$ & $\omega$ \\
  \end{tabular}
  }
  \end{center}
    \caption{ This table and Table~\ref{tab:GeometricDim2Comp2} provide a comparison between             the exponent                        $F(\sigma, q)$ 
              of the running time of our centralized algorithm for $S \times V$-reachability 
              (and $(1 + \epsilon)$-approximate distance computation) on $k$-NN graphs and 
              intersection (and overlap) graphs of $k$-ply neighborhood systems in two dimensions, with $k = n^q$, $|S| = n^{\sigma}$, with the exponent $N(\sigma, \mu)$ of the previous state-of-the-art algorithm.
              Highlighted rows indicate particularly significant improvements.
              Values for $q \ge 0.8$ appear in Table~\ref{tab:GeometricDim2Comp2}.
              }
    \label{tab:GeometricDim2Comp1}
\end{table}

\begin{table}[h!]
  \captionsetup{font=scriptsize}
  \begin{center}
  {\small 
    \begin{tabular}{l|c|c|c|c|c} 
      $q$ & $\mu$ & $\sigma'_{thr}(q,\mu)$ & $\sigma$ &  $F(\sigma, \rho)$ & $N(\sigma, \mu)$\\
      \hline
      	       0.8 & 1.4 & 0.875 & 0.9 & $\omega(0.9)$ & $\omega$\\ \hline
                   & 1.5 & 0.585 & 0.6 & $\omega \cdot 0.9 = 2.1344$ & 2.1\\
                                &&& 0.7 & $\omega(0.7)$ & 2.2\\
                                &&& 0.8 & $\omega(0.8)$ & 2.3\\
                                &&& 0.9 & $\omega(0.9)$ & $\omega$\\ \hline
                   &1.6 & 0.5344 & 0.6 & $\omega \cdot 0.9$ & 2.2\\
                    \rowcolor{Gray}
                                &&& 0.7 & $\omega(0.7)$ & 2.3\\
                    \rowcolor{Gray}
                                &&& 0.8 & $\omega(0.8)$ & $\omega$ \\
                                &&& 0.9 & $\omega(0.9)$ & $\omega$ \\ \hline
                   &1.7 & 0.4344 & 0.5 & $\omega \cdot 0.9$ & 2.2\\
                      \rowcolor{Gray}
                                &&& 0.6 & $\omega \cdot 0.9$ & 2.3\\
                       \rowcolor{Gray}
                                &&& 0.7 & $\omega(0.7)$ & $\omega$ \\
                        \rowcolor{Gray}
                                &&& 0.8 & $\omega(0.8)$ & $\omega$ \\
                                &&& 0.9 & $\omega(0.9)$ & $\omega$ \\ \hline
                  &1.8 & 0.3344 & 0.4 & $\omega \cdot 0.9$ & 2.2\\
                      \rowcolor{Gray}
                               &&& 0.5 & $\omega \cdot 0.9$ & 2.3\\
                      \rowcolor{Gray}
                                &&& 0.6 & $\omega \cdot 0.9$ & $\omega$\\
                       \rowcolor{Gray}
                                &&& 0.7 & $\omega(0.7)$ & $\omega$ \\
                        \rowcolor{Gray}
                                &&& 0.8 & $\omega(0.8)$ & $\omega$ \\
                                &&& 0.9 & $\omega(0.9)$ & $\omega$ \\ \hline
        0.9 & 1.4 & 0.875 & 0.9 & $\omega(0.9)$ & 2.3\\ 
            & 1.5 & 0.7530 & 0.8 & $\omega \cdot 0.95 = 2.2530$ & 2.3\\ 
                           &&& 0.9 & $\omega(0.9)$ & $\omega$ \\ \hline
            & 1.6 & 0.6530 & 0.7 & $\omega \cdot 0.95$ & 2.3\\ 
                           &&& 0.8 & $\omega \cdot 0.95$ & $\omega$ \\ 
                           &&& 0.9 & $\omega(0.9)$ & $\omega$ \\ \hline
             & 1.7 & 0.5530 & 0.6 & $\omega \cdot 0.95$ & 2.3\\ 
                           &&& 0.7 & $\omega \cdot 0.95$ & $\omega$ \\ 
                           &&& 0.8 & $\omega \cdot 0.95$ & $\omega$ \\ 
                           &&& 0.9 & $\omega(0.9)$ & $\omega$ \\ \hline
            & 1.8 & 0.4530 & 0.5 & $\omega \cdot 0.95$ & 2.3\\ 
                             &&& 0.6 & $\omega \cdot 0.95$ & $\omega$ \\ 
                            &&& 0.7 & $\omega \cdot 0.95$ & $\omega$ \\ 
                           &&& 0.8 & $\omega \cdot 0.95$ & $\omega$ \\ 
                           &&& 0.9 & $\omega(0.9)$ & $\omega$ \\ \hline
            & 1.9 & 0.3530 & 0.4 & $\omega \cdot 0.95$ & 2.3\\ 
                           &&& 0.5 & $\omega \cdot 0.95$ & $\omega$\\ 
                            &&& 0.6 & $\omega \cdot 0.95$ & $\omega$ \\ 
                            &&& 0.7 & $\omega \cdot 0.95$ & $\omega$ \\ 
                           &&& 0.8 & $\omega \cdot 0.95$ & $\omega$ \\ 
                           &&& 0.9 & $\omega(0.9)$ & $\omega$ \\ \hline            
  \end{tabular}
  }
  \end{center}
    \caption{ Continuation of Table~\ref{tab:GeometricDim2Comp1}. 
               Contains a comparison between             the exponent                        $F(\sigma, q)$ 
              of the running time of our centralized algorithm for $S \times V$-reachability 
              (and $(1 + \epsilon)$-approximate distance computation) on $k$-NN graphs and 
              intersection (and overlap) graphs of $k$-ply neighborhood systems in two dimensions, with $k = n^q$, $|S| = n^{\sigma}$, with the exponent $N(\sigma, \mu)$ of the previous state-of-the-art algorithm for $q = 0.8$ and $q = 0.9$.
              Highlighted rows indicate particularly significant improvements.               
              }
    \label{tab:GeometricDim2Comp2}
\end{table}

Similarly, for a general dimension $d \ge 2$ we have a separator of size 
$O(n^{\frac{d-1}{d}} \cdot k^{\frac{1}{d}})$, i.e., $\rho = \frac{q + (d-1)}{d}$. and thus, 
$q = d\rho - (d -1) = 1 - d(1 - \rho)$. Here the separator's exponent $\rho$ must be at least $\frac{d-1}{d}$. i.e., only those rows of Table~\ref{tab:comparisonRho2} are relevant 
for which $\rho \ge \frac{d-1}{d}$. (For dimensions $d = 2$ and $d = 3$ the conditions $\rho \ge 1/2$ and $\rho  \ge 2/3$ always 
hold for our results to be meaningful. In fact,
we even require $\rho  > \frac{\omega - 1}{2}$.
For higher dimensions $d \ge 4$, the condition $\rho \ge \frac{d-1}{d}$
becomes more stringent, and actually restricts feasibility intervals from 
Table~\ref{tab:rhoInterval2} (see Section~\ref{sec:SeparatorDAvail}).

To verify that also for any dimension $d \ge 3$, we obtain improvement in some non-trivial range of parameters, consider the densest regime of $d$-dimensional graphs with $\Theta(k \cdot n)$ edges and recursive separator of size $O( n^{\frac{d-1}{d}} \cdot k^{\frac{1}{d}}) = O(n^{\frac{d - 1 + q}{d}})$.
Then the size exponent is $\mu = 1 + q$, and the exponent of the separator is $\rho = \frac{q}{d} + \frac{d-1}{d}$.
We have 
\begin{align}\label{eq:geometricTimelarged}
        F(\sigma , q) = \max \{\omega \cdot \left(\frac{q + d-1}{d}\right) , \omega(\sigma)\}, \text{~and~}  
        N(\sigma, q) = \min \{\omega, 1 + q + \sigma\}.
\end{align}
and 
To have $\omega(\sigma) < 1 + q + \sigma$ we require $\sigma > \sigma_{ 1 + q}$.
(The threshold $\sigma_{1 + q}$ is such that $\omega(\sigma_{1 + q}) = 1 + q + \sigma_{1 + q}$.)
This value is defined as long as $q > \omega -2$.)
We also need that 
\[\omega \cdot  \frac{q + d -1}{d} < 1 + q + \sigma.\]
This holds when 
\begin{align}\label{eq:qForLarged}
    q > \frac{\omega (d -1) - d( 1 + \sigma)}{d-\omega}  = \frac{d(\omega -1) - \omega - d\sigma}{d - \omega}.
\end{align}
The right-hand-side is smaller than $1$ (and thus the range of $q$ is non-empty)
whenever $ \sigma > \omega -2$.
Hence the range of improvement is not empty for every dimension $d\ge 2$.

\begin{theorem}\label{thm:geometricGraphsTimeLarged}
Consider the $S\times V$-reachability (and $(1 + \epsilon)$-approximate distance computation) problem for $n$-vertex directed weighted graphs whose underlying unweighted undirected graph is a $k$-NN graph or an intersection (or $r$-overlap) graph of a $k$-ply neighborhood system in $d$ dimensions, for $k = n^q$, $|S| = n^{\sigma}$, $r = \Theta(1)$.
For any constant dimension $d \ge 3$, our algorithm solves this problem in $n^{F(\sigma, \rho)}$
centralized time deterministically or in parallel randomized polylogarithmic time with work as above.
(For approximate distance computation this expression needs to be multiplied by $O_{\epsilon}(\log W) \cdot \log^{O(1)} n$.) $F(\sigma, q)$ is given by Equation~\eqref{eq:geometricTimelarged}.

This exponent of centralized time (and parallel work) complexity is strictly better the previous state-of-the-art 
$N(\sigma,q)$ in the densest regime for the centralized setting and in all regimes for parallel setting as long as
$\sigma > \omega - 2$ and $q$ satisfies Inequality~\eqref{eq:qForLarged}.
The latter interval of $q$ is non-empty for all $\sigma > \omega -2$ and all $d \ge 3$.
\end{theorem}

%% file: appendix.tex
\appendix
\section{Proof Sketch of Lemma~\ref{lem:SeparatorBoundarySizeBounds}}
\label{sec:ProofSketchCohen}
\begin{proof}
Let $s_i = \sum_{t \in T_i} |\mathcal S(t)|$, and $\ell = \lceil \log_{1 -\lambda} n \rceil$.
Observe that for any index $1 \le j \le \ell$, we have
\begin{equation}\label{eq:boundedGenusCohen2}
\sum_{t \in T_j} |\mathcal B(t)| \le O\left(\sum_{i \le j} s_i \right).
\end{equation}
To see this, observe that a boundary set may contain only vertices that belong to separators of its ancestors.
Note also that each occurrence of a vertex $v \in \mathcal S(t)$ (in some node $t \in T_G)$
contributes to at most two boundary sets of descendants of $t$ on each level.
Thus, it contributes to at most $O(1)$ nodes in a subset $T_j$, for a given index $j$.

Note also that for every node $t \in T_i$ (for a fixed index $i$), we have 
$|\mathcal V(t) | \le  (1 -\lambda)^i n$, $|\mathcal S(t)| \le  O((1 - \lambda)^{i \rho} \cdot n^{\rho})$.
It follows that
\begin{equation}\label{eq:boundedGenusCohen3}
    s_i = \sum_{t \in T_i} |\mathcal S(t)| \le O\left( \sum_{t \in T_i} (1-\lambda)^{i\rho} n^{\rho}\right).
\end{equation}
For any index $j$, let $n_j = \sum_{t \in T_j} |\mathcal V(t)|$.
Note that $n_{j + 1} \le n_j + O(s_j)$, and 
\begin{equation}\label{eq:boundedGenusCohen3-2}
      s_j = \frac{1}{1 - \lambda} \cdot O\left(\frac{n_j}{(1 - \lambda)^j n} \cdot((1 - \lambda)^j \cdot n)^{\rho}\right).
\end{equation}
The constant hidden by the $O$-notation here and in the sequel may depend on $\lambda$ and $\rho$,
which are assumed to be constant as well.
It follows that 
\begin{align*}
    n_{j+1} \le n_j \cdot(1 + O\left(\left(\left(1 - \lambda \right)^j) \cdot n\right)^{\rho -1}\right).
\end{align*}
As $j \le \ell = \lfloor \log_{\frac{1}{1-\lambda}} n \rfloor$, we conclude that
\begin{align}\label{eq:boundedGenusCohen4}
    \begin{aligned}
        n_j &\le n \cdot  \Pi_{i = 0}^{j -1} (1 + O(( 1- \lambda)^i \cdot n)^{\rho -1})\\
            & \le O(n) \cdot e^{O(1) \cdot \sum_{i =0}^{j-1}((1 - \lambda)^i n)^{\rho -1}} = O(n).
   \end{aligned}
\end{align}
By~\eqref{eq:boundedGenusCohen3} and~\eqref{eq:boundedGenusCohen4} we have, 
\begin{align}\label{eq:boundedGenusCohen5}
    \begin{aligned}
        s_i &\le O \left (\sum_{t \in T_i} (1-\lambda)^{i\rho} \cdot n^{\rho} \right) 
            = O\left(\frac{n_i}{(1-\lambda)^{i + 1}n} \cdot (1-\lambda)^{i\rho} \cdot n^{\rho} \right)\\
            &= \frac{1}{1 - \lambda} \cdot \frac{O(n)}{(1-\lambda)^in} \cdot(1-\lambda)^{i\rho} \cdot n^{\rho}
            = \frac{1}{1 -\lambda} \cdot O \left(n^{\rho} \cdot \left(1 - \lambda \right)^{i(\rho -1)}\right).
    \end{aligned}
\end{align}
We also have,
\begin{align*}
    \sum_{t \in T_i} |\mathcal S(t)|^{\omega} \le \left(\frac{1}{1 - \lambda} \right)^{\rho} \cdot O\left(\frac{s_i}{((1 - \lambda)^in)^\rho} \cdot ((1-\lambda)^in)^{\rho\omega} \right).
\end{align*}
(The sum is maximized by setting its individual terms $|\mathcal S(t)|$ as large as possible.)
By Inequality~\eqref{eq:boundedGenusCohen5}, the right-hand-side is at most
\begin{align*}
            \left( \frac{1}{1 - \lambda} \right)^{\rho + 1} \cdot  O\left(\frac{n^{\rho} \cdot (1-\lambda)^{i (\rho -1)}}{(1 -\lambda)^{i\rho} \cdot n^{\rho}} \cdot ((1-\lambda)^in)^{\rho\omega} \right)
    = \left( \frac{1}{1 - \lambda} \right)^{\rho + 1} \cdot O(n^{\rho \omega} \cdot (1 - \lambda)^{i(\rho \omega -1)}).
\end{align*}
As $\rho \ge 1/2$, we have $\rho \omega -1 > 0$, and thus, $(1-\lambda)^{\rho\omega -1} < 1$.
Hence
 \begin{align}
    \begin{aligned}
        \sum_{t \in T_G} |\mathcal S(t)|^{\omega} &\le \sum_{i =0}^{\ell} \sum_{t \in T_i} |\mathcal S(t)|^{\omega}\\
                                         & \le \left( \frac{1}{1 - \lambda} \right)^{\rho + 1} \cdot \sum_{i = 0}^{\ell} O(n^{\rho \omega})[(1- \lambda)^{\rho \omega -1}]^i \\
                                         & \le O_{\lambda, \rho}(n^{\rho \omega}) \sum_{i = 0}^{\infty} [(1- \lambda)^{\rho \omega -1}]^i\\
                                         & = O_{\lambda,\rho}(n^{\rho \omega}) = O(n^{\rho \omega}).
    \end{aligned}    
 \end{align}
 
 Next, we argue that a similar bound applies for $\sum_{t \in T_G} |\mathcal B(t)|^{\omega}$ too.
For a fixed index $0 \le i \le \ell$, we focus on $\sum_{t \in T_i} |\mathcal B(t)|^{\omega}$.
Given a node $t' \in T_i$ and its descendant $t \in T_j$, the sum of sizes of separator sets $\mathcal S(\tilde t)$
of nodes along the $t' - t$ path $p(t'-t)$ in $T_G$ is 
\begin{align}\label{eq:boundedGenus1}
    \begin{aligned}
        \sum_{\tilde t \in p(t'-t)} |\mathcal S(\tilde t)| &= O(1) \cdot [n^{\rho} (1 - \lambda)^{i \rho} + n^{\rho} \cdot (1 - \lambda)^{(i + 1) \rho} + \ldots + n^{\rho}(1 - \lambda)^{j \rho}] \\
        & = O(n^{\rho}) \sum_{k = i}^{j} [(1- \lambda)^{\rho}]^{k} = O_{\lambda, \rho}(n^{\rho}) \cdot (1 - \lambda)^{i\rho}.        
    \end{aligned}
\end{align}
Fix an index $j \in \{0,1, \ldots \ell \}.$
To maximize  $\sum_{t \in T_j} |\mathcal B(t)|^{\omega}$ under the constraint given by~\eqref{eq:boundedGenusCohen2}, we need to maximize individual terms of this sum. Let $rt$ denote the root node of the decomposition tree $T_G$.
Recall that each of the at most $s_0$ vertices of $\mathcal S(rt)$ can contribute to at most two boundary sets in each level, and thus to $O(1)$ boundary sets of nodes of $T_j$.
While doing so, $\mathcal S(rt)$ can also contribute all separator sets along $rt - t$ path in $T_G$ (for some fixed node $t \in T_j$) to $\mathcal B(t)$. However, as we have seen (Inequality~\eqref{eq:boundedGenus1} with $i = 0$), its total contribution is still $O(n^{\rho})$.

This maximizes $\mathcal B(t)$, as it can contain only vertices that belong to a separator of one of the ancestors of $t$.
Therefore, the total contribution of $\mathcal S(rt)$ to $\sum_{t \in T_j } |\mathcal B(t)|^{\omega}$ is 
$O(1) \cdot O(n^{\rho})^{\omega}  = O(n^{\rho \cdot \omega})$.
(The sum $\sum_{t \in T_j} |\mathcal B(t)|^{\omega}$ is maximized when there are $O(1)$ nodes $t \in T_j$ that have in their 
boundary sets $\mathcal B(t)$ all vertices of $\mathcal S(rt)$ and all vertices from sets $\mathcal S(\tilde t)$ of nodes on $rt - t$  path.)
This argument applies not just to $rt$, but to each of the $O((1- \lambda)^{-1})$ nodes of $T_0$, increasing the overall maximum possible contribution of $T_0$ nodes to $\sum_{t \in T_j} |\mathcal B(t)|^{\omega}$ by a constant factor, i.e., it is $O(n^{\rho \omega}) \cdot \frac{1}{1- \lambda}$.

In addition, to maximize values of $\mathcal B(t)$ for other nodes $t \in T_j$ (other than those that received all vertices of $S(rt)$ or $S(t'')$ for some $t'' \in T_0$), each node $t' \in T_1$ may contribute to $O(1)$ nodes of $T_j$, and altogether all nodes of $T_1$ can contribute to $O((1 - \lambda)^{-2})$
nodes of $T_j$ (as there are $O((1- \lambda)^{-2})$ nodes of $T_1$).
Each of these $O((1 - \lambda)^{-2})$ nodes $t'$ may contribute up to a sum of sizes of separator sets along the $t'-t$ path to $\mathcal B(t)$ of a node $t \in T_j$.

The total number of vertices in such a path is still $O((n(1-\lambda))^{\rho})$ (by Inequality~\eqref{eq:boundedGenus1} with $i =1$), and thus, this is an upper bound on the size of sets $\mathcal B(t)$ created in this way.
Their overall contribution to $\sum_{t \in T_j} |\mathcal B(t)|^{\omega}$ is, therefore, 
\begin{align*}
        O((1 - \lambda) ^{-2}) \cdot O((n (1 - \lambda))^{\rho})^{\omega} = \frac{1}{1 -\lambda} \cdot O(n^{\rho \omega} \cdot( 1- \lambda)^{\rho \omega -1}).
\end{align*}
Generally by~\eqref{eq:boundedGenusCohen5}, for any $i \le j$, each of the at most
\begin{align*}
    \frac{s_i}{(n(1-\lambda)^{i + 1})^{\rho}} = \left (\frac{1}{1 - \lambda} \right)^{1 + \rho } \cdot \frac{O(n^{\rho} (1 - \lambda)^{i (\rho -1)} )}{n^{\rho} \cdot (1- \lambda)^{i\rho} } = O((1-\lambda)^{-i}) \cdot \left ( \frac{1}{1 - \lambda} \right)^{1 + \rho}.
\end{align*}
nodes of $T_i$ contributes to $O(1)$ nodes of $T_j$. 
Altogether they may contribute to $O_{\lambda, \rho}((1 - \lambda)^{-i}) = O((1 - \lambda)^{-i})$
nodes of $T_j$ whose boundary set is created in this way.
Each of these $T_i$-nodes $t'$ may contribute to $O(1)$ boundary sets of $T_j$ nodes. 
 Its contribution to each of these sets is all of its own separator vertices and all separator 
vertices along the path between $t'$ and the $T_j$-node to which it contributes.
Overall, it is still $O(n^{\rho } (1 - \lambda)^{i\rho})$ vertices in these separator sets along each such path.
Hence, overall by Inequality~\eqref{eq:boundedGenus1},
\begin{align*}
             \sum_{t \in T_j}|\mathcal B(t)|^{\omega} \le \sum_{i \le j} O((1 - \lambda)^{-i}) \cdot [O(n^{\rho} \cdot (1 - \lambda)^{i\rho})]^{\omega} = O(n^{\rho \omega}) \cdot \sum_{i \le j} (1 - \lambda)^{i (\rho \omega -1)}.
\end{align*}
As $\rho \omega -1 > 0$ and $1/2 < \lambda < 1$, we have $0 < (1 - \lambda)^{\rho \omega -1} < 1$,
and thus, 
\begin{align*}
    \sum_{i \le \ell} (1 - \lambda)^{i (\rho \omega - 1)} =  O_{\lambda, \rho}(1) =  O(1).
\end{align*}
Hence, 
\begin{align*}
    \sum_{t \in T_j} |\mathcal B(t)|^{\omega} = O(n^{\rho\omega}),
\end{align*}
and thus,
\begin{align*}
\sum_{t \in T_G} |B(t)|^{\omega} = \sum_{j = 0}^{\ell} \sum_{t \in T_j} |\mathcal B(t)|^{\omega} = O(n^{\rho \omega} \cdot \log n).
\end{align*}
\end{proof}

\section{A Comparison for General Densities }\label{app:comparison}
In this appendix we compare our algorithm
for $n$-vertex $m$-edge graphs with treewidth $O(n^{\rho})$ and
$m = n^{\mu}$, $0 < \rho < 1$, $1 < \mu \le 1 + \rho$, in the
centralized setting with the naive centralized algorithm for 
$S \times V$ reachability and approximate distance computation.
As everywhere throughout the paper, $G$ may be directed and weighted, while the structural assumptions are about the undirected unweighted
underlying graph.
This is an extension of the scenario that we analyzed in Section~\ref{sec:reachDistFlatBound} in which $\mu = 1 + \rho$.

It is easy to see that the exponent of the running time of our algorithm is still
\[ F'(\sigma, \rho) = \max \{ \omega \rho + 1 - \rho, \omega(\sigma) \}.\]
On the other hand, the exponent of the running time of the state-of-the-art (naive) centralized
solution improves now to 
\[ N(\sigma, \mu) = \min \{\omega, \mu + \sigma \}.\]
For $F'(\sigma, \rho)$ to be smaller than $N(\sigma, \mu)$, we need
\[\omega(\sigma) < \mu + \sigma.\]
As we argued in Section~\ref{sec:basicSparse}, for any $\mu > \omega - 1$, there exists a threshold value
$\sigma'_{\mu} $, $0 < \sigma'_{\mu} < 1 $, such that 
\[\omega(\sigma'_{\mu}) = \sigma'_{\mu} + \mu, \]
and for all $\sigma > \sigma'_{\mu}$ we have
\[\omega(\sigma) < \sigma + \mu.\]
We also need the inequality 
\[\mu + \sigma > \omega\rho + 1 - \rho = \rho(\omega - 1) + 1\]
to be satisfied.
This is true for 
\begin{align}\label{eq:sparseFlatTreeW1}
 \sigma > \rho(\omega - 1) - (\mu -1).
\end{align}
Observe that $\omega \rho + (1 - \rho) < \omega$, for all $0 < \rho < 1$.
Also, for all $\sigma$,  $\sigma < 1$, $\omega(\sigma) < \omega$.
Thus the condition
\begin{align}\label{eq:sparseFlatTreeW2}
    \sigma > \max\{\sigma'_{\mu}, \rho(\omega -1) - (\mu -1) \}
\end{align}
is equivalent to $F'(\sigma, \rho) < N(\sigma, \mu)$.

For the condition~\eqref{eq:sparseFlatTreeW1} to be feasible, we need its right-hand-side to be smaller than $1$, i.e., 
$\rho(\omega - 1) - (\mu -1) < 1$, i.e., $\rho(\omega -1) < \mu$.
But, in fact, for $\sigma'_{\mu}$ to exist (i.e., for condition $\omega(\sigma) < \mu + \sigma$ to be feasible),
we need a stronger condition $\omega - 1 < \mu$.
Hence the relevant range of $\mu$ is 
\begin{equation}\label{eq:sparseTreeWidthMu}
     \omega - 1 < \mu \le 1 + \rho.
\end{equation}
For each $\mu$ in the range~\eqref{eq:sparseTreeWidthMu} and each
\begin{align}\label{eq:sigmaThresh}
    \sigma  > \sigma'_{thr} = \max \{ \sigma'_{\mu}, \rho(\omega - 1) - (\mu -1) \},
\end{align}
we have $F'(\sigma, \rho) < N(\sigma, \mu)$.
Note also that $ 1 + \rho \ge \mu > \omega -1$ implies $\rho > \omega - 2$.

In Tables~\ref{tab:FlatBoundSparseComp1} and~\ref{tab:FlatBoundSparseComp2} we provide sample values of $F'(\sigma, \rho)$ and $N(\sigma, \mu)$
within the range in which $F'(\sigma, \rho) < N(\sigma, \mu)$.
\begin{theorem}\label{thm:flatBoundDistance}
 Let $\epsilon$, $W$, $G$, $\rho$, $\mu$, $S$ and $\sigma$ be as in Theorem~\ref{thm:treewidthDistance}.
 Consider the centralized setting and a general regime $\mu \le 1 + \rho$.
 Then the exponent $F'(\sigma, \rho)$ of our algorithm outperforms the exponent 
 $N(\sigma, \mu) = \min \{\omega, \mu + \sigma \}$ of the state-of-art algorithm 
 for $S \times V$ reachability and $(1 + \epsilon)$-approximate distance computation problem in a wide range of parameters $\sigma$, $\rho$ and $\mu$.
 Specifically, for any $\rho$, $\mu$ such that
 \[
         1 + \rho \ge \mu > \omega -1,
 \]
 there is a non-empty range $\sigma > \sigma'_{thr}(\rho, \mu)$ with
 $\sigma'_{thr}(\rho, \mu)$ given by Equation~\eqref{eq:sigmaThresh}, in which 
 $F'(\sigma, \rho) < N(\sigma, \rho)$.
 See Tables~\ref{tab:FlatBoundSparseComp1} and~\ref{tab:FlatBoundSparseComp2} for more details.
\end{theorem}
\begin{table}[h!]
  \captionsetup{font=scriptsize}
  \begin{center}
  {\small 
    \begin{tabular}{l|c|c|c|c|c} 
      $\rho$ & $\mu$ & $\sigma'_{thr}$ & $\sigma$ & $F'(\sigma, \rho)$ & $N(\sigma, \mu)$\\
      \hline
               0.4 & 1.4 & 0.875 & 0.9 & $\omega(0.9) = 2.2942$ & 2.3\\ \hline
      	      0.5 & $1.4$ & 0.875 & 0.9 & $ \omega(0.9)$ & 2.3 \\ \hline
                   & 1.5 & 0.57 & 0.6 & $\omega(0.6) = 2.0926$ & 2.1\\
                               &&& 0.7 & $\omega(0.7) = 2.1530 $ & 2.2\\
                               &&& 0.8 & $\omega(0.8) = 2.2209$ &  2.3\\
                               &&& 0.9 & $\omega(0.9)$ & $\omega$\\ \hline
            \rowcolor{Gray}
              0.6 & 1.4 & 0.875 & 0.9 & $\omega(0.9)$ & 2.3 \\
                  & 1.5 & 0.57 & 0.6 & $\omega(0.6)$ & 2.1\\
                             &&& 0.7 & $\omega(0.7) $ & 2.2\\
                             &&& 0.8 & $\omega(0.8) $ & 2.3\\
                             &&& 0.9 & $\omega(0.9) $ & $\omega$\\ \hline
                & 1.6 & 0.45 & 0.5 & $\omega(0.5) = 2.0430$ & 2.1\\ 
                             &&&0.6 & $\omega(0.6)$ & 2.2\\
                        \rowcolor{Gray}
                             &&&0.7 & $\omega(0.7)$ & 2.3\\
                        \rowcolor{Gray}
                             &&&0.8 & $\omega(0.8)$ & $\omega$\\
                            &&&0.9 & $\omega(0.9)$ & $\omega$\\ \hline
              0.7 & 1.4 & 0.875 & 0.9 & $\omega(0.9)$ & 2.3 \\ \hline
                 & 1.5 & 0.57 & 0.6 & $\omega(0.6)$ & 2.1 \\
                            &&& 0.7 & $\omega(0.7)$  & 2.2 \\
                             &&& 0.8 & $\omega(0.8)$ & 2.3 \\
                             &&& 0.9 & $\omega(0.9)$ & $\omega$ \\ \hline 
                 & 1.6  & 0.46 & 0.5 & $\omega(0.5) = 2.0428$ & 2.1\\
                              &&& 0.6 & $\omega(0.6)$  & 2.2\\   
                     \rowcolor{Gray}
                               &&& 0.7 & $\omega(0.7)$ & 2.3\\
                    \rowcolor{Gray}
                               &&& 0.8 & $\omega(0.8)$ & $\omega$\\
                                &&& 0.9 & $\omega(0.9) $ & $\omega$\\ \hline
                 & 1.7  & 0.3 & 0.4 & $\omega(0.4) = 2.0095$ & 2.1\\
                   \rowcolor{Gray}
                             &&& 0.5 & $\omega(0.5)$  & 2.2\\ 
                    \rowcolor{Gray}
                              &&& 0.6 & $\omega(0.6)$ & 2.3\\
                    \rowcolor{Gray}
                              &&& 0.7 & $\omega(0.7)$ & $\omega$\\
                    \rowcolor{Gray}
                                &&& 0.8 & $\omega(0.8)$ & $\omega$\\
                               &&& 0.9 & $\omega(0.9) $ & $\omega$\\ \hline

  \end{tabular}
  }
  \end{center}
    \caption{The table provides a comparison between the exponent $F'(\sigma, \rho)$ of our algorithm for $S \times V$-reachability/approximate distance computation  ($|S| = n^{\sigma}$), 
              with the current (naive) state-of-the-art bound $N(\sigma, \mu)$ on 
              $n$-vertex graphs with treewidth $n^{\rho}$ with $m = n^{\mu}$ edges. 
               The comparison for $\rho = 0.8$ and $\rho = 0.9$ is provided in Table~\ref{tab:FlatBoundSparseComp2}.
              }
    \label{tab:FlatBoundSparseComp1}
\end{table}

\begin{table}[h!]
  \captionsetup{font=scriptsize}
  \begin{center}
  {\small 
    \begin{tabular}{l|c|c|c|c|c} 
      $\rho$ & $\mu$ & $\sigma'_{thr}$ & $\sigma$ & $F'(\sigma, \rho)$ & $N(\sigma, \mu)$\\
      \hline
       0.8 & 1.4 & 0.875 & 0.9 & $\omega(0.9)$ & 2.3 \\ \hline
                & 1.5 & 0.597 & 0.6 & $\omega \cdot  \rho +  (1 - \rho) = 2.0972$ & 2.1\\
                             &&& 0.7 & $\omega(0.7)$ & 2.2\\
                             &&& 0.8 & $\omega(0.8)$ & 2.3\\
                              &&& 0.9 & $\omega(0.9)$ & $\omega$\\ \hline
                & 1.6 & 0.45 & 0.5 & $\omega \cdot  \rho +  (1 - \rho)$ & 2.1\\
                             &&& 0.6 & $\omega \cdot \rho + (1 - \rho)$ & 2.2\\
                       \rowcolor{Gray}
                             &&& 0.7 & $\omega(0.7)$ & 2.3\\
                    \rowcolor{Gray}
                              &&& 0.8 & $\omega(0.8)$ & $\omega$\\
                              &&& 0.9 & $\omega(0.9)$ & $\omega$\\ \hline
                 & 1.7 & 0.397 & 0.4 & $\omega \cdot  \rho +  (1 - \rho)$ & 2.1\\
                   \rowcolor{Gray}
                             &&& 0.5 & $\omega \cdot  \rho +  (1 - \rho)$ & 2.2\\
                    \rowcolor{Gray}
                             &&& 0.6 & $\omega \cdot  \rho +  (1 - \rho)$ & 2.3\\
                    \rowcolor{Gray}
                              &&& 0.7 & $\omega(0.7)$ & $\omega$\\
                    \rowcolor{Gray}
                              &&& 0.8 & $\omega(0.8)$ & $\omega$\\
                               &&& 0.9 & $\omega(0.9)$ & $\omega$\\ \hline
                & 1.8 & 0.297 & 0.3 & $\omega \cdot  \rho +  (1 - \rho)$ & 2.1\\
                \rowcolor{Gray}
                             &&& 0.4 & $\omega \cdot  \rho +  (1 - \rho)$ & 2.2\\
                \rowcolor{Gray}
                             &&& 0.5 & $\omega \cdot  \rho +  (1 - \rho)$ & 2.3\\
                \rowcolor{Gray}
                              &&& 0.6 & $\omega \cdot  \rho +  (1 - \rho)$ & $\omega$\\
                \rowcolor{Gray}
                              &&& 0.7 & $\omega(0.7)$ & $\omega$\\
                \rowcolor{Gray}
                               &&& 0.8 & $\omega(0.8)$ & $\omega$\\ 
                               &&& 0.9 & $\omega(0.9)$ & $\omega$\\ \hline
      	       0.9 & $1.4$ & 0.875 & 0.9 & $ \omega(0.9)$ & 2.3 \\ \hline
                  & 1.5 & 0.734 & 0.8 & $\omega\rho + ( 1- \rho) = 2.2344$ & 2.3\\
                              &&& 0.9 & $\omega(0.9)$ & $\omega$\\ \hline
                  & 1.6 & 0.634 & 0.7 & $\omega\rho + ( 1- \rho)$ & 2.3\\
                        \rowcolor{Gray}
                               &&& 0.8 & $\omega\rho + ( 1- \rho)$ & $\omega$\\ 
                               &&& 0.9 & $\omega(0.9)$ & $\omega$\\ \hline
                  & 1.7 & 0.534 & 0.6 & $\omega\rho + ( 1- \rho)$ & 2.3\\
                        \rowcolor{Gray}
                               &&& 0.7 & $\omega\rho + ( 1- \rho)$ & $\omega$\\ 
                         \rowcolor{Gray}
                               &&& 0.8 & $\omega\rho + ( 1- \rho)$ & $\omega$\\ 
                               &&& 0.9 & $\omega(0.9)$ & $\omega$\\ \hline
                 & 1.8 & 0.434 & 0.5 & $\omega\rho + ( 1- \rho)$ & 2.3\\
                       \rowcolor{Gray}
                               &&& 0.6 & $\omega\rho + ( 1- \rho)$ & $\omega$\\ 
                        \rowcolor{Gray}
                               &&& 0.7 & $\omega\rho + ( 1- \rho)$ & $\omega$\\ 
                         \rowcolor{Gray}
                               &&& 0.8 & $\omega\rho + ( 1- \rho)$  & $\omega$\\ 
                               &&& 0.9 & $\omega(0.9)$ & $\omega$\\ \hline
                  & 1.9 & 0.334 & 0.4 & $\omega\rho + ( 1- \rho)$ & 2.3\\
                       \rowcolor{Gray}
                               &&& 0.5 & $\omega\rho + ( 1- \rho)$ & $\omega$\\ 
                        \rowcolor{Gray}
                               &&& 0.6 & $\omega\rho + ( 1- \rho)$ & $\omega$\\ 
                         \rowcolor{Gray}
                               &&& 0.7 & $\omega\rho + ( 1- \rho)$  & $\omega$\\ 
                         \rowcolor{Gray}
                               &&& 0.8 & $\omega\rho + ( 1- \rho)$ & $\omega$\\ 
                               &&& 0.9 & $\omega(0.9)$ & $\omega$\\ \hline                                   
  \end{tabular}
  }
  \end{center}
    \caption{ Continuation of Table~\ref{tab:FlatBoundSparseComp1} for $\rho = 0.8$ and $\rho = 0.9$.
              The table provides a comparison between the exponent $F'(\sigma, \rho)$ of our algorithm 
              for $S \times V$-reachability/approximate distance computation; $|S| = n^{\sigma}$, 
              with the current (naive) state-of-the-art bound $N(\sigma, \mu)$ on 
              $n$-vertex graphs with treewidth $n^{\rho}$ with $m = n^{\mu}$ edges.              
              }
    \label{tab:FlatBoundSparseComp2}
\end{table}

\section{Approximate Distances using Approximate Distance Products}\label{sec:Argument1}
Recall that our algorithm from Section~\ref{sec:distanceProducts} is given a graph 
$G = (V,E)$, a subset $S \subseteq V$ of sources, $|S| = n^{\sigma}$, $n = |V|$
and a $(1 + \epsilon, O(\log n))$-hopset $G' = (V,H)$ for $G$.
It creates the augmented graph $\tilde G = (V, E \cup H)$ and an adjacency matrix $A$ for $\tilde G$
(with self-loops). The self-loops are $0$ values on its diagonal.
The matrix $B = B^{(0)}$ is initialized as a rectangular sub-matrix of $A$ that contains only the rows
that correspond to the set of sources $S$.

After iteration $i$, $i = 0,1,\ldots, O(\log n)$, the matrix $B$ is updated by $B^{(i + 1)} = B^{(i)} \star_{\xi} A$, where $\star_{\xi}$ stands for a $(1 + \xi)$-approximate distance product.
(Here $\xi > 0$ is a parameter.)
Below we argue that the algorithm computes a $( 1 + \epsilon) \cdot (1 + \xi)^{O(\log n)}$-approximation of all
$S \times V$ distances (in $G$) efficiently.

One can assume that our original edge weights (in the augmented graph $\tilde G$)
are such that the minimal non-zero weight is equal to $1$.
Denote by $W$ the maximum edge weight.
Furthermore, by multiplying all the edge weights by $1/\xi$, and rounding them up, one obtains integral edge weights.
Observe that the resulting edge weights are in $\{\lceil 1/\xi, \rceil, \ldots, M \}$, $M =      
\left\lceil \frac{W}{\xi} \right\rceil$.
These edge weights are, up to a multiplicative factor of $1 + \xi$, 
the same as inflated by the factor $1/\xi$ original edge weights.
Thus, by computing distances in the resulting integer-weighted graph and scaling them down 
by a factor of $\xi$, one obtains a $(1 + \xi)$-approximation of original distances.

Hence we assume hereafter that all edge weights are non-negative integers, upper bounded by $M$.
Hence, when computing the first approximate distance product $B^{(1)} = B^{(0)} \star_{\xi } A$,
all non-infinity entries in $B^{(0)}$ and $A$ are integers. 
Moreover, we will argue that all entries in the product matrix $B^{(1)}$ are integers as well and this is also the case in general, i.e., for $B^{(i)}$, for every $i = 0,1, \ldots, h + 1$. 
Hence from now on we focus on computing an approximate distance product $B^{( i + 1)} = B^{(i)} \star_{\xi} A$
for integer-valued matrices $B^{(i)}$ and $A$ (except for some entries equal to infinity),
with all finite values in $\{0, 1, \ldots, M' \}$, with $M' = O(Mn)$.
(Non-infinity entries of $B^{(i)}$ are bounded by $O(M n)$ because they provide a 
$(1 + \xi)^{i +1}$-approximation to $(i + 1)$-restricted distances in $G$, which are at most 
$M n$. Hence, these entries are at most $(1 + \xi)^{O(\log n)} M n$.
We will set $\xi = O(\frac{\xi'}{\log n})$, for a parameter $0 < \xi' < 1$, to ensure that 
$(1 + \xi)^{O(\log n)} = O(1)$.)

The procedure relies on a simple (albeit more time-consuming) routine for computing an 
exact distance product $C = B' \star A'$ of two matrices as above (with dimensions $n^{\sigma} \times n$ and $n \times n$, respectively).
The latter procedure is due to Yuval~\cite{Youval76}, see also~\cite{AGM91, APSPDirectedZwick}.
This procedure multiplies instead two matrices $\bar B$ and $\bar A$ 
of the same dimensions as $B'$ and $A'$ respectively, using \emph{ordinary} (i.e., $(+, \cdot)$) matrix product.
These matrices, however, contain entries of larger magnitude.
Specifically, each such entry may be at most $(n + 1)^{M'}$.
As a result, the product requires $n^{\omega(\sigma) + o(1)}$
operations on $O(M' \cdot \log n)$-bit integers. Hence, the overall time
required by this computation is $\tilde O(M') \cdot n^{\omega(\sigma)}$.
\begin{cor}(\cite{Youval76,AGM91,APSPDirectedZwick} )~\label{cor:minPlusProduct}
An exact distance product between two integer-valued matrices of dimensions $n^{\sigma} \times n$
and $n \times n$ with non-negative entries upper-bounded by at most $M'$ can be computed in $\tilde O(M') \cdot n^{\omega(\sigma) + o(1)}$
time.
\end{cor}
In \textsf{PRAM} setting, a naive implementation of this routine requires $\tilde O(M')$ time
and $\tilde O(M') \cdot n^{\omega(\sigma) + o(1)}$ work.
\cite{APSPDirectedZwick} points out however, that one can instead use Chinese 
Remainders Theorem (CRT) to compute the product $B' \cdot A'$ more efficiently.
Specifically, one selects $O(M')$ different primes with $O(\log n)$ bits in each,
and computes $B' \cdot A'$ modulo these primes. 
Then the result $B' \cdot A'$ can be reconstructed (entry-wise) using CRT. This results in  polylogarithmic time and work $\tilde O(M') \cdot n^{\omega(\sigma) + o(1)}$.

Next, we use this routine for describing the approximate distance 
product $B' \star_{\xi} A'$ (due to~\cite{APSPDirectedZwick}).
We assume that both $M'$ and $R = \frac{1}{\xi}$
are integer powers of $2$.
If this is not the case, one can set $M''$ (respectively, $R''$)
to be the closest integer power of $2$ greater than $M'$ (respectively, $R$),
and set $\xi'' = \frac{1}{R''}$.
Computing the product with these values results in estimates that are even better 
than desired, while the time and work complexities increase only
by a constant factor.
(The dependence of the complexity on $M'$ is at most logarithmic, 
while the dependence on $\xi$ is at most $O(\frac{1}{\xi} \log \frac{1}{\xi})$.)
Hence the assumption that $M'$ and $R$ are integer powers of $2$ is without 
loss of generality.
Furthermore, we can assume that $M' > R$ as otherwise exact distance product provides the desired time and work 
complexities. (Recall that we aim at complexity roughly $O(1/\xi) \cdot n^{\omega(\sigma) + o(1)}$.)
We write $M' = 2^m, R = 2^r, m > r$.

The algorithm runs for $m - r+1$ iterations. 
On each iteration $k = 0,1,\ldots, m-r$,
the algorithm defines two matrices $B'_k$ and $A'_k$
in the following way.
First, it replaces all entries larger than $R \cdot 2^k = 2^{r + k}$
by infinity. Second, it scales down all remaining entries by dividing them by $2^k$
and rounding them up.
Observe that all entries in $B'_k$ and $A'_k$ are in $\{0,1, \ldots, R \} \cup \{\infty \}$.
We now compute an exact distance product $C'_k = B'_k \star A'_k$
(for every $k \in \{0,1, \dots, m-r\}$),
and scale it back by multiplying every entry of it by $2^k$.
Finally, the matrix $C'$ is defined by setting each entry $C'[i,j]$ to be equal to 
\begin{equation}\label{eq:Cij}
   C'[i,j] = \min \{2^k C'_k[i,j]~|~ k \in \{0,1, \dots, m-r \} \}. 
\end{equation}

\begin{lemma}\label{lem:integerEntriesCij}
$C'$ is a $(1 + O(\xi))$-approximate distance product of $B'$ and $A'$,
and moreover, all its non-infinity entries are integers.
\end{lemma}
\begin{proof}
Denote by $C = B' \star A'$ the (exact) distance product of $B'$ and $A'$.
For every $k \in \{0, 1, \ldots, m-r \}$, all entries of $B'_k$ and $A'_k$ belong to 
$\{0,1, \dots, R\} \cup \{\infty\}$.
 Thus, all entries of $C'_k = B'_k \star A'_k$ are integers as well, and this is, therefore, 
 also the case with each entry $C'[i,j]$ of $C'$ (see Equation~\eqref{eq:Cij}).

 Also, for every $k \in \{0,1, \dots, m-r \}$, and every entry $C'_{k}[i,j]$
 satisfies 
 \[
        C'_{k}[i,j] = \underset{1 \le t \le n}{\min }  \{B'_k[i,t] + A'_k[t,j] \} \ge \underset{1 \le t \le n}{\min } \{(B'[i,t] + A'[t,j])/2^k\}.
 \]

(The inequality holds because each entry $B'_k[i,t]$ is obtained from $B'[i,t]$ by either replacing it by $\infty$, or by dividing it by $2^k$ and rounding it up.) Hence we get that
\begin{equation}\label{eq: CijApprox}
    2^k C'_k[i,j] \ge \min_{1 \le t \le n}\{ B'[i,t] + A'[t,j] \} = C[i,j],
    \end{equation}
and therefore, 
\[
    C'[i,j] = \underset{k}{\min} \{2^k C'_k[i,j] \} \ge C[i,j].
\]
In the opposite direction, for a fixed entry $C[i,j]$ of the product matrix $C$, 
consider the index $t$ such that
\begin{equation}\label{eq:Cij2}
      C[i,j] = B'[i,t] + A'[t,j] = \underset{1 \le q \le n}{\min} \{B'[i,q] + A'[q,j] \}.
  \end{equation}
  Suppose that $B'[i,t] \ge A'[t,j]$. (Otherwise, the argument is symmetric.)

  If $B'[i,t] \le R$, then (for $k = 0$) we have $B'_0[i,t] = B'[i,t]/2^0  = B'[i,t] $,
  and $A'_0[t,j] =  A'[t,j]/2^0 = A'[t,j]$.
  For any index $q \in [n]$, $B'_0[i,q] \ge B'[i,q]$ and $A'_0[q,j] \ge A'[q,j]$.
  (Strict inequality may occur if, e.g., $B'[i,q] > R$, and then $B'_0[i,q] = \infty$.)
  Hence (by Equation~\eqref{eq:Cij2}),
  \begin{align}\label{eq:Cij3}
      \begin{aligned}
          B'_0[i,q] + A'_0[q,j] \ge B'[i,q] + A'[q,j] &\ge B'[i,t] + A'[t,j] \\
                                                          & = \min_q \{B'[i,q] + A'[q,j] \} \\
                                                          & = B'_0[i,t] + A'_0[t,j].
        \end{aligned}
  \end{align}
Thus, $C'_0[i,j] = B'_0[i,t] + A'_0[t,j] = B'[i,t]  + A'[t,j] = C[i,j]$.

Otherwise, $B'[i, t] > R$. Let $k \in \{1,2, \ldots, m-r \}$ be the index such that $R \cdot 2^{k-1} < B'[i,t] \le R \cdot 2^k$.
Then
\[
     B'_k[i,t] = \bigg \lceil \frac{B'[i,t]}{2^k} \bigg \rceil  \text{~and~}  A'[t,j] = \bigg \lceil \frac{A'[t,j]}{2^k} \bigg \rceil.
\]
Hence
\begin{align}\label{eq:Cij4}
    \begin{aligned}
                      2^k C'_k[i,j] &= \underset{ 1\le q \le n}{\min} \{2^k(B'_k[i,q] + A'_k[q,j]) \}\\
                                    & \le 2^k\left (B'_k[i,t] + A'_k[t,j] \right) \\
                                    & \le 2^k\left( \frac{B'[i,t]}{2^k} + \frac{A'[t,j]}{2^k} + 2\right)\\
                                    &= C[i,j] + 2^{k+1}.       
    \end{aligned}
\end{align}
As $C[i,j] = B'[i,t] + A'[t,j] \ge B'[i,t] > 2^{k-1}R$, Equation~\eqref{eq:Cij4} implies that
\begin{align*}
    \begin{aligned}
            2^k C'_k[i,j] \le C[i,j] + 2^{k + 1} &= C[i,j] \left(1 + \frac{2^{k + 1}}{C[i,j]}\right) \\
                                       &< C[i,j] \left( 1+ \frac{2^{k + 1}}{2^{k-1}R}\right) = C[i,j]\left( 1+ \frac{4}{R}\right)\\
                                       & = C[i,j](1 + 4\xi).    
    \end{aligned}
\end{align*}
Hence 
\[C'[i,j] = \min_{k'} \{2^{k'} C'_k[i,j]\} \le 2^k C'_k[i,j] \le  C[i,j](1 + 4\xi), \]
proving the lemma.
\end{proof}
The algorithm for computing $B' \star_{\xi} A'$ involves at most $m = O(\log M') = O(\log M + \log n)$ computations of exact distance products with all non-infinity entries bounded by $R = O (1/\xi)$.
By Corollary~\ref{cor:minPlusProduct}, each such computation requires $\tilde O(R) \cdot n^{\omega(\sigma) + o(1)}$ 
centralized time, or polylogarithmic \textsf{PRAM} time and work $\tilde O(R) \cdot n^{\omega(\sigma) + o(1)}$.
Hence the overall time of computing $h = O(\log n)$ such distance products 
\[
     B^{(1)} = B^{(0)} \star_{\xi} A,  B^{(2)} = B^{(1)} \star_{\xi} A, \ldots  B^{(h)} = B^{(h-1)} \star_{\xi} A, 
\]
is $O(\log M') \cdot \tilde O(1/\xi) \cdot n^{\omega(\sigma) + o(1)} = O(\log M) \cdot \tilde O(1/\xi) \cdot n^{\omega(\sigma) + o(1)} $.
(In \textsf{PRAM} setting, this expression bounds the work complexity of the algorithm while the time complexity is 
polylogarithmic in $n$. Note that all the $m = O(\log M')$ computations of exact distance products that are needed 
for computing a single approximate one can be done in parallel.)

As $G' = (V,H)$ is a $(1 + \epsilon, O(\log n))$-hopset of $G = (V,E)$ and $A$ is the adjacency matrix of the augmented graph $\tilde G = (V, E \cup H)$, it follows that this computation produces $(1 + \epsilon) \cdot (1 + \xi)^{O(\log n)}$-approximate $S \times V$ distances. We set $\xi = \frac{\epsilon}{O(\log n)}$, and obtain $(1 + O(\epsilon))$-approximate 
$S \times V$-distances in centralized time $O(\log W) \cdot \tilde O(1/\xi) \cdot n^{\omega(\sigma)}$
(and polylogarithmic parallel time and work given by this expression).
To summarize:
\begin{theorem}\label{thm:distanceProductsAppdx}
    For a parameter $\epsilon > 0$, suppose that we are given a $(1 + \epsilon, O(\log n))$-hopset $G' = (V,H)$
    of a graph $G = (V,E)$ with non-negative edge weights, with aspect ratio at most $W$, and a subset $S \subseteq V$ of $|S| = n^{\sigma}$ sources. Then the algorithm from Section~\ref{sec:distanceProducts} computes $(1 + O(\epsilon))$-approximate $S \times V$ distances within 
    $O(\log W) \cdot \tilde O(1/\epsilon) \cdot n^{\omega(\sigma)}$ centralized time, or within parallel polylogarithmic in $n$ time and work bounded by this expression.
\end{theorem}

\section{Tree Decompositions}\label{sec:treeDecompAppedix}
Lipton and Tarjan (\cite{LT83}, Corollary $3$) showed that in planar graphs, one can modify a recursive separator of size $O(\sqrt n)$
that partitions the graph into two parts of size at most $\lambda \cdot n$ each, for a constant $1/2 < \lambda < 1$, and guarantee $\lambda = 1/2$, while increasing the size of the separator by only a constant factor.

In this appendix we generalize their argument to graphs that admit $k^{\rho}$-recursive separators,
for any constant $0 < \rho < 1$.
Moreover, we show that with a small overhead, this argument applies even if the input graph admits 
recursive separators with a flat bound $O(n^{\rho})$, for a constant $0 < \rho < 1$.
Specifically, the original separator can be converted into a separator of size $O(n^{\rho} \cdot \log n)$ with ratio $1/2$.
\begin{lemma}\label{lem:tdAppedix}
Let $c_{sep} > 0$, $0 < \rho < 1$, $1/2 < \lambda < 1$ be constants.\\
\begin{enumerate}
    \item \label{enum:lem:tdAppedix1} 
Suppose that a graph $G = (V, E)$, $|V| = n$, admits a $k^{\rho}$-recursive separator with ratio $\lambda$.
(See Section~\ref{sec:separatorDef} for relevant definitions.)
Then $G$ admits a separator $C'$ of size $|C'| = O_{\rho, \lambda} (c_{sep} \cdot n^{\rho})$ with ratio $1/2$,
i.e., both parts $A'$ and $B'$ of $V \setminus C'$ satisfy $|A'|, |B'| \le n/2$.
\item \label{enum:lem:tdAppedix2} 
Suppose that $G = (V,E)$ admits a recursive separator with a flat bound of $c_{sep} \cdot n^{\rho}$
on its size. Then $G$ admits a separator $C'$ of size 
\begin{equation}\label{eq:lem:tdAppedix}
  |C'| = O_{\rho, \lambda}(c_{sep} \cdot \log n \cdot n^{\rho}),  
\end{equation}
with ratio $1/2$.
\end{enumerate}\label{enum:lem:tdAppedix}
\end{lemma}
\begin{proof}
In case~\ref{enum:lem:tdAppedix1} we assume that for any subset $\mathcal U \subset V$ of vertices,
there is a separator $C^*$ of size at most $c_{sep} \cdot |\mathcal U|^{\rho}$ such that
$\mathcal U \setminus C^*$ decomposes into parts $A^*$ and $B^*$ of sizes at most $\lambda |\mathcal U|$ each.
In case~\ref{enum:lem:tdAppedix2} the size of the separator $C^*$  is at most $c_{sep} \cdot n^{\rho}$, 
even when $|\mathcal U| \ll n$.

For both cases our argument follows closely the proof of~\cite{LT83}, Corollary $3$.
We analyze the two cases together, and split the analysis towards the end of the proof.

We define sequences of sets $(A_i, B_i, C_i, D_i)_{i \ge 0}$ that satisfy the following four properties:
\begin{enumerate}
\item \label{enum:LTProp1} 
For every $i$, $A_i, B_i, C_i, D_i$ is a partition of $V$.
\item \label{enum:LTProp2} 
There are no edges between $A_i$ and $B_i$, $A_i$ and $D_i$ and $B_i$ and $D_i$.
(In particular, if $D_i = \phi$, then $C_i$ is a separator, and $A_i$, $B_i$ are the two parts.)
\item \label{enum:LTProp3} 
$|A_i| \le |B_i| \le |A_i|+  |C_i| + |D_i|$.
\item \label{enum:LTProp4} 
For every $i$ such that both $(A_i, B_i, C_i, D_i)$ and $(A_{i + 1}, B_{i + 1}, C_{i + 1}, D_{i + 1})$  
are defined,we have
\[|D_{i + 1}| \le \lambda \cdot | D_i |, \]
\end{enumerate}
  i.e., eventually the sets $D_i$ vanish.

To define this sequence, we start with defining $A_0 = B_0 = C_0 = \phi$, $D_0 = V$.
The properties~\ref{enum:LTProp1}-\ref{enum:LTProp4} hold (partly vacuously).
Assume now that $A_i, B_i, C_i, D_i$ are defined for some $i$. Build a balanced separator $C^*$ for $G(D_i)$,
the graph induced by $D_i$.
Denote by $A^*$, $B^*$ the two parts of $D_i\setminus C^*$.
We have $|A^*|, |B^*| \le \lambda |D_i| $ in both cases.

In case~\ref{enum:lem:tdAppedix1} we also have 
\[
      |C^*| \le c_{sep} \cdot |D_i|^{\rho},
\]
while in case~\ref{enum:lem:tdAppedix2} the bound becomes
\[
      |C^*| \le c_{sep} \cdot n^{\rho},
\]
In both cases there are no edges between $A^*$ and $B^*$.
See Figure~\ref{fig:FigureTreeDecomp} for an illustration.

\begin{figure}[ht!]
 \captionsetup{font=scriptsize}
	\centering
	\fbox{
          \includegraphics[scale=0.25]{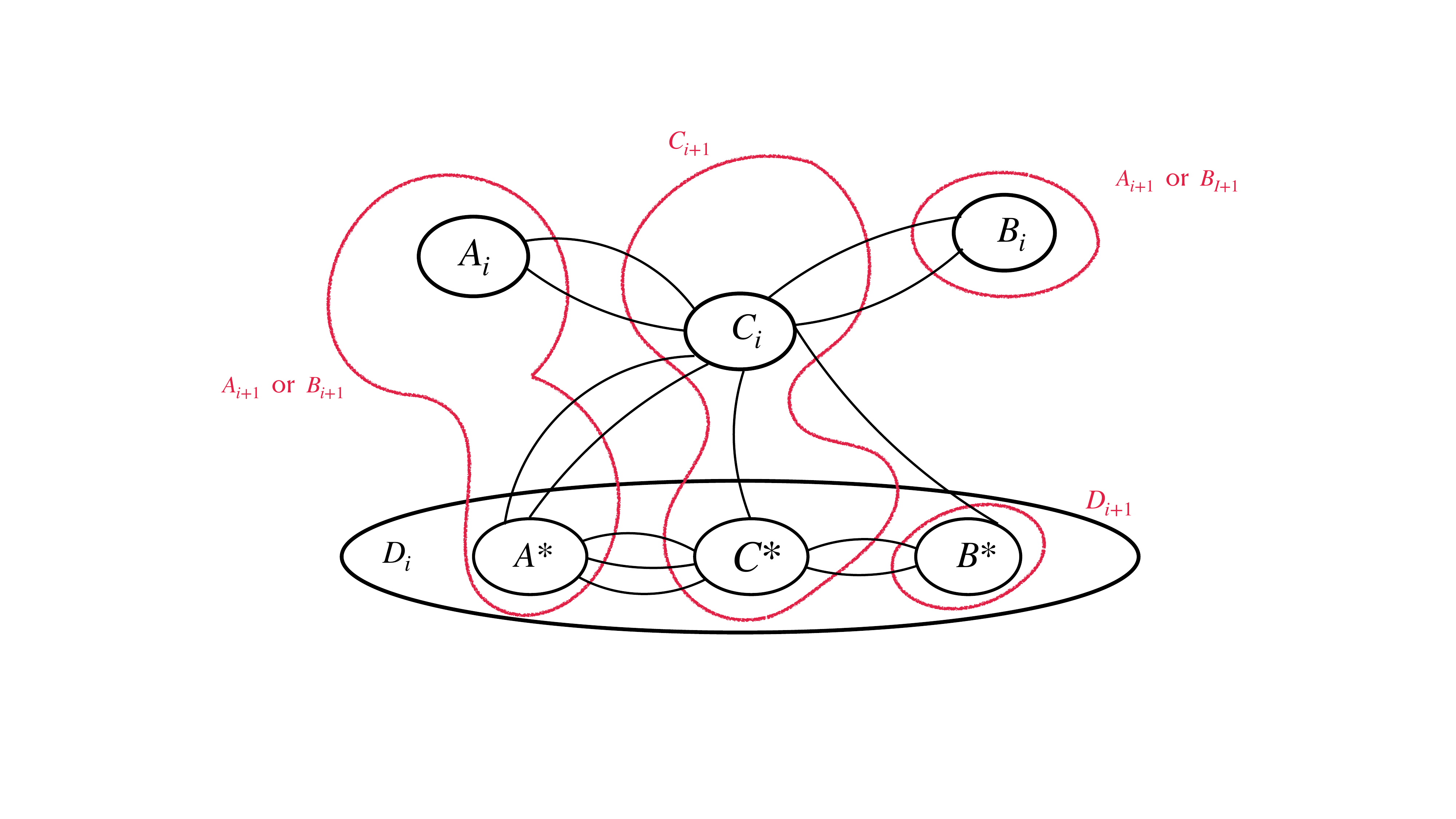}
	}
	\caption{Different sets and incidences between them in the proof of Lemma~\ref{lem:tdAppedix}.}
	\label{fig:FigureTreeDecomp}
\end{figure}

Assume without loss of generality that $|A^*| \le |B^*|$.
Let $A_{i+1}$ (respectively, $B_{i + 1}$) be the smaller (respectively, larger)
set among $A_i \cup A^*$ and $B_i$.
Let $C_{i + 1} = C_i \cup C^*$ and $D_{i + 1} = B^*$.

By the inductive hypothesis, and since $C^*$ is a separator of $D_i$ into two parts $A^*$, $B^*$,
it is easy to verify that the sets $A_i \cup A^*$, $B_i$, $C_i \cup C^*$, $B^*$ form a partition.
Also, there are no edges between $A_i \cup A^*$ and $B_i$, because there are no $A_i \times B_i$
edges and no $B_i \times D_i$ edges, and $A^* \subseteq D_i$.
Similarly, there are no edges between $A_i \cup A^*$ and $D_{i + 1} = B^*$, because there are no 
$A_i \times D_i$ edges and $B^* \subseteq D_i$, and because there are no $A^* \times B^*$ edges.
Finally, there are no edges between $B_i$ and $D_{i + 1} = B^* \subseteq D_i$.

To verify property~\ref{enum:LTProp3} assume first that $A_{i + 1} = A_i \cup A^*$, 
$B_{i + 1} = B_i$, i.e.,
\begin{align*}
\begin{aligned}
       |A_{i + 1}| &= |A_i| + |A^*| 
               \le |B_i| = |B_{i + 1}|.
\end{aligned}
\end{align*}
Then, by the inductive hypothesis, 
\begin{align*}
    \begin{aligned}
        |B_{i + 1}| = |B_i| &\le |A_i| + |C_i| + |D_i| \\
                            &= |A_i| + |A^*| + |C_i| + |C^*| + |B^*| \\
                            &= |A_{i + 1}| + |C_{ i + 1}| + |D_{ i + 1}|,
    \end{aligned}
\end{align*}
as required.
Otherwise, we have $A_{i +1} = B_i$, $B_{i + 1} = A_i \cup A^*$ and 
\[|A_i| \le |B_i| = |A_{i + 1}| \le |A_i| + |A^*| = |B_{i+1}|.\]
Then, 
\begin{align*}
    \begin{aligned}
        |B_{i + 1}| = |A_i| + |A^*| &\le |B_i| + |B^*| \\
                     &  = |A_{i + 1}| + |D_{i + 1}| \\
                     &\le |A_{i + 1}| + |C_{i + 1}| + |D_{i + 1}|.
    \end{aligned}
\end{align*}
This completes the proof of property~\ref{enum:LTProp3}.

Property~\ref{enum:LTProp4} ($|D_{i + 1}| \le \lambda \cdot |D_i|$)
follows from the definition of separator (with ratio $\lambda$).

The sequence $(A_i, B_i, C_i, D_i)_{i \ge 0}$ can be continued as long as $D_i \neq \phi$.
Let $k$ be the last index in the sequence.
Consider the resulting sets $A_k, B_k, C_k, D_k = \phi$.
We claim that $C = C_k$ is the separator of $G$ (into two parts $A= A_k$, $B = B_k$)
that satisfies the assumption of the lemma.

As there are no edges between $A$ and $B$, the set $C$ is indeed a separator.
By property~\ref{enum:LTProp3} we have $|A| \le |B| \le |A| + |C|$.
As $|A| + |B| + |C| = n$, we conclude that 
\[
      2|A| \le |A| + |B| \le n,
\]
i.e., $|A| \le n/2$.

Also, 
\[
      2|B| \le |A| + |C| +  |B| = n,
\]
and thus, $|B| \le n/2$ too.

Finally, in case~\ref{enum:lem:tdAppedix1} (i.e., when the graph $G$ admits 
a $k^{\rho}$-separator) the size of the separator $C$ is bounded by 
\begin{align*}
      \begin{aligned}
          |C| &\le \sum_{i = 0}^{\infty} c_{sep} \cdot (n \cdot \lambda^i)^{\rho} \\
              & = c_{sep} \cdot n^{\rho} \cdot \frac{1}{1 - \lambda^{\rho}} \\
              & = O _{c_{sep}, \lambda, \rho}(n^{\rho}).           
      \end{aligned}
\end{align*}
In case~\ref{enum:lem:tdAppedix2} (when $G = (V,E)$  admits a flat bound of $c_{sep} \cdot n^{\rho}$
on the size of separator for any induced subgraph $G(\mathcal U)$, $\mathcal U \subseteq V$),
the routine inserts to the separator at most $c_{sep} \cdot n^{\rho}$ vertices on each of the at most $O(\log n)$ 
steps of the process described in this proof.

Thus, the total size of the separator is at most
\[
    O_{c_{sep}, \lambda, \rho} \left(n^{\rho} \cdot \log n \right)
\]
in this case.
\end{proof}
Another property that is important for us is that for any vertex $v$ in the separator $S$, it should have
at least one neighbour in each of the two parts into which $V \setminus S$ decomposes.
A separator that satisfies that property will be called \emph{double-incident}.

Given a separator $S$ and two parts $A, B$ of $V\setminus S$, one can convert $S$ 
into a doubly-incident separator in the following way:
iteratively, we pick a vertex $v \in S$ that has neighbours only in $A$ (respectively, only in $B$)
and add to $A$ (respectively to $B$).
A vertex that has neighbours neither in $A$ nor in $B$ can be added to either set.
Obviously, as a result we obtain a doubly-incident separator.
\begin{cor}\label{cor:doublyIncident}
   A graph $G = (V, E)$, $|V| = n$,
that admits a $k^{\rho}$-recursive separator with ratio $\lambda$,
for some $0 < \rho < 1$, $1/2 < \lambda < 1$,
admits a doubly-incident separator $C$ of size
$|C| = O(n^{\rho})$ and such that both parts $A$ and $B$ satisfy 
\[
     |A|, |B| \le \frac{n}{2} + O(n^{\rho}).
\]
If $G$ admits a flat bound of $O(n^{\rho})$ on the size
of its recursive separator with ratio $\lambda$, 
then it admits a doubly-incident separator $C$ of size
$|C| = O(n^{\rho} \cdot \log n)$ and such that 
\[
    |A|, |B| \le \frac{n}{2} + O(n^{\rho} \cdot \log n).
\]
\end{cor}

\section{Finite Element Graphs}\label{sec:finiteElemnt}
Another graph family that admits small separators is \emph{finite element graphs}~\cite{LT83}.
These graphs arise in the finite element method in Numerical Analysis.

Given a planar graph $G' = (V,E)$, and a planar embedding of $G'$, a finite element graph $G$ of $G'$
is obtained by adding all possible diagonals between vertices that lie on the same face of $G'$.
(Observe that this adds edges only to faces with at least four vertices.)
The planar graph $G'$ is called the \emph{skeleton} of the finite element graph $G$.
Each of its faces is called an \emph{element} $G$.
For an integer parameter $k \ge 4$, we say that $G$ is a \emph{$k$-finite element graph} if each of its elements 
contains at most $k$ vertices.

Lipton and Tarjan~\cite{LT83} proved the following separator theorem for finite element graphs.
\begin{theorem}~\cite{LT83}\label{thm:liptonTarjan}
  Let $G$ be an $n$-vertex $k$-finite element graph, for an integer parameter $k \ge 4$.
  Then $G$ admits a balanced separator with $O(k\sqrt n)$ vertices.
\end{theorem}
We next argue that the balanced separator can be computed very efficiently, even in parallel setting.
To that end we first sketch the proof (due to~\cite{LT83}) of Theorem~\ref{thm:liptonTarjan}.
\begin{proof}[proof of Theorem~\ref{thm:liptonTarjan}]
Given a $k$-finite 
element graph $G = (V,E)$, let $G^* = (V, E^*)$ be its planar skeleton.
Let $G^{**}$ be the planar graph obtained from $G^*$ by adding a new vertex $v_F$ for every face $F$
of $G^*$ with $|F| \ge 4$ (with at least four vertices). We also connect every vertex on $F$ with $v_F$.
We also assign unit vertex cost to every original vertex of $G$, and vertex cost of $\frac{1}{100}$
to every new vertex $v_F$.

Next, we compute a planar separator $(A^{**}, B^{**}, C^{**})$ for $G^{**}$ with these costs.
($A^{**}$, $B^{**}$ and $C^{**}$ are vertex subsets of $G^{**}$, with $A^{**}$ and $B^{**}$
having total costs at most $\frac{2}{3} (n + f' \cdot \frac{1}{100} )$ each, 
and $|C^{**}| = O(\sqrt {n + f'})$, where $f'$ is the number of faces with at least $4$
vertices. By Euler's planar formula, $f' = O(n)$, i.e., $|C^{**}| = O(\sqrt n)$.
Also, there are no edges between $A^{**}$ and $B^{**}$.

Once this separator is computed, we remove new vertices $v_F$ (for faces $F$ for which a new vertex $v_F$
was added) from $C^{**}$, and obtain a subset $C \subseteq V$ of size $O(\sqrt n)$ as well.
Let $A$, $B$ be the sets $A^{**}$, $B^{**}$, respectively, after removing from them new vertices.
We have $|A|, |B| \le \frac{2}{3} (n + \frac{n}{100}) \le 0.7 n$.

The problem is, however, that edges between $A$ and $B$ in $G$ may be present.
Note, however, that for every such edge $e = (u,v), u \in A, v \in B$, there was necessarily a (removed) 
separator vertex $v_F \in C^{**}$, for a face $F$ with $|F| \ge 4$.
We call these faces $F$ with $|F| \ge 4$ such that $v_F \in C^{**}$
\emph{bad faces}. Observe that as $|C^{**}| = O(\sqrt n)$ and the cost of every new vertex $v_F$ is $\frac{1}{100}$, we conclude that there are only $100 \cdot O(\sqrt n) = O(\sqrt n)$ bad faces.
To correct the separator $C$, for every bad face $F$, we add all its vertices that belong to $A$ (or to $B$) to the separator $C$.
By doing so we increase the size of the separator by an additive $O(\sqrt n \cdot k)$,
as every bad face contributes at most $k$ vertices.

It is easy to verify at this point that there are no edges between $A$ and $B$.
\end{proof}
Now we argue that this constructive proof can be efficiently parallelized.
Given a finite-element graph $G = (V,E)$, along with a planar embedding of its
skeleton\footnote{A planar embedding of a planar graph can be computed in $O(\log^2 n)$ depth and $O(n)$
work~\cite{KLEIN1988190}.}, we compute the graph $G^{**}$ as above.
(Computing the vertices $v_F$ and connecting them to all vertices of $F$ can clearly be done in polylogarithmic time and $O(n)$ processors, by assigning a processor to every vertex and face.)

We then compute a planar separator of $G^{**}$ in $O(\log^2 n)$ time, $O(n^{1 + \epsilon})$ processors,
for any positive constant $\epsilon > 0$, by Gazit-Miller's algorithm~\cite{GazitM87}.
Finally, we add to the separator all the vertices from bad faces. This again can be easily implemented within our resources. To summarize:
\begin{theorem}\label{thm:finiteGraphs}
    A balanced separator of a $k$-finite $n$-vertex element graph can be computed in randomized 
    polylogarithmic time, using $O(n^{1 + \epsilon})$ processors, for any constant $\epsilon > 0$.
    The size of the separator is $O(k \sqrt n)$.
\end{theorem}
By using Theorem~\ref{thm:finiteGraphs} recursively we can compute the separator decomposition $T_G$ of a 
given $k$-finite element graph $G$ within essentially the same resources.
(Note that any induced subgraph of a $k$-finite element graph is a $k$-finite element graph as well.
The time complexity grows by a factor of $O(\log n)$.)

In the centralized setting, a balanced separator of an $n$-vertex planar graph $G$ can be computed in $O(n)$ time~\cite{LT83}. As a result, a separator decomposition $T_G$ can be computed in $O(n \log n)$ time.
Hence, a separator decomposition of a $k$-finite element $n$-vertex graph $G^* = (V, E^*)$
whose skeleton is provided as a part of the input can also be computed in $O(n \log n)$ time.
We derive the following corollary:
\begin{cor}\label{cor:finiteElement}
    Given a directed weighted graph with non-negative edge weights and aspect ratio $W$, such that its undirected unweighted skeleton is an
    $n$-vertex $k$-finite element graph, for a parameter $k = k(n) \ge 4$, along with a planar
    embedding of its skeleton, our algorithm constructs an $O(\log n)$-shortcut for it in randomized polylogarithmic time and using $(\sqrt n k)^{\omega + o(1)} + n^{1 + \epsilon} = (\sqrt n k)^{\omega + o(1)}$ processors. For a $((1 + \epsilon), O(\log n))$-hopset, for a parameter $\epsilon > 0$, the centralized running time and the number of processors are multiplied by $O_{\epsilon} (\log W) \cdot \log^{O(1)} n$.
\end{cor}
Note that this algorithm works from scratch.
We also note that $n$-vertex $k$-finite element graphs may contain up to $O(nk)$ edges, and that this upper bound is tight.

To verify the upper bound consider an $n$-vertex $k$-finite element graph $G^* = (V, E^*)$, and 
let $G = (V,E)$ be its skeleton graph.
Let $F_1, F_2, \ldots, F_t$, for some integer $t \ge 1$, be its faces.
Then $\sum_{i =1}^{t} |F_i| = 2|E| = O(n) $. 
Also, $|E^*| = O(\sum_{i=1}^{t} |F_i|^2)$.
As for every $i \in [t]$, $|F_i| \le k$, the sum $\sum_{i=1}^{t} |F_i|^2$
is maximized by having every $|F_i| = k$.
Hence 
\begin{align*}
    |E^*| = O(\sum_{i=1}^{t} |F_i|^2) = O(\frac{n}{k} \cdot k^2) = O(n \cdot k).
\end{align*}
 To see that this bound is tight, consider a $k$-finite element graph whose
 planar skeleton graph contains $O(n/k)$ faces with $\Omega(k)$ vertices each.
\section{Tables for Section~\ref{sec:excludedMinor}}\label{sec:tablesExcludedMinor}
 \begin{table}[h!]
  \captionsetup{font=scriptsize}
  \begin{center}
  {\small 
    \begin{tabular}{c|c|c|c|c|c} 
     $\mu$ & $\eta$ &$\sigma_{thr}$ & $\sigma$ &$g(\eta, \sigma)$  & $N(\mu , \sigma)$\\
      \hline
                1.38 & $[0.38, 0.49]$ & 0.965 & 0.97 & $\omega(0.97) = 2.3481$ & 2.35\\
                                  && & 0.98 & $\omega(0.98) = 2.3559$ & 2.36\\
                                   && & 0.99 &  $\omega(0.99) = 2.36387$ & 2.37\\ \hline                         
                1.39 & $[0.39, 0.47]$ & 0.916 & 0.92 & $\omega(0.92) = 2.3095$ & 2.31\\
                               && & 0.93 & $\omega(0.93) = 2.3171$ & 2.32 \\
                               && & 0.94 & $\omega(0.94)  =2.3248$ & 2.33\\
                               && & 0.95 & $\omega(0.95) =  2.3324$ & 2.34\\
                               && & 0.96 & $\omega(0.96) = 2.3403$ & 2.35 \\
                               && & 0.97 & $\omega(0.97)$ & 2.36\\
                               && & 0.98 & $\omega(0.98)$ & 2.37\\
                               && & 0.99 & $\omega(0.99)$ & $\omega$ \\ \hline
                        &0.48 & 0.934 & 0.94 & $\omega(0.94)$ & 2.33\\
                                 &&  & 0.95 & $\omega(0.95)$ & 2.34 \\
                                 &&  & 0.96 & $\omega(0.96)$ & 2.35\\
                                 &&  & 0.97 & $\omega(0.97)$ & 2.36\\
                                 &&  & 0.98 & $\omega(0.98)$ & 2.37\\
                                 &&  & 0.99 & $\omega(0.99)$ &$\omega$\\ \hline
                        &0.49  & 0.9578 & 0.96 & $\omega(0.96)$ & 2.35 \\
                               && & 0.97 &  $\omega(0.97)$ & 2.36 \\
                               && & 0.98 &  $\omega(0.98)$ & 2.37\\
                               && & 0.99 & $\omega(0.99)$ & $\omega$ \\ \hline
                1.4 & $[0.4, 0.46]$ & 0.876 & 0.88 & $\omega(0.88) =  2.2793$ & 2.28 \\
                               && & 0.89 & $\omega(0.89) =  2.2868$ & 2.29 \\
                               && & 0.9 & $\omega(0.9) =  2.2942$ & 2.3\\
                               && & 0.91 & $\omega(0.91) = 2.3019$ & 2.31 \\
                               && & 0.92 & $\omega(0.92)$ & 2.32 \\
                               && & 0.93 &  $\omega(0.93)$ & 2.33\\
                               && & 0.94 &  $\omega(0.94)$ & 2.34 \\
                               && & 0.95 & $\omega(0.95)$ & 2.35 \\
                               &&  & 0.96 & $\omega(0.96)$ & 2.36\\
                               &&  & 0.97 & $\omega(0.97)$ & 2.37\\
                               &&  & 0.98 & $\omega(0.98)$ & $\omega$\\
                               &&  & 0.99 & $\omega(0.99)$ & $\omega$\\ \hline
                  & 0.47 & 0.9004 & 0.91 & $\omega(0.91)$ & 2.31\\
                  && & 0.92 & $\omega(0.92)$ & 2.32\\
                 && & \vdots & \vdots & \vdots \\
                  && & 0.99 & $\omega(0.99)$ & $\omega$ \\ \hline
                  & 0.48  & 0.924  & 0.93 & $\omega \cdot 0.98 = 2.3241$ & 2.33 \\
                  &&& 0.94 & $\omega(0.94)$ & 2.34\\
                  &&& 0.95 & $\omega(0.95)$ & 2.35\\
                   &&& 0.96 & $\omega(0.96)$ & 2.36\\
                    &&& 0.97 & $\omega(0.97)$ & 2.37\\
                    &&& 0.98 & $\omega(0.98)$ & $\omega$\\
                    &&& 0.99 & $\omega(0.99)$ & $\omega$\\ \hline
                  & 0.49 & 0.9478 & 0.95 & $\omega \cdot 0.99 = 2.3478$ & 2.35 \\
                               &&& 0.96 & $\omega \cdot 0.99$ & 2.36\\
                               &&& 0.97 & $\omega(0.97)$ & 2.37\\
                               &&& 0.98 & $\omega(0.98)$ & $\omega$\\
                               &&& 0.99 & $\omega(0.99)$ & $\omega$ \\   
        \end{tabular}
  }
  \end{center}
    \caption{The table provides a comparison of the exponent $g(n, \sigma)$ of our algorithm
              for $S \times V$-reachability ($|S| = n^{\sigma}$) on graphs with excluded $K_h$-minor 
              ($h = n^{\eta}$) with the exponent $N(\mu, \sigma)$ of the state-of-the-art naive solution 
              ($m = n^{\mu}$) in the range of parameters where $g(n, \sigma) < N(\mu, \sigma)$.
              This table provides values for $\omega - 2 < \eta \le 1.4$.
              Higher values ($1.4 < \eta < 1.5$) appear in Tables~\ref{tab:sigmaThr2}-\ref{tab:sigmaThr4}. The same bounds apply also to the $k$-finite element graphs, when $k = n^{\eta}$.
          }
    \label{tab:sigmaThr1}
\end{table}

\begin{table}[h!]
  \captionsetup{font=scriptsize}
  \begin{center}
  {\small 
    \begin{tabular}{c|c|c|c|c|c} 
     $\mu$ & $\eta$ &$\sigma_{thr}$ & $\sigma$ &$g(\eta, \sigma)$  & $N(\mu , \sigma)$\\
      \hline
                 1.41 & $[0.41, 0.44]$ & 0.839 & 0.84 & $ \omega(0.84) =2.2498$ & 2.25\\
                               &&&0.85 &  $ \omega(0.85)  = 2.2569$ & 2.26\\
                               &&& 0.86 & $\omega(0.86)  = 2.2644$ & 2.27\\
                               &&& 0.87 &  $\omega(0.87)  = 2.2719 $ & 2.28\\
                               &&& 0.88 & $\omega(0.88)$ & 2.29\\
                               &&& 0.89 & $\omega(0.89)$ & 2.3\\
                               &&& 0.9 & $\omega(0.9)$ & 2.31\\
                               &&& \vdots & \vdots &\vdots \\
                               &&& 0.95 & $\omega(0.95)$ & 2.36\\
                               &&& 0.96 & $\omega(0.96)$ & 2.37\\
                               &&& 0.97 & $\omega(0.97)$ & $\omega$\\
                               &&& 0.98 & $\omega(0.98)$ & $\omega$\\
                               &&& 0.99 & $\omega(0.99)$ & $\omega$\\   \hline

                        & 0.45 & 0.8430 & 0.85 & $\omega(0.85)$ & 2.26\\
                                       &&& 0.86 & $\omega(0.86)$ & 2.27\\
                                       &&& \vdots & \vdots & \vdots\\
                                       &&& 0.96 & $\omega(0.96)$ &2.37\\
                                       &&& 0.97 & $\omega(0.97)$ & $\omega$ \\
                                       &&& 0.98 & $\omega(0.98)$ &  $\omega$\\
                                       &&& 0.99 & $\omega(0.99)$ & $\omega$\\ \hline
                                       
                          & 0.46 & 0.867& 0.87 & $\omega(0.87)$ & 2.28 \\
                           &&& 0.88 & $\omega(0.88)$ & 2.29\\
                           &&& \vdots & \vdots & \vdots \\
                           &&& 0.99 & $\omega(0.99)$ & $\omega$\\ \hline
                        & 0.47 & 0.8904 & 0.9 & $ \omega \cdot 0.97 =2.3004$ & 2.31 \\
                                      &&& 0.91 & $\omega(0.91)$ & 2.32\\
                                      &&& 0.92 & $\omega(0.92)$ & 2.33 \\
                                       &&& \vdots & \vdots & \vdots \\
                                      &&& 0.99 & $\omega(0.99)$ & $\omega$\\ \hline
                        & 0.48 & 0.9141 & 0.92 & $\omega \cdot 0.98 =2.3241$ & 2.33\\
                                     &&&  0.93 & $\omega \cdot 0.98$ & 2.34\\
                                     &&& 0.94 & $\omega(0.94)$ & 2.35\\
                                     &&& 0.95 & $\omega(0.95)$ & 2.36\\
                                     &&& \vdots & \vdots & \vdots \\
                                     &&& 0.99 & $\omega(0.99)$ & $\omega$\\ \hline
                        &0.49 & 0.9378& 0.94 &$ \omega \cdot 0.99 = 2.3478$ & 2.35 \\
                                      &&& 0.95 & $ \omega \cdot 0.99$ & 2.35\\
                                      &&& 0.96 & $ \omega \cdot 0.99$ & 2.36 \\
                                      &&& 0.97 & $\omega(0.97)$ & $\omega$\\
                                      &&&0.98 & $\omega(0.98)$ & $\omega$\\
                                      &&&0.99 & $\omega(0.99)$ & $\omega$\\ \hline
                                                   
        \end{tabular}
  }
  \end{center}
    \caption{Continuation of Table~\ref{tab:sigmaThr1}. 
    Contains values of the functions $g(\eta, \sigma)$, $N(\mu, \sigma)$ for $\eta = 1.41$. 
    Higher values $(1.41 < \eta < 1.5)$ appear in Tables~\ref{tab:sigmaThr3} and~\ref{tab:sigmaThr4}.
          }
    \label{tab:sigmaThr2}
\end{table}
\begin{table}[h!]
  \captionsetup{font=scriptsize}
  \begin{center}
  {\small 
    \begin{tabular}{c|c|c|c|c|c} 
     $\mu$ & $\eta$ &$\sigma_{thr}$ & $\sigma$ &$g(\eta, \sigma)$  & $N(\mu , \sigma)$\\
      \hline
                 1.45 & 0.45 & 0.803 & 0.81 & $\omega \cdot 0.95 =  2.2530$ & 2.26\\
                               &&&0.82 & $\omega \cdot 0.95$ & 2.27\\
                               &&& 0.83 & $\omega \cdot 0.95$ & 2.28\\
                               &&& 0.84 & $\omega \cdot 0.95$ & 2.29\\
                               &&& 0.85 & $\omega(0.85) =2.2569$ & 2.3\\
                               &&& 0.86 & $\omega(0.86) =2.2644$ & 2.31\\
                               &&& \vdots & \vdots &\vdots \\
                        \rowcolor{Gray}
                               &&& 0.9 & $\omega(0.9) = 2.2942$ & 2.35\\
                        \rowcolor{Gray}
                               &&& 0.91 & $\omega(0.91) = 2.3019$ & 2.36\\
                        \rowcolor{Gray}
                               &&& 0.92 & $\omega(0.92) = 2.3095$ & 2.37\\
                        \rowcolor{Gray}
                               &&& 0.93 & $\omega(0.93) =2.3171$ & $\omega$\\
                                &&& \vdots & \vdots &\vdots \\
                               &&& 0.99 & $\omega(0.99) = 2.3637$ & $\omega$\\   \hline
                        &0.46 & 0.8267 & 0.83 & $\omega \cdot 0.96 = 2.2767$ & 2.28\\
                                     &&& 0.84 & $\omega \cdot 0.96$ & 2.29 \\
                                     &&& 0.85 & $ \omega \cdot 0.96$ & 2.3\\
                                     &&& 0.86 & $ \omega \cdot 0.96$ & 2.31\\
                                     &&& 0.87 & $ \omega \cdot 0.96$ & 2.32\\
                            \rowcolor{Gray}
                                     &&& 0.88 & $\omega(0.88) =  2.2797$ & 2.33\\
                              \rowcolor{Gray}
                                     &&& 0.89 & $\omega(0.89) = 2.2868$ & 2.34\\
                              \rowcolor{Gray}
                                     &&& 0.9 & $\omega(0.9) = 2.2942$ & 2.35\\
                                  \rowcolor{Gray}
                                     &&& 0.91 & $\omega(0.91) = 2.3019 $ & 2.36\\
                            \rowcolor{Gray}
                                     &&& 0.92 & $\omega(0.92) =  2.3095$ & 2.37\\
                              \rowcolor{Gray}
                                     &&& 0.93 & $\omega(0.93) = 2.3171$ & $\omega$\\
                                    &&& \vdots & \vdots &\vdots \\
                                     &&& 0.99 & $\omega(0.99) =  2.3637$ & $\omega$ \\ \hline
                        &0.48 & 0.874 & 0.88 &  $\omega \cdot 0.98 =2.3241$ & 2.33\\
                                    &&& 0.9 & $\omega \cdot 0.98$ & 2.35\\
                                    &&& 0.93 & $\omega \cdot 0.98$ & $\omega$ \\
                                    &&& 0.94 & $\omega(0.94) = 2.3248 $ & $\omega$\\
                                    &&& 0.95 & $\omega(0.95) =  2.33324$ & $\omega$ \\
                                  &&& \vdots & \vdots & \vdots \\
                                    &&& 0.99 & $\omega(0.99)$ & $\omega$ \\
                                                                        
        \end{tabular}
  }
  \end{center}
    \caption{Continuation of Tables~\ref{tab:sigmaThr1} and~\ref{tab:sigmaThr2}.
              Contains values of the functions $g(\eta, \sigma)$, $N(\mu, \sigma)$ for $\eta = 1.45$.
              Higher values $(1.45 < \eta < 1.5)$ appear in Table~\ref{tab:sigmaThr4}.
              Highlighted rows indicate particularly significant improvements.
          }
    \label{tab:sigmaThr3}
\end{table}

\begin{table}[h!]
  \captionsetup{font=scriptsize}
  \begin{center}
  {\small 
    \begin{tabular}{c|c|c|c|c|c} 
     $\mu$ & $\eta$ &$\sigma_{thr}$ & $\sigma$ &$g(\eta, \sigma)$  & $N(\mu , \sigma)$\\
      \hline
                 1.47 & 0.47 & 0.8304 & 0.84 & $\omega \cdot 0.97 = 2.3004$ & 2.31\\
                               &&&0.85 & $\omega \cdot 0.97$ & 2.32\\
                       \rowcolor{Gray}
                               &&& 0.9 & $\omega \cdot 0.97$ & 2.37\\
                         \rowcolor{Gray}
                               &&& 0.91 & $\omega(0.91) = 2.3019$ & $\omega$ \\
                        \rowcolor{Gray}
                               &&& 0.92 & $\omega(0.92) = 2.3095$ & $\omega$ \\
                                &&& \vdots & \vdots &\vdots \\                                
                               &&& 0.99 & $\omega(0.99) = 2.3637$ & $\omega$\\   \hline
                        &0.49 & 0.8778 & 0.88 & $\omega \cdot 0.99 = 2.3478$ & 2.35\\
                                     &&& 0.9 & $\omega \cdot 0.99$ & 2.37\\
                                     &&& $[0.91, 0.96]$ & $\omega \cdot 0.99$ & $\omega$ \\
                                     &&& 0.97 & $\omega(0.97) =  2.3481$ & $\omega$\\
                                     &&& 0.98 & $\omega(0.98) =  2.3559$ & $\omega$ \\
                                     &&& 0.99 & $\omega(0.99)$ & $\omega$ \\
            \hline
                 1.49 & 0.49 & 0.8578 & 0.86 & $\omega \cdot 0.99$ & 2.35\\
                               &&&0.87 & $\omega \cdot 0.99$ & 2.36\\
                               &&& 0.88 & $\omega \cdot 0.99$ & 2.37\\
                               &&& $[0.89, 0.96]$ & $\omega \cdot 0.99$ & $\omega$ \\                       
                               &&& 0.97 & $\omega(0.97)$ & $\omega$\\ 
                               &&& 0.98 & $\omega(0.98)$ & $\omega$\\ 
                               &&& 0.99 & $\omega(0.99)$ & $\omega$\\ 
        \end{tabular}
  }
  \end{center}
    \caption{Continuation of Tables~\ref{tab:sigmaThr1},~\ref{tab:sigmaThr2} and~\ref{tab:sigmaThr3}.
              Contains values of the functions $g(\eta, \sigma)$, $N(\mu, \sigma)$ for $\eta = 1.47$.
              Highlighted rows indicate particularly significant improvements.
          }
    \label{tab:sigmaThr4}
\end{table}


%% file: DirectedReachabilityBib.bib
@inproceedings{LeGallBestSquare,
author = {Le Gall, Fran\c{c}ois},
title = {Powers of Tensors and Fast Matrix Multiplication},
year = {2014},
isbn = {9781450325011},
publisher = {Association for Computing Machinery},
address = {New York, NY, USA},
doi = {10.1145/2608628.2608664},
booktitle = {Proceedings of the 39th International Symposium on Symbolic and Algebraic Computation},
pages = {296-303},
numpages = {8},
series = {ISSAC '14}
}

@inproceedings{LeGallBestRectangular,
    author = {Gall, Fran\c{c}ois Le and Urrutia, Florent},
    title = {Improved Rectangular Matrix Multiplication Using Powers of the Coppersmith-Winograd Tensor},
    year = {2018},
    isbn = {9781611975031},
    publisher = {Society for Industrial and Applied Mathematics},
     booktitle = {Proceedings of the Twenty-Ninth Annual ACM-SIAM Symposium on Discrete Algorithms},
    pages = {1029-1046},
    numpages = {18},
    address = {New Orleans, Louisiana, USA},
    series = {SODA '18}
    }

@inproceedings{gall2023faster,
      title={Faster Rectangular Matrix Multiplication by Combination Loss Analysis}, 
      author={Le Gall, Fran\c{c}ois},
	booktitle = {Proceedings of the 2024 Annual ACM-SIAM Symposium on Discrete Algorithms (SODA)},
	pages = {3765-3791},
	year = {2024},
	address = {Alexandria, Virginia, USA},
	doi = {10.1137/1.9781611977912.133}
}

@inproceedings{VassilevskaBestRectangular,
        author = {Virginia Vassilevska Williams and Yinzhan Xu and Zixuan Xu and Renfei Zhou},
	title = {New Bounds for Matrix Multiplication: from Alpha to Omega},
	booktitle = {Proceedings of the 2024 Annual ACM-SIAM Symposium on Discrete Algorithms (SODA)},
	chapter = {},
	year = {2024},
	pages = {3792-3835},
	address = {Alexandria, Virginia, USA},
	doi = {10.1137/1.9781611977912.134}
}

@inProceedings{KoganParterDiShortcuts1,
	author = {Shimon Kogan and Merav Parter},
	title = {New Diameter-Reducing Shortcuts and Directed Hopsets: Breaking the $\sqrt n$ Barrier},
	booktitle = {Proceedings of the 2022 Annual ACM-SIAM Symposium on Discrete Algorithms 	 (SODA)},
	year = {2022},
	address = {Held Virtually},
	pages = {1326-1341}
 }

@InProceedings{KoganParterDiShortcuts2,
  author =	{Kogan, Shimon and Parter, Merav},
  title =	{Beating Matrix Multiplication for $n^{1/3}$-Directed Shortcuts},
  booktitle =	{49th International Colloquium on Automata, Languages, and Programming (ICALP 2022)},
  pages =	{82:1--82:20},
  series =	{Leibniz International Proceedings in Informatics (LIPIcs)},
  ISBN =	{978-3-95977-235-8},
  year =	{2022},
  volume =	{229},
  address =	{Paris, France},
  doi =		{10.4230/LIPIcs.ICALP.2022.82}
}

@inProceedings{KoganParterDiShortcuts3,
	author = {Shimon Kogan and Merav Parter},
	title = {Faster and Unified Algorithms for Diameter Reducing Shortcuts and Minimum Chain Covers},
	booktitle = {Proceedings of the 2023 Annual ACM-SIAM Symposium on Discrete Algorithms (SODA)},
	year = {2023},
	pages = {212-239},
	address = {Florence, Italy},
	doi = {10.1137/1.9781611977554.ch9}
}

@article{APSPDirectedZwick,
author = {Zwick, Uri},
title = {All Pairs Shortest Paths Using Bridging Sets and Rectangular Matrix Multiplication},
year = {2002},
issue_date = {May 2002},
publisher = {Association for Computing Machinery},
address = {New York, NY, USA},
volume = {49},
number = {3},
issn = {0004-5411},
journal = {J. ACM},
doi = {10.1145/567112.567114},

}

@article{COPPERSMITH,
title = {Matrix multiplication via arithmetic progressions},
journal = {Journal of Symbolic Computation},
volume = {9},
number = {3},
pages = {251-280},
year = {1990},
note = {Computational algebraic complexity editorial},
issn = {0747-7171},
doi = {https://doi.org/10.1016/S0747-7171(08)80013-2},
author = {Don Coppersmith and Shmuel Winograd}
}

@inproceedings{ElkinN22,
  author       = {Michael Elkin and
                  Ofer Neiman},
  editor       = {Petra Berenbrink and
                  Benjamin Monmege},
  title        = {Centralized, Parallel, and Distributed Multi-Source Shortest Paths
                  via Hopsets and Rectangular Matrix Multiplication},
  booktitle    = {39th International Symposium on Theoretical Aspects of Computer Science,
                  {STACS} 2022, March 15-18, 2022, Marseille, France (Virtual Conference)},
  series       = {LIPIcs},
  volume       = {219},
  pages        = {27:1--27:22},
  publisher    = {Schloss Dagstuhl - Leibniz-Zentrum f{\"{u}}r Informatik},
  year         = {2022},
  doi          = {10.4230/LIPICS.STACS.2022.27}
}

@inProceedings{BernsteinWein,
author = {Aaron Bernstein and Nicole Wein},
title = {Closing the Gap Between Directed Hopsets and Shortcut Sets},
booktitle = {Proceedings of the 2023 Annual ACM-SIAM Symposium on Discrete Algorithms (SODA)},
chapter = {},
pages = {163-182},
doi = {10.1137/1.9781611977554.ch7},
year = {2023},
address = {Florence, Italy}
}

@article{Cohen93,
title = {Efficient Parallel Shortest-Paths in Digraphs with a Separator Decomposition},
journal = {Journal of Algorithms},
volume = {21},
number = {2},
pages = {331-357},
year = {1996},
issn = {0196-6774},
doi = {https://doi.org/10.1006/jagm.1996.0048},
author = {Edith Cohen},
}

@inproceedings{AGM91,
author = {Alon, Noga and Galil, Zvi and Margalit, Oded},
title = {On the exponent of the all pairs shortest path problem},
year = {1991},
isbn = {0818624450},
publisher = {IEEE Computer Society},
address = {USA},
booktitle = {Proceedings of the 32nd Annual Symposium on Foundations of Computer Science},
pages = {569–575},
numpages = {7},
location = {San Juan, Puerto Rico},
series = {SFCS '91}
}

@article{Youval76,
        author={Gideon Yuval},
         title={An Algorithm for Finding All Shortest Paths using $n^{2.81}$ Infinite-Precision Multiplications},
        journal={Inf. Process. Lett.},
        year={1976},
        volume={4},
        pages={155-156}
        
        }

@article{LT83,
author = {Lipton, Richard J. and Tarjan, Robert Endre},
title = {A Separator Theorem for Planar Graphs},
journal = {SIAM Journal on Applied Mathematics},
volume = {36},
number = {2},
pages = {177-189},
year = {1979},
doi = {10.1137/0136016}
}

@INPROCEEDINGS{MillerTengVavasis91,

  author={Miller, G.L. and Teng, S.-H. and Vavasis, S.A.},

  booktitle={[1991] Proceedings 32nd Annual Symposium of Foundations of Computer Science}, 

  title={A unified geometric approach to graph separators}, 

  year={1991},

  volume={},

  number={},

  pages={538-547},

  doi={10.1109/SFCS.1991.185417}}

@inproceedings{MillerVavasis91,
  author       = {Gary L. Miller and
                  Stephen A. Vavasis},
  editor       = {Alok Aggarwal},
  title        = {Density Graphs and Separators},
  booktitle    = {Proceedings of the Second Annual {ACM/SIGACT-SIAM} Symposium on Discrete
                  Algorithms, 28-30 January 1991, San Francisco, California, {USA}},
  pages        = {331--336},
  publisher    = {{ACM/SIAM}},
  year         = {1991}
}

@inproceedings{MillerThurston90,
  author       = {Gary L. Miller and
                  William P. Thurston},
  editor       = {Harriet Ortiz},
  title        = {Separators in Two and Three Dimensions},
  booktitle    = {Proceedings of the 22nd Annual {ACM} Symposium on Theory of Computing,
                  May 13-17, 1990, Baltimore, Maryland, {USA}},
  pages        = {300--309},
  publisher    = {{ACM}},
  year         = {1990},
  doi          = {10.1145/100216.100255}
}

@book{teng1991points,
  title={Points, Spheres, and Separators: a unified geometric approach to graph partitioning},
  author={Teng, Shang-Hua},
  year={1991},
  publisher={Carnegie Mellon University}
}

@article{EppsteinMillerTeng95,
  author       = {David Eppstein and
                  Gary L. Miller and
                  Shang{-}Hua Teng},
  title        = {A Deterministic Linear Time Algorithm for Geometric Separators and
                  its Applications},
  journal      = {Fundam. Informaticae},
  volume       = {22},
  number       = {4},
  pages        = {309--329},
  year         = {1995},
  doi          = {10.3233/FI-1995-2241}
}

@article{MillerTTV97,
  author       = {Gary L. Miller and
                  Shang{-}Hua Teng and
                  William P. Thurston and
                  Stephen A. Vavasis},
  title        = {Separators for sphere-packings and nearest neighbor graphs},
  journal      = {J. {ACM}},
  volume       = {44},
  number       = {1},
  pages        = {1--29},
  year         = {1997},
  doi          = {10.1145/256292.256294}
}

@article{KLEIN1988190,
title = {An efficient parallel algorithm for planarity},
journal = {Journal of Computer and System Sciences},
volume = {37},
number = {2},
pages = {190-246},
year = {1988},
issn = {0022-0000},
doi = {https://doi.org/10.1016/0022-0000(88)90006-2},
author = {Philip N. Klein and John H. Reif}
}

@inproceedings{GazitM87,
  author       = {Hillel Gazit and
                  Gary L. Miller},
  title        = {A Parallel Algorithm for Finding a Separator in Planar Graphs},
  booktitle    = {28th Annual Symposium on Foundations of Computer Science, Los Angeles,
                  California, USA, 27-29 October 1987},
  pages        = {238--248},
  publisher    = {{IEEE} Computer Society},
  year         = {1987},
  url          = {https://doi.org/10.1109/SFCS.1987.3},
  doi          = {10.1109/SFCS.1987.3}
}

@article{GILBERTHT84,
title = {A separator theorem for graphs of bounded genus},
journal = {Journal of Algorithms},
volume = {5},
number = {3},
pages = {391-407},
year = {1984},
issn = {0196-6774},
doi = {https://doi.org/10.1016/0196-6774(84)90019-1}
}

@misc{davies2025,
      title={Strongly sublinear separators and bounded asymptotic dimension for sphere intersection graphs}, 
      author={James Davies and Agelos Georgakopoulos and Meike Hatzel and Rose McCarty},
      year={2025},
      eprint={2504.00932},
      archivePrefix={arXiv},
      primaryClass={math.CO}
}

@article{Kostochka84,
  author       = {Alexandr V. Kostochka},
  title        = {Lower bound of the Hadwiger number of graphs by their average degree},
  journal      = {Comb.},
  volume       = {4},
  number       = {4},
  pages        = {307--316},
  year         = {1984},
  doi          = {10.1007/BF02579141}
}

@article{Thomason82,
  author       = {Andrew Thomason},
  title        = {Critically partitionable graphs {II}},
  journal      = {Discret. Math.},
  volume       = {41},
  number       = {1},
  pages        = {67--77},
  year         = {1982},
  doi          = {10.1016/0012-365X(82)90083-8}
}

@inproceedings{AlonSeymourThomas,
author = {Alon, N. and Seymour, P. and Thomas, R.},
title = {A separator theorem for graphs with an excluded minor and its applications},
year = {1990},
isbn = {0897913612},
publisher = {Association for Computing Machinery},
address = {New York, NY, USA},
doi = {10.1145/100216.100254},
booktitle = {Proceedings of the Twenty-Second Annual ACM Symposium on Theory of Computing},
pages = {293–299},
numpages = {7},
location = {Baltimore, Maryland, USA},
series = {STOC '90}
}

@inproceedings{PlotkinRaoSmith,
author = {Plotkin, Serge and Rao, Satish and Smith, Warren D.},
title = {Shallow excluded minors and improved graph decompositions},
year = {1994},
isbn = {0898713293},
publisher = {Society for Industrial and Applied Mathematics},
address = {USA},
booktitle = {Proceedings of the Fifth Annual ACM-SIAM Symposium on Discrete Algorithms},
pages = {462–470},
numpages = {9},
location = {Arlington, Virginia, USA},
series = {SODA '94}
}

@inproceedings{KawarabayashiR10,
  author       = {Ken{-}ichi Kawarabayashi and
                  Bruce A. Reed},
  title        = {A Separator Theorem in Minor-Closed Classes},
  booktitle    = {51th Annual {IEEE} Symposium on Foundations of Computer Science, {FOCS}
                  2010, October 23-26, 2010, Las Vegas, Nevada, {USA}},
  pages        = {153--162},
  publisher    = {{IEEE} Computer Society},
  year         = {2010},
  doi          = {10.1109/FOCS.2010.22}
}

@article{ReedW09,
  author       = {Bruce A. Reed and
                  David R. Wood},
  title        = {A linear-time algorithm to find a separator in a graph excluding a
                  minor},
  journal      = {{ACM} Trans. Algorithms},
  volume       = {5},
  number       = {4},
  pages        = {39:1--39:16},
  year         = {2009},
  url          = {https://doi.org/10.1145/1597036.1597043},
  doi          = {10.1145/1597036.1597043},
  bibsource    = {dblp computer science bibliography, https://dblp.org}
}

@INPROCEEDINGS{WulffNisen11,

  author={Wulff-Nilsen, Christian},

  booktitle={2011 IEEE 52nd Annual Symposium on Foundations of Computer Science}, 

  title={Separator Theorems for Minor-Free and Shallow Minor-Free Graphs with Applications}, 

  year={2011},
   pages={37-46},
  doi={10.1109/FOCS.2011.15}
}

@inproceedings{BernsteinGS21,
  author       = {Aaron Bernstein and
                  Maximilian Probst Gutenberg and
                  Thatchaphol Saranurak},
  title        = {Deterministic Decremental {SSSP} and Approximate Min-Cost Flow in
                  Almost-Linear Time},
  booktitle    = {62nd {IEEE} Annual Symposium on Foundations of Computer Science, {FOCS}
                  2021, Denver, CO, USA, February 7-10, 2022},
  pages        = {1000--1008},
  publisher    = {{IEEE}},
  year         = {2021},
  doi          = {10.1109/FOCS52979.2021.00100}
}

@inproceedings{ChuzhoyS21,
  author       = {Julia Chuzhoy and
                  Thatchaphol Saranurak},
  editor       = {D{\'{a}}niel Marx},
  title        = {Deterministic Algorithms for Decremental Shortest Paths via Layered
                  Core Decomposition},
  booktitle    = {Proceedings of the 2021 {ACM-SIAM} Symposium on Discrete Algorithms,
                  {SODA} 2021, Virtual Conference, January 10 - 13, 2021},
  pages        = {2478--2496},
  publisher    = {{SIAM}},
  year         = {2021},
  doi          = {10.1137/1.9781611976465.147}
}

@article{BODLAENDERGHK95,
title = {Approximating Treewidth, Pathwidth, Frontsize, and Shortest Elimination Tree},
journal = {Journal of Algorithms},
volume = {18},
number = {2},
pages = {238-255},
year = {1995},
issn = {0196-6774},
doi = {https://doi.org/10.1006/jagm.1995.1009},
author = {H.L. Bodlaender and J.R. Gilbert and H. Hafsteinsson and T. Kloks},
}

@article{LiptonRoseTarjan79,
author = {Lipton, Richard J. and Rose, Donald J. and Tarjan, Robert Endre},
title = {Generalized Nested Dissection},
journal = {SIAM Journal on Numerical Analysis},
volume = {16},
number = {2},
pages = {346-358},
year = {1979},
doi = {10.1137/0716027}
 }

@inproceedings{CaoFinemanRuseell20,
author = {Cao, Nairen and Fineman, Jeremy T. and Russell, Katina},
title = {Efficient construction of directed hopsets and parallel approximate shortest paths},
year = {2020},
isbn = {9781450369794},
publisher = {Association for Computing Machinery},
address = {New York, NY, USA},
doi = {10.1145/3357713.3384270},
booktitle = {Proceedings of the 52nd Annual ACM SIGACT Symposium on Theory of Computing},
pages = {336–349},
numpages = {14},
location = {Chicago, IL, USA},
series = {STOC 2020}
}

@article{YUSTER2010,
title = {Single source shortest paths in H-minor free graphs},
journal = {Theoretical Computer Science},
volume = {411},
number = {34},
pages = {3042-3047},
year = {2010},
issn = {0304-3975},
doi = {https://doi.org/10.1016/j.tcs.2010.04.028},
author = {Raphael Yuster},
}

@article{Dijkstra59,
title = {A note on two problems in connexion with graphs},
journal = {Numerische Mathematik},
volume = {1},
pages = {269,271},
year = {1959},
issn = {0945-3245},
doi = {10.1007/BF01386390},
author = {Dijkstra, E. W.},
}

@article{Munro71,
  author={J. Ian Munro},
  title={Efficient Determination of the Transitive Closure of a Directed Graph},
  year={1971},
  cdate={31536000000},
  journal={Inf. Process. Lett.},
  volume={1},
  number={2},
  pages={56-58}
}

@article{Purdom70,
author = {Purdom, Paul},
title = {A transitive closure algorithm},
year = {1970},
issue_date = {Mar 1970},
publisher = {BIT Computer Science and Numerical Mathematics},
address = {USA},
volume = {10},
number = {1},
issn = {0006-3835},
doi = {10.1007/BF01940892},
journal = {BIT},
month = mar,
pages = {76–94},
numpages = {19}
}

@article{ChanM2010,
author = {Chan, Timothy M.},
title = {More Algorithms for All-Pairs Shortest Paths in Weighted Graphs},
journal = {SIAM Journal on Computing},
volume = {39},
number = {5},
pages = {2075-2089},
year = {2010},
doi = {10.1137/08071990X}
}

@article{ALON199,
title = {On the Exponent of the All Pairs Shortest Path Problem},
journal = {Journal of Computer and System Sciences},
volume = {54},
number = {2},
pages = {255-262},
year = {1997},
issn = {0022-0000},
doi = {https://doi.org/10.1006/jcss.1997.1388},
author = {Noga Alon and Zvi Galil and Oded Margalit}
}

@inproceedings{Seidel92,
author = {Seidel, Raimund},
title = {On the all-pairs-shortest-path problem},
year = {1992},
isbn = {0897915119},
publisher = {Association for Computing Machinery},
address = {New York, NY, USA},
doi = {10.1145/129712.129784},
booktitle = {Proceedings of the Twenty-Fourth Annual ACM Symposium on Theory of Computing},
pages = {745–749},
numpages = {5},
location = {Victoria, British Columbia, Canada},
series = {STOC '92}
}

@article{RodittyS11,
  author       = {Liam Roditty and
                  Asaf Shapira},
  title        = {All-pairs shortest paths with a sublinear additive error},
  journal      = {{ACM} Trans. Algorithms},
  volume       = {7},
  number       = {4},
  pages        = {45:1--45:12},
  year         = {2011},
  doi          = {10.1145/2000807.2000813}
}

@inproceedings{Fineman24,
  author       = {Jeremy T. Fineman},
  editor       = {Bojan Mohar and
                  Igor Shinkar and
                  Ryan O'Donnell},
  title        = {Single-Source Shortest Paths with Negative Real Weights in \emph{{\~{O}}(mn\({}^{\mbox{8/9}}\))}
                  Time},
  booktitle    = {Proceedings of the 56th Annual {ACM} Symposium on Theory of Computing,
                  {STOC} 2024, Vancouver, BC, Canada, June 24-28, 2024},
  pages        = {3--14},
  publisher    = {{ACM}},
  year         = {2024},
  doi          = {10.1145/3618260.3649614}
}

@inproceedings{CaoF23,
  author       = {Nairen Cao and
                  Jeremy T. Fineman},
  editor       = {Nikhil Bansal and
                  Viswanath Nagarajan},
  title        = {Parallel Exact Shortest Paths in Almost Linear Work and Square Root
                  Depth},
  booktitle    = {Proceedings of the 2023 {ACM-SIAM} Symposium on Discrete Algorithms,
                  {SODA} 2023, Florence, Italy, January 22-25, 2023},
  pages        = {4354--4372},
  publisher    = {{SIAM}},
  year         = {2023},
  doi          = {10.1137/1.9781611977554.CH166}
}

@inproceedings{CaoFR22,
author = {Cao, Nairen and Fineman, Jeremy T. and Russell, Katina},
title = {Parallel Shortest Paths with Negative Edge Weights},
year = {2022},
isbn = {9781450391467},
publisher = {Association for Computing Machinery},
address = {New York, NY, USA},
doi = {10.1145/3490148.3538583},
pages = {177–190},
numpages = {14},
location = {Philadelphia, PA, USA},
series = {SPAA '22}
}

@inproceedings{CaoFR20,
  author       = {Nairen Cao and
                  Jeremy T. Fineman and
                  Katina Russell},
  editor       = {Christian Scheideler and
                  Michael Spear},
  title        = {Improved Work Span Tradeoff for Single Source Reachability and Approximate
                  Shortest Paths},
  booktitle    = {{SPAA} '20: 32nd {ACM} Symposium on Parallelism in Algorithms and
                  Architectures, Virtual Event, USA, July 15-17, 2020},
  pages        = {511--513},
  publisher    = {{ACM}},
  year         = {2020},
  doi          = {10.1145/3350755.3400222},
  bibsource    = {dblp computer science bibliography, https://dblp.org}
}

@article{Fineman20,
  author       = {Jeremy T. Fineman},
  title        = {Nearly Work-Efficient Parallel Algorithm for Digraph Reachability},
  journal      = {{SIAM} J. Comput.},
  volume       = {49},
  number       = {5},
  year         = {2020},
  doi          = {10.1137/18M1197850},
  bibsource    = {dblp computer science bibliography, https://dblp.org}
}

@inproceedings{RozhonHMZ23,
author = {Rozho\v{n}, V\'{a}clav and Haeupler, Bernhard and Martinsson, Anders and Grunau, Christoph and Zuzic, Goran},
title = {Parallel Breadth-First Search and Exact Shortest Paths and Stronger Notions for Approximate Distances},
year = {2023},
isbn = {9781450399135},
publisher = {Association for Computing Machinery},
address = {New York, NY, USA},
doi = {10.1145/3564246.3585235},
booktitle = {Proceedings of the 55th Annual ACM Symposium on Theory of Computing},
pages = {321–334},
numpages = {14},
location = {Orlando, FL, USA},
series = {STOC 2023}
}

@misc{debergC2025,
      title={An $O(n\log n)$ Algorithm for Single-Source Shortest Paths in Disk Graphs}, 
      author={Mark de Berg and Sergio Cabello},
      year={2025},
      eprint={2506.07571},
      archivePrefix={arXiv},
      primaryClass={cs.CG},
}

@article{HaitaoJ20,
author = {Wang, Haitao and Xue, Jie},
title = {Near-Optimal Algorithms for Shortest Paths in Weighted Unit-Disk Graphs},
year = {2020},
issue_date = {Dec 2020},
publisher = {Springer-Verlag},
address = {Berlin, Heidelberg},
volume = {64},
number = {4},
issn = {0179-5376},
doi = {10.1007/s00454-020-00219-7},
journal = {Discrete Comput. Geom.},
month = dec,
pages = {1141–1166},
numpages = {26}
}

@misc{brewer2024,
      title={An Improved Algorithm for Shortest Paths in Weighted Unit-Disk Graphs}, 
      author={Bruce W. Brewer and Haitao Wang},
      year={2024},
      eprint={2407.03176},
      archivePrefix={arXiv},
      primaryClass={cs.CG},
}

@article{ChanS19a,
  author       = {Timothy M. Chan and
                  Dimitrios Skrepetos},
  title        = {Approximate shortest paths and distance oracles in weighted unit-disk
                  graphs},
  journal      = {J. Comput. Geom.},
  volume       = {10},
  number       = {2},
  pages        = {3--20},
  year         = {2019},
  doi          = {10.20382/JOCG.V10I2A2},
  bibsource    = {dblp computer science bibliography, https://dblp.org}  
}

@misc{abboud2015,
      title={Approximation and Fixed Parameter Subquadratic Algorithms for Radius and Diameter}, 
      author={Amir Abboud and Virginia Vassilevska Williams and Joshua Wang},
      year={2015},
      eprint={1506.01799},
      archivePrefix={arXiv},
      primaryClass={cs.DS},
      url={https://arxiv.org/abs/1506.01799}, 
}

@article{ChaudZar00,
    author = {S. Chaudhuri and C. D. Zaroliagis},
    title = {Shortest Paths in Digraphs of Small Treewidth. {P}art {I}: {S}equential Algorithms},
    journal = {Algorithmica},
    volume = {27},
    pages = {212-226},
    year = 2000,
    url = {https://doi.org/10.1007/s004530010016}
}

@article{ChaudZar98,
title = {Shortest paths in digraphs of small treewidth. {P}art {II}: {O}ptimal parallel algorithms},
journal = {Theoretical Computer Science},
volume = {203},
number = {2},
pages = {205-223},
year = {1998},
issn = {0304-3975},
doi = {https://doi.org/10.1016/S0304-3975(98)00021-8},
url = {https://www.sciencedirect.com/science/article/pii/S0304397598000218},
author = {Shiva Chaudhuri and Christos D. Zaroliagis},
}

@misc{fomin2015,
      title={Fully polynomial-time parameterized computations for graphs and matrices of low treewidth}, 
      author={Fedor V. Fomin and Daniel Lokshtanov and Michał Pilipczuk and Saket Saurabh and Marcin Wrochna},
      year={2015},
      eprint={1511.01379},
      archivePrefix={arXiv},
      primaryClass={cs.DS},
      url={https://arxiv.org/abs/1511.01379}, 
}
